\documentclass[12pt]{article}
\usepackage[margin=1in]{geometry}

\usepackage{times,upgreek,morenotations,rotating}

\usepackage[utf8]{inputenc} 
\usepackage[T1]{fontenc}    
\usepackage{url}            
\usepackage{booktabs}       
\usepackage{amsfonts}       
\usepackage{nicefrac}       
\usepackage{microtype}      

\usepackage{morenotations} 

\usepackage{times}
\usepackage{graphicx} 
\usepackage{subfigure}
\usepackage{cleveref}
\usepackage{bm} 
\usepackage{fancyhdr}
\usepackage{graphicx}
\usepackage{color} 
\usepackage{booktabs}

\title{\papertitle} 

\author{
Richard
Nock\footnote{The Australian National University $\&$ the University
  of Sydney} $\:\:\:$ Stephen Hardy $\:\:\:$
Wilko Henecka $\:\:\:$ Hamish Ivey-Law\\
Giorgio Patrini\footnote{Now with the University of Amsterdam}  $\:\:\:$ Guillaume Smith $\:\:\:$ Brian Thorne\\\\
{\large N1 Analytics / Data61}\\\\
  \texttt{firstname.lastname@data61.csiro.au}\\
 \texttt{g.patrini@uva.nl}
}
\date{}

\begin{document}
\thispagestyle{empty}
\maketitle

\begin{abstract}

  \noindent Consider two data providers, each maintaining 
   records of different feature sets about common entities. They
   aim to learn a linear model over the whole set of features. This
   problem of \textit{federated learning} over \textit{vertically
   partitioned data} includes a crucial upstream issue:
   \textit{entity resolution}, \textit{i.e.} finding the correspondence between
   the rows of the datasets. It is well known that entity resolution,
   just like learning, is mistake-prone in the real world. Despite the 
importance of the problem, there has been no formal assessment of how errors in entity
   resolution impact learning. 

In this paper, we provide a thorough answer to this question,
answering how optimal classifiers, empirical losses, margins and
generalisation abilities are affected. While our answer spans a wide set of
losses --- going beyond proper, convex, or classification calibrated
---, it brings simple practical arguments to upgrade entity resolution
as a
preprocessing step to learning. One of these suggests that
entity resolution should be aimed at controlling or minimizing the
number of matching errors
between examples of distinct classes. In our experiments, we modify a simple
token-based entity resolution algorithm so that it indeed aims at avoiding
matching rows belonging to different classes, and perform experiments
in the setting where entity resolution relies on noisy data,
which is very relevant to real world domains. Notably, our approach
covers the case where one peer \textit{does not} have classes, or a
noisy record of classes. Experiments display that using the class
information during entity resolution can buy significant uplift for learning at little expense
from the complexity standpoint. 

\end{abstract}

\section{Introduction}\label{sec-int}


With the ever-expanding collection of data, it is becoming common
practice for organisations to cooperate with the objective of leveraging
their joint collections of data \cite{dhhiopstwPP,gascon2017privacy},
with a wider push to create and organise data marketplaces as
followers to the more monolithic data warehouse \cite{wTR}. Organisations are fully aware of the potential gain of combining their 
data assets, specifically in terms of increased statistical power for analytics and 
predictive tasks. For example, hospitals and medical facilities could leverage the medical history of common patients
in order to prevent chronic diseases and risks of future hospitalisation.

The problem of learning
models using the data collected and kept/maintained by different parties ---
\textit{federated learning} for short \cite{kmyrtsbFL} --- has become as much a
necessity as a concrete research challenge, expanding beyond machine
learning through fields
like databases and privacy. 
Among other features, work in the area can be classified in terms of
(a) whether the data is vertically or
horizontally partitioned and (b) the family of
models being learned. The overwhelming majority of previous work on secure distributed
learning considers a \textit{horizontal} data partition in which data
providers record the same features for different entities.  Solutions can take
advantage of the separability of loss functions which decompose the
loss by examples. Relevant approaches can be found \textit{e.g.} in
\cite{DBLP:journals/corr/XieWBB16} (and
references therein). 

In a vertical data partition, which is our setting, data providers can
record \textit{different} features for the same entities. The vertical
data partition case is more challenging than the horizontal one
\cite{gascon2017privacy}. To see this, notice that in the later
case, gathering all the data in one place makes any conventional
learning algorithm fit to learn from the whole data. In the vertical
partition case however, gathering the data in one place would not solve the
problem
since we would still have to figure out the correspondence between entities
of the different datasets to learn from the union of all features. Vertical data partition is more relevant to the
setting where different organisations would sit in the same market,
thus aggregating different features for the same customers. The
technical problem to overcome is that loss functions are in general not separable
over \textit{features}. With the exception of the unhinged loss
\cite{vmwLW}, this would be the case for most proper,
classification calibrated and/or non-convex losses
\cite{bjmCCj,nnBD,rwCB}. 
A way to overcome this problem is to
join the datasets \textit{upstream}, using a broad family of
techniques we refer to as entity resolution (or entity matching, record linkage,
\cite{christen12}). For the whole pipeline
--- from matching to learning ---
to be fully and properly optimized taking into account eventual
additional constraints (like privacy), it is paramount to tackle
and answer the following question: 
\begin{center}
"\textit{how does entity-resolution impact learning ?}",
\end{center}
in
particular because error-free entity resolution is often not
available in the
real-world \cite{hsRW}, see Figure
\ref{fig:entity-resolution}. Case studies report that exact matching can be very
damaging when identifiers are not stable and error-prone: 25$\%$ of
true matches would have been missed by exact matching in a census
operation \cite{sEP,wRL}. In fact, one might expect such errors to
just snowball with those of learning: for example, wrong matches of a hospital database with
pharmaceutical records with the objective to improve preventive
treatments could be disastrous on the predictive performances of a
model learned from the joined databases.

\begin{figure}[t]
\begin{center}
\includegraphics[trim=10bp 210bp 10bp
180bp,clip,width=.88\textwidth]{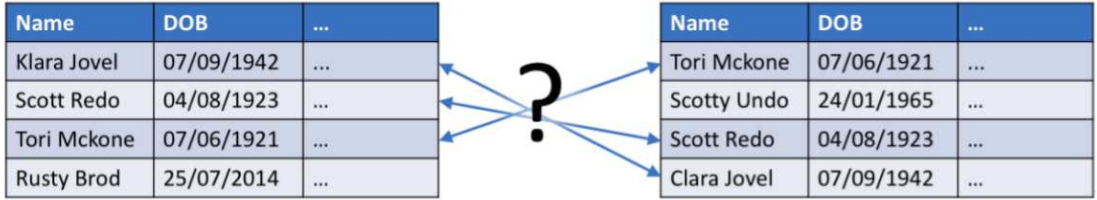} 
\caption{The problem of entity resolution. In this example, peers $\FDP$ and $\LDP$
  share common features (name, date of birth --- DOB) that could
  be used to craft an unique identifier, but the entries
  are noisy so it becomes hard to match rows between peers.}
\label{fig:entity-resolution}
\end{center}
\bignegspace
\end{figure}

To our knowledge, there has been no formal treatment of
this question so far, and the question is open not just for machine
learning as post-processing step to entity-resolution \cite{gmER}. As
a consequence perhaps, some work just assumes that the solution to
entity-resolution is known \textit{a priori} \cite{gascon2017privacy}.\\

\noindent \textbf{Our contribution} --- In this paper, we provide the first
detailed answer to this question and hint on how it can be used to
improve entity resolution as an upstream process to federated
learning with vertically partitioned data. We focus on a popular class of models for federated learning,
linear models \cite{gascon2017privacy,DBLP:journals/corr/XieWBB16}. To
summarize our theoretical contribution, we bound the variation of
several key quantities as computed from the error-prone entity-resolved
dataset on one hand, and also from the
\textit{ideal} dataset for which we would know the optimal
correspondence on the other hand. These
key quantities include:
\begin{itemize}
\item [(i)] the relative deviation between the optimal classifiers;
\item [(ii)] the deviation between their respective losses;
\item [(iii)] the deviation in their respective generalization
abilities;
\end{itemize}
More importantly, we carry this analysis for \textit{any}
Ridge-regularized loss that
satisfies some mild differentiability conditions, thus not necessarily
being convex, nor classification-calibrated, nor even proper. 

Overall, our results shed light on large margin classification in
the context of federated learning, and how it
brings resilience in learning after entity resolution. Indeed, we
show that it yields \textit{immunity} to entity resolution mistakes ---
examples receive the right class from the classifier learned from
error-prone entity-resolved data if they would receive large margin classification from
the optimal, "ideal" classifier learned from the ideal data. Federated
learning in the
vertical partition setting increases the number of features and
is thereby likely to increase margins as well. Hence, such a
theoretical result on immunity represents a very strong argument for federated learning.

On a broader agenda including impacts for practical entity-resolution
algorithms, our analysis suggests that there exists a small set of
controls defined from entity
resolution mistakes that essentially drive \textit{all} deviations highlighted
before. Being able to 
control them essentially leads to a strong handle on how
entity-resolution impacts learning, \textit{from the classifier
learned to its rates for generalization, with respect to the ideal
classifier}. The most prominent of these knobs is the
errors made by entity resolution \textit{across classes},
\textit{i.e.} wrongly linking observations that belong to different
classes. Our theory suggests that focusing on such mistakes during entity
resolution can bring significant leverage for the classifier learned afterwards. We exemplify this experimentally, by modifying a simple token-based
greedy entity-resolution algorithm to integrate the constraint of 
carrying out entity resolution within classes
\cite{cIO,gmER}, assuming that one peer has knowledge of the classes
but the other one may not --- either classes are noisy or just not
present ---. We perform simulated experiments on fifteen distinct UCI domains,
simulated to investigate the key parameters of federated learning in
the setting where peers share the knowledge of \textit{some} features
(such as gender, age, postal code for customers), which can
furthermore be \textit{noisy}. 
Experiments display that even when only one peer has the
knowledge of classes, significant improvements can be obtained over
the approach that performs entity resolution without using classes,
and can even compete with the result of the learner that has access to
the (unknown) ideally entity-resolved data. 

\noindent The rest of this paper is organised as follows. Section
\ref{sec:defs} gives definitions. Section \ref{sec:taylor} shows how
to reduce the analysis for a general loss to that of a specific kind
of loss called \textit{Taylor loss}. Sections \ref{sec:relerr} through
\ref{sec-rade}
develop our theoretical results, and Section \ref{sec:experiments}
provide experiments. A last Section discusses and concludes our
paper. An Appendix, starting page \pageref{sec-toc}, provides all proofs.

\section{Definitions}
\label{sec:defs}

\noindent \textbf{Supervised learning, losses ---} Let $[n] = \{1, 2, ..., n\}$. In the ordinary batch supervised
learning setting, one is given a set of $m$ examples
$\hat{S} \defeq \{(\hat{\ve{x}}_i, y_i), i \in [m]\}$, where $\hat{\ve{x}}_i
\in {\mathcal{X}} \subseteq {\mathbb{R}}^d$ is an observation
(${\mathcal{X}}$ is called the domain) and $y_i
\in \{-1,1\}$ is a label, or class (the "hat" notation shall be
explained below). Our objective is to learn a linear
classifier $\ve{\theta}
\in {\Theta}$ for some fixed ${\Theta} \subseteq {\mathbb{R}}^d$. $\ve{\theta}$ gives
a label to some $\ve{x} \in {\mathcal{X}}$ equal to the sign of
$\ve{\theta}^\top \ve{x} \in {\mathbb{R}}$. The goodness of fit of
$\ve{\theta}$ on $\hat{S}$ is measured by a loss function. We
essentially consider
two categories of losses. The first is the set of Ridge-regularized
losses. Each element, $\ell_F$, is
defined by $\loss_F({\hat{S}}, \ve{\theta};
\gamma, \Gamma) \defeq L + R$ with
\begin{eqnarray}
L\defeq \frac{1}{m}
\cdot \sum_i F(y_i \ve{\theta}^\top \hat{\ve{x}}_i) \:\:, \:\: R \defeq \gamma
\ve{\theta}^\top \Gamma \ve{\theta} \label{ridgeregLoss}\:\:.
\end{eqnarray} 
Here, $\gamma >0$ and $\Gamma$ is symmetric positive
definite. $F : \mathbb{R}\rightarrow \mathbb{R}$ is $C^2$ and satisfies $|F'(0)|, |F''(0)|\ll \infty$ where
"$\ll$" means finite. Note that this is a very general definition as
for example
we do \textit{not} assume that $F$ is convex nor even classification
calibrated \cite{bjmCCj}.

The other set of losses we consider, called \textit{Taylor
  losses}, is such that $L$ simplifies as a degree-two polynomial:
\begin{eqnarray}
L & \defeq & a +
\frac{b}{m} \cdot \sum_i y_i \ve{\theta}^\top \hat{\ve{x}}_i +
\frac{c}{m} \cdot \sum_i (y_i \ve{\theta}^\top \hat{\ve{x}}_i)^2 \:\:, \label{ridgeregTaylorLoss}
\end{eqnarray}
with $a,b, c \in
\mathbb{R}$. Taylor
losses have been used in secure federated learning
\cite{Aono:2016:SSL:2857705.2857731,dhhiopstwPP}. \\

\noindent \textbf{Federated learning ---} In federated learning,
$\hat{S}$ is built from separate data-handling sources, called
\textit{peers}. In our vertical partition setting, we have two peers
\FDP~and \LDP,
each of which has the description of the $m$ examples on a
\textit{subset} of the $d$ features. It may be the case that only one peer
(\FDP~by default) has labels. In addition to
learning a classifier, federated
learning thus faces the mandatory preprocessing step
of \textit{matching} rows in the datasets of \FDP~and \LDP~to build
dataset $\hat{S}$, a preprocessing step we define as \textit{entity
  resolution} \cite{christen12}. 

The observed dataset $\hat{S}$ is created from an \textit{unknown} dataset $S \defeq \{({\ve{x}}_i, y_i), i \in
[m]\}$ whose columns have been split between $\FDP$ and $\LDP$. If we define $\X \in \mathbb{R}^{d \times m}$ as the 
matrix storing (columnwise) observations of $S$, then each row of $\X$ is held by \FDP~or
\LDP. The "or" need not be exclusive as some rows may be present in
both \FDP~and \LDP~\cite{pnhcFL}. Also, duplicating rows in $\X$ does not change the
learning problem. There is thus both an ideal $\X$ and an estimated
observation matrix $\hat{\X}$ giving the observations of
$\hat{S}$ and built from entity-resolution. To understand how the
differences between $\hat{\X}$ and $\X$ impact learning, we need to
drill down into the formalization of $\hat{\X}$.
Both matrices can be represented by block
matrices, with each distinct feature row present exactly once, as:
\begin{eqnarray}
\X \defeq \left[
\begin{array}{c}
\X_\anchor\\\cline{1-1}
\X_\shuffle
\end{array}
\right] & , & \hat{\X} \defeq \left[
\begin{array}{c}
\X_\anchor\\\cline{1-1}
\hat{\X}_\shuffle \defeq \X_\shuffle \PERM_*
\end{array}
\right] \:\:,\label{blockm}
\end{eqnarray}
where $\PERM_* \in
\{0,1\}^{m\times m}$ is a
permutation matrix (unknown) capturing the mistakes of entity-resolution if
$\PERM_* \neq \matrice{i}_m$ (the identity matrix). From convention (\ref{blockm}), the
features of \FDP~are not affected by entity-resolution: we call them
\textit{anchor} features. Because the features of \LDP  are affected by
entity-resolution, we call them~\textit{shuffle}
features. A folklore fact \cite{bierens04} (Chapter
I.5) is that any permutation matrix can be factored as a product of \textit{elementary}
permutation matrices, each of which swaps two rows/columns of
$\matrice{i}_m$. 
So, suppose 
\begin{eqnarray}
\PERM_* & = & \prod_{t=1}^T \PERM_{t} \:\:,
\label{eqFACTPSTAR}
\end{eqnarray}
where $\PERM_{t}$ is an elementary permutation matrix, where $T$,
the \textit{size} of $\PERM_*$, is
unknown. We let $\ua{t}, \va{t}\in [m]$ the two column indexes in
\FDP~affected by $\PERM_t$. 
$\hat{\X}$ can be progressively constructed from a
sequence $\hat{\X}_0, \hat{\X}_1, ..., \hat{\X}_T$ where $\hat{\X}_0 =
\X$, $\hat{\X}_T = \hat{\X}$ and for $t\geq 1$,
\begin{eqnarray}
\hat{\X}_t & \defeq & \left[
\begin{array}{c}
\X_\anchor\\\cline{1-1}
\hat{\X}_{t \shuffle}
\end{array}
\right]\:\:, \hat{\X}_{t \shuffle} \defeq \X_\shuffle \prod_{j=1}^t \PERM_{j}\:\:.
\end{eqnarray}
Let $\hat{\X}_t \defeq [\hat{\ve{x}}_{t1} \:\: \hat{\ve{x}}_{t2} \:\: \cdots
\:\: \hat{\ve{x}}_{tn}]$ denote the column vector decomposition of $\hat{\X}_t$ (with
$\hat{\ve{x}}_{0i} \defeq \ve{x}_i$) and let $\hat{S}_t$ be the training sample obtained from the $t$
first permutations in the sequence. Hence, $\hat{S}_0 =
S$, $\hat{S}_T = \hat{S}$ and $\hat{S}_t \defeq \{(\hat{\ve{x}}_{ti}, y_i), i\in
[m]\}$. 
We let $\ub{t}$
(resp. $\vb{t}$) denote the indices in $[m]$ of the shuffle features in
$\X$ that are in observation $\ua{t}$ (resp. $\va{t}$) and that will
be permuted by $\PERM_t$, creating $\hat{\X}_{t}$ from $\hat{\X}_{t-1}$. For example, if $\ub{t} = \va{t}, \vb{t} = \ua{t}$, then
$\PERM_{t}$ correctly reconstructs observations in indexes $\ua{t}$ and
$\va{t}$ in $\X$. 
\begin{figure}[t]
\centering
\includegraphics[trim=25bp 280bp 330bp
10bp,clip,width=.90\linewidth]{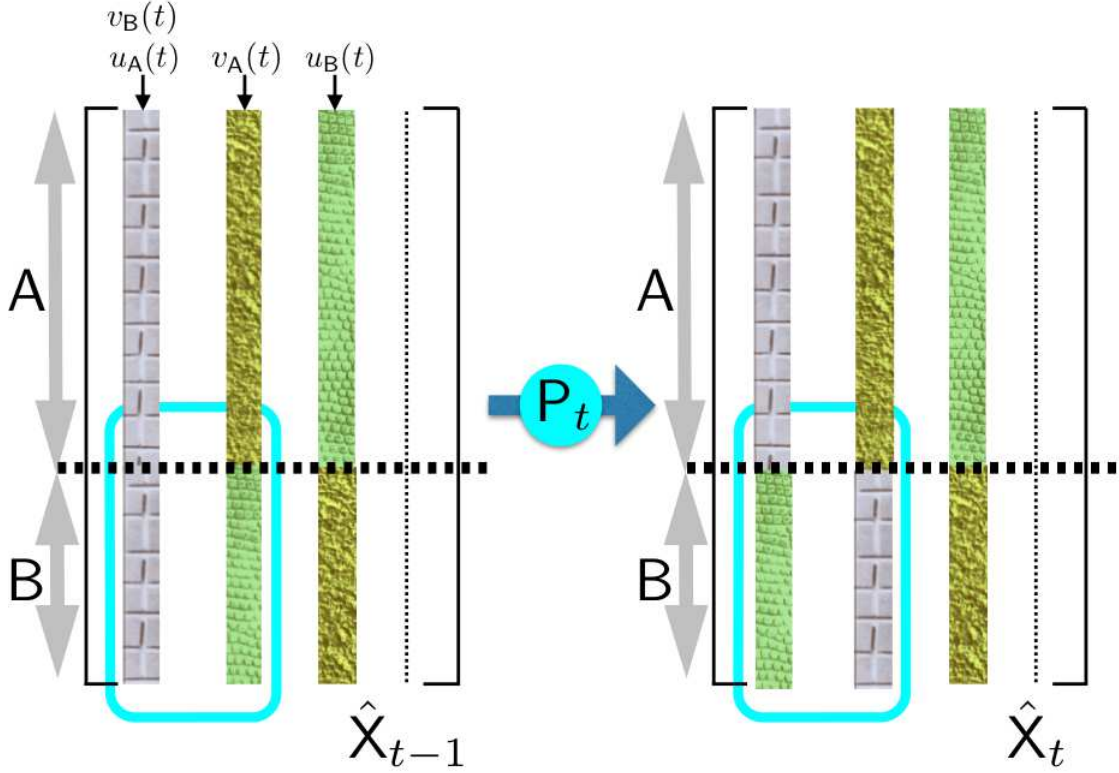} 
\caption{Permutation $\PERM_t$ applied to observation matrix
  $\hat{\X}_{t-1}$ and subsequent matrix $\hat{\X}_t$, using notations $\ua{t}$, $\ub{t}$, $\va{t}$ and
  $\vb{t}$. Textures represent observations of $\X$ (best viewed in color).\label{fig:perm}}
\end{figure}
 Figure \ref{fig:perm} illustrates the use of these
 notations. \\

\noindent \textbf{Key parameters of $\PERM_*$ ---} it is clear that
all mistakes of entity-resolution are captured by $\PERM_*$, so it is
not surprising that all our results depend on some key parameters
of $\PERM_*$. A key property is how errors "accumulate" through the
factorization of $\PERM_*$ in eq. (\ref{eqFACTPSTAR}). Hereafter,
$\ve{w}_{\F}$ for $\ve{w}\in \mathbb{R}^d$ denotes the subvector of $\ve{w}$ containing the
features of peer $\F \in \{\anchor, \shuffle\}$.
\begin{center}
\colorbox{gray!10}{\fbox{%
    \parbox{0.95\textwidth}{
\vspace{-0.4cm}
\begin{definition}\label{defACCURATE}
We say that $\PERM_t$ is $(\epsilon, \tau)$-\textbf{accurate}
for some $\epsilon, \tau\geq 0, \epsilon \leq 1$ iff for any $\ve{w}\in \mathbb{R}^d$,
\begin{eqnarray}
|(\hat{\ve{x}}_{ti} - \ve{x}_{i})_\shuffle^\top \ve{w}_\shuffle| &
\hspace{-0.95cm} \leq \hspace{-0.95cm}& \epsilon \cdot |\ve{x}_{i}^\top \ve{w}| + \tau \|\ve{w}\|_2 , \forall i \in [m]\:, \label{defACCURATE1}\\
|(\ve{x}_{\uf{t}} -
\ve{x}_{\vf{t}})_{\F}^\top\ve{w}_{\F}| & \hspace{-0.95cm}\leq \hspace{-0.95cm}& \epsilon \cdot \max_{i\in
  \{\uf{t}, \vf{t}\}} |\ve{x}_{i}^\top \ve{w}|\nonumber\\
  & & +
\tau \|\ve{w}\|_2\:\:, \forall \F \in \{\anchor, \shuffle\}\:\:.\label{defACCURATE2}
\end{eqnarray}
We say that $\PERM_*$ is $(\epsilon, \tau)$-accurate iff each $\PERM_t$ is
$(\epsilon, \tau)$-accurate, $\forall t = 1, 2, ..., T$. 
\end{definition}
\vspace{-0.4cm}
}}
}
\end{center}
If we consider that vectors $\hat{\ve{x}}_{ti} - \ve{x}_{i}, \ve{x}_{\uf{t}} -
\ve{x}_{\vf{t}}$ quantify errors attributable to 
$\PERM_t$, then $(\epsilon, \tau)$-accuracy postulates that errors along
any direction are bounded by a fraction of the norm of original observations,
plus a penalty. In the
context of the inequalities,
$\tau$ is homogeneous to a norm while 
$\epsilon$ is ``unit-free''. For that reason, we define an important quantity aggregating $\epsilon$ and a ``unit-free'' $\tau$:
\begin{eqnarray}
\xi & \defeq & \epsilon + \frac{\tau}{X_*}\:\:,\label{defXI}
\end{eqnarray}
where $X_* \defeq
  \max_i \|\ve{x}_i\|_2$ is the max norm in (the columns of)
  $\X$. Remark that we always have
\begin{eqnarray}
\xi & \leq & 3 \:\:, \forall \PERM_*\:\:. \label{eqXI}
\end{eqnarray}
Indeed, it is always true that $\hat{\ve{x}}_{ti} \leq
  2X_*$ and so $\|\hat{\ve{x}}_{ti} - \ve{x}_{i}\|_2 \leq 3X_*,
  \forall i \in [m]$, so regardless of $\PERM_*$, we can always choose $\epsilon =
  0, \tau = 3X_*$, making $\xi$ satisfy ineq. (\ref{eqXI}). Much
  smaller values are possible: for
  example, when entity-resolution is so good that errors eventually slightly
  change norms but not directions (\textit{e.g.} $\hat{\ve{x}}_{ti} - \ve{x}_{i} =
  \epsilon_i \cdot \ve{x}_i, \forall i$), then we may end up with
  $\epsilon$ close to zero and $\tau = 0$, resulting in $\xi$ close
  to zero as well. The reason why it is desirable for $\xi$ to be significantly
  smaller is given in the following definition.
\begin{center}
\colorbox{gray!10}{\fbox{%
    \parbox{0.95\textwidth}{
\vspace{-0.4cm}
\begin{definition}\label{defSIZE}
We say that $\PERM_*$ is \textbf{$\alpha$-bounded} 
for some $0< \alpha \leq 1$
iff its size satisfies 
\begin{eqnarray*}
T & \leq & \left(\frac{m}{\xi}\right)^{\frac{1-\alpha}{2}}\:\:.
\end{eqnarray*}
\end{definition}
\vspace{-0.4cm}
}}
}
\end{center}
It is crucial to remark that this puts a constraint on the size $T$
since in all cases we shall require $T = O(\sqrt{m})$ whenever $\xi$
is not small (say $\xi = 3$, ineq. \ref{eqXI}). This
constraint is considerably weakened when the magnitude of errors
($\xi$) gets small, so that we can end up with $\PERM_*$ 
$\alpha$-bounded for $\alpha$ very
close to 1, which shall be a highly desirable feature for the theory
to follow. Notice also that a permutation can always be decomposed
in elementary permutations with $T\leq m$, yet to achieve a
particular level of $(\epsilon,
\tau)$-accuracy, we may need more than the minimal size
factorisation. It seems however more than reasonable to assume that we shall
still have $T = O(m)$ in all cases, which does not fundamentally change the picture of the
constraint imposed by $\alpha$-boundedness. 
Finally, we let $T_+ \leq T$ denote the number of \textit{class mismatch}
permutations in the factorization, \textit{i.e.} for which $y_{\ua{t}}
\neq y_{\va{t}}$ and let
\begin{eqnarray}
\rho & \defeq & \frac{T_+}{T} \in [0,1]\:\:\label{defrho}
\end{eqnarray}
define the proportion of elementary permutations that act between
classes. \\

\noindent \textbf{Key parameters for our results ---} 
Remarkably, all our results on how mistakes of entity resolution
affect learning essentially depend on \textit{three} parameters
only, each characterizing a distinct unknown: the ideal
classifier $\ve{\theta}_0^*$ ($\deltatheta$), permutation $\PERM_*$
($\deltaperm$) and the ideal dataset $S$ ($\deltaset$):\begin{center}
\colorbox{gray!10}{\fbox{%
    \parbox{0.95\textwidth}{
\begin{eqnarray}
\deltatheta \defeq \|\ve{\theta}_0^* \|_2 X_* \:\:, \:\: \deltaperm
\defeq \frac{\sqrt{\xi} \rho}{4}\:\:, \:\: \deltaset \defeq
\left\| \frac{1}{mX_*}\cdot \sum_i
y_i\ve{x}_i \right\|_2 \:\:.\label{defunkwn}
\end{eqnarray}
}}
}
\end{center}
It is not hard to see that $\deltatheta$ is an upperbound on the
margin achieved by the ideal classifier on the true dataset (Section
\ref{sec-marg}), $\deltaperm$ aggregates class mismatch and the
magnitude of errors in permutations, and $\deltaset$ is the norm of a normalized
version of a sufficient statistics for the class in $S$ called the
mean operator \cite{pnrcAN}. These can globally be seen as penalties
--- the smaller they are, the less impact has $\PERM_*$ on
learning. The most important with respect to the design of entity
resolution algorithms for federated learning, $\deltaperm$, displays an interesting regime: 
when $\PERM_*$ is "good enough" that $\rho = 0$ --- that is,
we make no entity-resolution mistakes \textit{between} classes ---,
we have $\deltaperm = 0$, which can bring substantially better bounds on
all our results.
\section{Anayzing Losses via Taylor losses}
\label{sec:taylor}

We let $\taylorloss$
denote a Taylor loss. The importance of Taylor losses in our
context is provided by our first Theorem below which, by means of
words, says that any of our losses $\loss_F$ sufficiently regularized can be locally
approximated in a neighborhood of any local minimum by a particular
convex Taylor loss with very specific parameterization, crucial for
our next results. Following \cite{bMA}, we let $\lambda_1^\uparrow(.)$
denote the smallest eigenvalue. We let
$\mathcal{C}$ denote the set of local minima of $\loss_F({\hat{S}}, \ve{\theta};
\gamma, \Gamma)$ --- omitting dependences in ${\hat{S}}, \gamma, \Gamma$ --- and for any $\ve{\theta}^*\in \mathcal{C}$, we let
$\mathcal{N}(\ve{\theta}^*)$ denote an open neighborhood
of $\ve{\theta}^*$
over which $\loss_F$ is convex, which is guaranteed to be non empty by the
assumptions on $F$. Parameters $a, b, c$ below refer to
those in eq. (\ref{ridgeregTaylorLoss}).
\begin{theorem}\label{lemmaTaylor}
$\forall \lambda^\circ > 0$ and sample $\hat{S}$, there exists
$\lambda^*>0$ such that for any loss $\loss_F({\hat{S}}, \ve{\theta};
\gamma, \Gamma_F)$ satisfying $\gamma \lambda_1^\uparrow(\Gamma_F)
\geq \lambda^*$ and any $\ve{\theta}^* \in \mathcal{C}$, there exists a
convex Taylor loss $\opttaylorloss({\hat{S}}, \ve{\theta};
\gamma, \Gamma_T)$ such that\\
\noindent (i) $a = F(0), b = F'(0)$,\\
\noindent (ii) $ \arg\min_{\ve{\theta}} \opttaylorloss({\hat{S}}, \ve{\theta};\gamma,
\Gamma_T) = \ve{\theta}^*$, and\\
\noindent (iii) $\gamma
\lambda_1^\uparrow(\Gamma_T) \geq \lambda^\circ$.\\
Furthermore, if $F$ is strictly
convex, then $c>0$.
\end{theorem}
(Proof in Appendix, Subsection
\ref{proof_lemmaTaylor}) Even when not as crucial as for (i -- iii),
the proof of Theorem \ref{lemmaTaylor} shows that we also have $\loss_F({\hat{S}}, \ve{\theta}^*;
\gamma, \Gamma_F) = \opttaylorloss({\hat{S}}, \ve{\theta}^*;
\gamma, \Gamma_T)$, \textit{i.e.} both losses coincide at the local
optimum for $\ell_F$. 

A natural question is what is the strength of the regularization imposed on
$\ell_F$ ($\lambda^*$). While one can figure out worst cases $F$ --- not defining usual
losses --- for which $\lambda^*$ is large, we show that a
popular subset of proper losses yield reasonable values for
$\lambda^*$ \cite{nnBD}: such losses are strictly convex, non-negative and have no class-dependent
misclassification cost. It can be shown for any such loss that there exists
a \textit{permissible} $\psi$ such that
 $F \defeq
F_\psi$ with
\begin{eqnarray}
F_\psi (z) & \defeq & \frac{\psi(0) + \psi^\star(-z)}{\psi(0) - \psi(1/2)}
\defeq a_\psi + \frac{\psi^\star(-z)}{b_\psi} \:\:, \label{deffpsi}
\end{eqnarray}
where $\star$ is the convex conjugate \cite{nnBD}. A permissible $\psi$ satisfies $\mathrm{dom}(\psi)
\supseteq [0,1]$, $\psi$ strictly convex, differentiable and symmetric
with respect to $1/2$. We add the condition that $\psi'$ is
concave on $[0,1/2]$ and denote this set of losses as 
\textit{regular symmetric proper losses} (\rspl). Popular examples of \rspl s include the square,
logistic and Matsushita losses \cite{nnBD}, the square loss also
being a Taylor loss. We let $\hat{X}_* \defeq \max_i \|\hat{\ve{x}}_i\|_2$.
\begin{lemma}\label{lemmaTaylor2}
Whenever $\ell_F$ is a \rspl~in Theorem \ref{lemmaTaylor}, we can pick
$\lambda^* \defeq \lambda^\circ + F_{\psi}''(0) \hat{X}^2_* / 2$. 
\end{lemma}
(Proof in Appendix, Subsection
\ref{proof_lemmaTaylor2}) As examples, $F_{\psi}''(0)$ is respectively
$1/4$, $1/2$ for the logistic and Matsushita losses \cite{nnBD}, which results in
a relatively small value for $\lambda^*$.

We can briefly summarize this Section as follows: when sufficiently
regularized, essentially \textit{any} locally optimal classifier for \textit{any} loss
$\ell_F$ is also optimal for \textit{some} specific regularized convex Taylor loss
in which we have $a = F(0), b = F'(0)$. So, we focus in
what follows on the sequence of optimal classifiers for such
Taylor losses, in which the \textit{sequence} is defined by a
progressive application to $S$ of the unknown elementary permutations
defining $\PERM_*$ as in eq. (\ref{eqFACTPSTAR}):
\begin{eqnarray}
\ve{\theta}^*_t & \defeq & \arg\min_{\ve{\theta}}
\taylorlossparams{a}{b}{c} ({\hat{S}}_t, \ve{\theta}; \gamma,
\Gamma)\:\:,\label{optTaylor}
\end{eqnarray}
for $a = F(0), b = F'(0)$ and $c\in \mathbb{R}_*$, where $a, b, c$
refer to coefficients in
eq. (\ref{ridgeregTaylorLoss}). The particular case $c=0$ yields similar bounds with weaker assumptions but
essentially corresponds to a single loss, the unhinged loss
\cite{vmwLW}. We shall assume without loss of generality that the null
vector is not optimal for the Taylor loss, which is in fact guaranteed
by the regularization --- alternatively, it would also hold when $F'(0) \neq 0$, which is
for example ensured by classification calibrated losses
\cite{bjmCCj}. 
\begin{figure}[t]
\centering
\includegraphics[trim=80bp 240bp 400bp
50bp,clip,width=.60\linewidth]{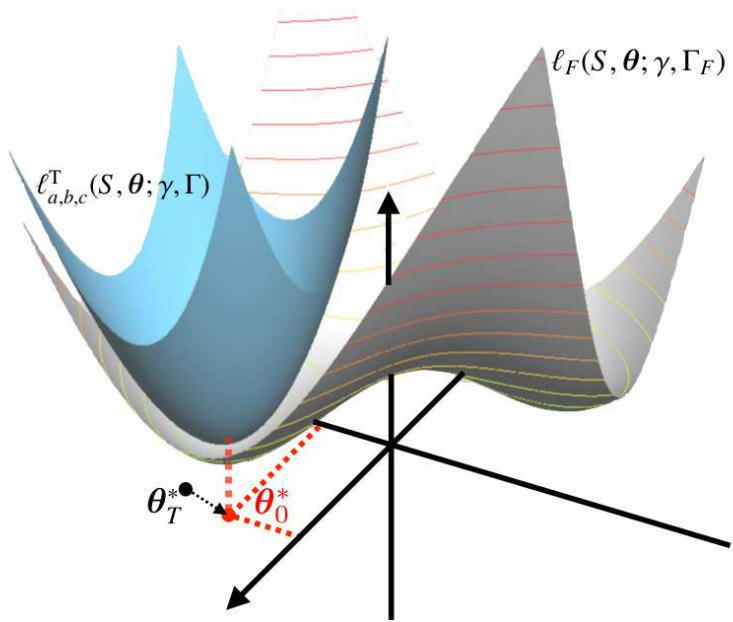} 
\caption{We study the variations of the Taylor loss $\taylorlossparams{a}{b}{c}$ (blue) which approximates
  a loss $\loss_F$ around a local minimum (grey), as a function of the mistakes
  done by entity resolution: all our results essentially show
  convergence of the optimal classifier learned from such mistakes,
  $\ve{\theta}^*_T$, to the optimal classifier learned from the ideal
  dataset, $\ve{\theta}^*_0$. This convergence holds for both the
  Taylor loss and the original loss as well (see text).\label{fig:approxloss}}
\end{figure}

The interest in focusing on Ridge-regularized Taylor losses is
three-folds: (i)  we
are trivially able to compute the optimum of such losses given the
constraints imposed on Theorem \ref{lemmaTaylor}; we thus get access
to a
fine-grained analysis of $\ve{\theta}^*_T$; (ii) our fine-grained bounds are
\textit{directly} relevant to the numeroux approaches to secure encrypted
federated learning that choose to directly optimize a Taylor loss
(often, the square loss Ridge-regularized)
\cite{eahEA,gascon2017privacy,gjpyPP,nwijbtPP}, or an approximation of
a loss function via a Taylor loss
\cite{Aono:2016:SSL:2857705.2857731,dhhiopstwPP}; (iii) provided
$\loss_F$ is continuous, which is a very weak assumption, our bounds
that analyze the convergence of classifiers or losses also
apply asymptotically to \textit{any} loss $\loss_F$ defined in Section
\ref{sec:defs}. This is summarized in Figure \ref{fig:approxloss}.
Our next objective is to compute the deviation between
$\ve{\theta}^*_0$ and $\ve{\theta}^*_T$ --- respectively
the ideal classifier (optimal on the ideal, perfectly
entity-resolved dataset $S$) and the classifier we learn after our
mistake-prone entity-resolved dataset $\hat{S}$ --- and then see how this impacts learning from
a variety of different standpoints. 

\section{Bounds on the relative deviation between optimal classifiers}
\label{sec:relerr}

To obtain our results, we shall need an assumption regarding
the data and learning problem parameters. We denote $\sigma(\mathcal{S})$
as the standard deviation of a discrete set $\mathcal{S} \subset
\mathbb{R}$ and define the \textbf{stretch} of vector $\ve{x}$ along
direction $\ve{w}\neq \ve{0}$ as: $\vstretch(\ve{x},\ve{w}) \defeq \|\ve{x}\|_2
|\cos(\ve{x}, \ve{w})|$. Let us denote for short
$v_{\mathrm{s}}(\ve{w})$ (resp. $\mu_{\mathrm{s}}(\ve{w})$) the
variance of stretches (resp. average of squared stretches) in $S$:
\begin{eqnarray}
\mu_{\mathrm{s}}(\ve{w}) & \defeq & \expect [\vstretch(\ve{x},\ve{w})^2]\:\:,\nonumber\\
v_{\mathrm{s}}(\ve{w}) & \defeq & \mu_{\mathrm{s}}(\ve{w}) 
- (\expect [\vstretch(\ve{x},\ve{w})])^2\:\:.
\end{eqnarray}
Notice that both $\mu_{\mathrm{s}}$ and $v_{\mathrm{s}}$ are invariant
to changes in the norm of $\ve{w}$. For 
$i\in \{-1, +1\}$, let
\begin{eqnarray}
U(i) \defeq \inf_{\ve{w}\neq \ve{0}} c\cdot \left\{
\begin{array}{rcl}
 (1-\epsilon)^2 v_{\mathrm{s}}(\ve{w}) & \mbox{if} & i = +1\\
(1+\epsilon)^2 \mu_{\mathrm{s}}(\ve{w})  +
  \tau^2 & \mbox{if} & i = -1
\end{array}
\right..
\end{eqnarray}
\begin{definition}\label{def:DMC}
We say that the \textbf{data-model calibration} assumption holds iff
the following two constraints are satisfied: (a)
(\textit{Maxnorm-variance regularization}) Ridge regularization parameters $\gamma, \Gamma$
  are chosen so that
\begin{eqnarray*}
\frac{ X_*^2}{\frac{1}{2} U(\mathrm{sign}(c))  + \gamma
  \lambda_1^\uparrow(\Gamma)} & \hspace{-0.2cm} \leq \hspace{-0.2cm}&
\frac{1}{2}\cdot \min\left\{\frac{1}{|F'(0)|}, \frac{1}{2|c|}\right\}\:,
\end{eqnarray*}
(b) (\textit{Minimal data size}): $m \geq 4 \xi$.
\end{definition}
Condition (a) imposes the Taylor loss to be sufficiently regularized
and explains why we state Theorem \ref{lemmaTaylor} with its condition
(iii). We remark that constraint (a) is all the less demanding as
$\epsilon,\tau$ are small, and $|U|$ is always $O(X_*^2)$, so the
constraint on regularization is roughly $\gamma
  \lambda_1^\uparrow(\Gamma) \geq u X_*^2$, for a constant $u\geq 0$ that
  can be very small when $F$ is convex ($c>0$, Theorem
  \ref{lemmaTaylor}). Condition (b) just postulates that $m$ is larger
  than a small constant, \textit{e.g.} $m\geq 12$ if we consider
  ineq. (\ref{eqXI}).

We now state our first result on how
$\ve{\theta}^*_T$ deviates from $\ve{\theta}^*_0$. 
\begin{theorem}\label{thAPPROX1}
Suppose $\PERM_*$ is $(\epsilon, \tau)$-accurate and the data-model
calibration assumption holds. Then we have:
\begin{eqnarray}
\frac{\|\ve{\theta}_T^* - \ve{\theta}_0^* \|_2}{\|\ve{\theta}_0^* \|_2} & \leq & \frac{\xi}{n} \cdot T^2 \cdot \left( 1+
\frac{\deltaperm}{\deltatheta}\right) \:\:.\label{eqthAPPROX00}
\end{eqnarray}
If, furthermore, $\PERM_*$ is $\alpha$-bounded, then 
$\|\ve{\theta}_T^* - \ve{\theta}_0^* \|_2 / \|\ve{\theta}_0^* \|_2\leq
C(m) \cdot (1+ (\deltaperm/\deltatheta))$, with 
\begin{eqnarray}
C(m) & \defeq & \left(\frac{\xi}{m}\right)^{\alpha}\:\:.\label{defCM}
\end{eqnarray}
\end{theorem}
(proof in Appendix, Section \ref{proof_thAPPROX1})
Remember that the most stringent assumption
is that of $\alpha$-boundedness, so essentially as long as we have
access to more data that can be linked by keeping entity-resolution errors bounded enough
\textit{in size} (say, $T = o(\sqrt{m})$), the impact of
entity-resolution on the drift between optimal classifiers \textit{vanishes}
with $m$:
\begin{eqnarray}
C(m) & \rightarrow_{+\infty} & 0 \:\:, \forall \alpha \in (0, 1]\:\:.
\end{eqnarray}
We also remark that when $\deltaperm = 0$, which happens when there
are no entity-resolution mistakes between classes due to $\PERM_*$, then under the
assumptions of Theorem \ref{thAPPROX1}, we simply have
\begin{eqnarray}
\frac{\|\ve{\theta}_T^* - \ve{\theta}_0^* \|_2}{\|\ve{\theta}_0^*
  \|_2} & \leq & C(m) \:\:.\nonumber
\end{eqnarray}
We now quantify how the bounded drifts guaranteed by Theorem
\ref{thAPPROX1} translate to learning.
\section{Optimal large margin classification implies
  immunity to entity resolution mistakes} \label{sec-marg}
We 
show that under the conditions of Theorem \ref{thAPPROX1}, large
margin classification by the ideal classifier ($\ve{\theta}_0^* $)
survives the changes brought by $\PERM_*$ on learning, in the sense that
the related examples will also be given
the \textit{same, right} class by the classifier we learn, $\ve{\theta}_T^*$  --- the corresponding
margin, however, may vary. We formalize the definition now.
\begin{definition}\label{defIMMUNEMARGIN}
Fix $\kappa > 0$. We say that $\ve{\theta}_T^*$ is immune to entity resolution at margin $\kappa$ iff for any example $(\ve{x},
y)$, if $y (\ve{\theta}_0^*)^\top \ve{x} > \kappa$, then $y (\ve{\theta}_T^*)^\top \ve{x} > 0$.
\end{definition}
Hence, $(\ve{x},
y)$ receives the right class by both $\ve{\theta}_0^*$ and $\ve{\theta}_T^*$.
We can now formalize the immunity property.
\begin{theorem}\label{thIMMUNE}
Suppose $\PERM_*$ is $(\epsilon, \tau)$-accurate and
$\alpha$-bounded, and the data-model calibration assumption
holds. For any $\kappa > 0$, $\ve{\theta}_T^*$ is immune to
entity resolution at margin $\kappa$ if
\begin{eqnarray}
 m & > &  \xi \cdot \left(
  \frac{\deltatheta + \deltaperm}{\kappa} \right)^{\frac{1}{\alpha}}\:\:.\label{eqNXIKAPPA}
\end{eqnarray}
\end{theorem}
(proof in Appendix, Section \ref{app:proof-thIMMUNE})
Eq. (\ref{eqNXIKAPPA}) is interesting for the relationships between
$m$ (data), $\xi$ (permutation) and $\kappa$ (margin) to achieve
immunity. Consider a permutation $\PERM_*$ for which $\rho = 0$. Since the maximal optimal margin is bounded by $\deltatheta$ by Cauchy-Schwartz
inequality, Theorem \ref{thIMMUNE} says that picking $\kappa \defeq \delta \cdot \deltatheta$ for
$0<\delta<1$ brings immunity at margin $\kappa$ if  $\delta >  C(m)$
where $C(m)$ is defined in Theorem \ref{thAPPROX1}, so the lowest
possible margin from which immunity holds converges to zero at rate
$1/m^\alpha$. When it is not the case that $\rho = 0$ however, the
picture can be very different if $\deltatheta$ is very small compared
to $\deltaperm$.
\section{Taylor losses of optimal classifiers on
  the ideal data}\label{sec:diffloss}

In this Section, we essentially show that under the
assumptions of Theorem \ref{thIMMUNE}, it holds that (little-oh
wrt $m \rightarrow \infty$):
\begin{eqnarray}
\taylorlossparams{a}{b}{c} (S, \ve{\theta}^*_T; \gamma,
\Gamma) - \taylorlossparams{a}{b}{c} (S, \ve{\theta}^*_0; \gamma,
\Gamma) & = & o(1)\:\:,
\end{eqnarray} 
\textit{i.e.} our classifier converges to the same loss \textit{on the
ideal data} $S$ as the ideal classifier,
and the convergence is governed by
  $C(m)$, therefore displaying a rate
  proportional to $1/m^\alpha$.
\begin{theorem}\label{thDIFFLOSS}
Denote for short $\taylorloss(\ve{\theta}) \defeq \taylorlossparams{a}{b}{c} (S, \ve{\theta}; \gamma,
\Gamma)$, with $a, b, c$ as in eq. (\ref{optTaylor}). If $\PERM_*$ is $(\epsilon, \tau)$-accurate and
$\alpha$-bounded, and the data-model calibration assumption
holds, then:
\begin{eqnarray}
\taylorloss(\ve{\theta}^*_T) - \taylorloss(\ve{\theta}^*_0) & \leq &
C(m) \cdot (\deltatheta + \deltaperm) 
\cdot A\:\:,\label{eqDIFFLOSS1}
\end{eqnarray}
where $A \defeq |F'(0)|\deltaset + \left( 3\deltatheta +
  2\deltaperm\right)(|c|+ d
      \gamma\lambda^\downarrow_1(\Gamma) / X_*^2)$.
\end{theorem}
(proof in Appendix, Section \ref{app:proof-thDIFFLOSS}) We remark a
difference with Theorem \ref{thIMMUNE}: the bound also depends on
$\deltaset$, and so on the norm of a sufficient statistics for the
class, the mean operator \cite{pnrcAN}.

\section{Generalization abilities}\label{sec-rade}
\newcommand{\genB}{Q}

Suppose that ideal sample $S$ is obtained
i.i.d. from some unknown distribution $\mathcal{D}$, before it is
"split" between \FDP~and \LDP, and then reconstructed to form our
training sample $\hat{S}$. What is the
generalization ability of classifier $\ve{\theta}_T^*$, learned on $\hat{S}$ ? This question
is non-trivial because it entails the impact of entity
resolution on generalization, and not just on training, that is, we
want to upperbound $\Pr_{(\ve{x}, y)\sim \mathcal{D}} [y (\ve{\theta}_T^*)^\top \ve{x}
\leq 0 ]$ with high probability given that the data we have
access to may not exactly reflect sampling from $\mathcal{D}$, or even
existing examples. In
essence, we show that provided the size of $\PERM_*$ is further
bounded (say, $T = o(m^{1/4})$), the guarantees on the rate of convergence for
generalization of $\ve{\theta}^*_T$ are
of the \textit{same order} as the one for $\ve{\theta}^*_0$. To get
this result, we first note that Ridge regularization implies that the norm
of $\ve{\theta}^*_0$ is bounded, say as $\|\ve{\theta}^*_0\|_2 \leq
\theta_*$ for some $\theta_*$. Let us then define $R^*_m \defeq
X_*\theta_* / \sqrt{m}$, 
which is an upperbound for the empirical Rademacher complexity of
$\ve{\theta}^*_0$ \cite{kstOT} (Theorem 3). It comes that with probability
$\geq 1-\delta$, we shall have $\Pr_{(\ve{x}, y)\sim \mathcal{D}} \left[y (\ve{\theta}^*_0)^\top \ve{x}
\leq 0 \right] \leq \genB$ with 
\begin{eqnarray}
\genB \defeq \taylorlossparams{F(0)}{F'(0)}{c} (S, \ve{\theta}^*_0; \gamma,
\Gamma) + 2LR^*_m + \sqrt{\frac{\ln(2/\delta)}{2m}}\label{defBM1}
\end{eqnarray}
\cite{bmRA} (Theorem 7), where $L$ is the Lipschitz constant of the
Ridge-regularized Taylor loss. The question we answer now is how we
can bound $\Pr_{(\ve{x}, y)\sim \mathcal{D}} [y (\ve{\theta}_T^*)^\top \ve{x}
\leq 0 ]$ as a function of $\genB$, which, we recall, quantifies
the generalization abilities of $\ve{\theta}^*_0$. 
\begin{theorem}\label{thGENTHETA}
With probability at least $1-\delta$ over the sampling of $S$
according to $\mathcal{D}^m$, as long as permutation $\PERM_*$ that creates $\hat{S}$ from $S$ is $(\epsilon, \tau)$-accurate and
$\alpha$-bounded and the data-model calibration assumption
holds, it holds that
\begin{eqnarray}
\lefteqn{\Pr_{(\ve{x}, y)\sim \mathcal{D}} \left[y (\ve{\theta}^*_T)^\top \ve{x}
\leq 0 \right]}\nonumber\\
 & \leq & \genB + C(m) \cdot (\deltatheta + \deltaperm) 
\cdot\left( A + \frac{2L}{\sqrt{m}}\right)\:\:,\label{eqGENER1}
\end{eqnarray}
where $A$ is defined in Theorem \ref{thDIFFLOSS}.
\end{theorem}
(proof in Appendix, Section \ref{app:proof-thGENTHETA}) Hence, with
high probability, entity resolution
affects generalization \textit{only} through the additional penalty to
$\genB$ in ineq. (\ref{eqGENER1}), which is factored by $C(m)$. In
consequence, if $\PERM_*$ is "small" enough so that $\alpha \geq 1/2$,
then we keep the rate of order $O(1/\sqrt{m})$ of the entity-resolution-free
case of ineq. (\ref{defBM1}). 

\renewcommand\thesubfigure{(\alph{subfigure})}

\section{Experiments}
\label{sec:experiments}

\subsection{Setting}\label{sub:sett}

We consider
the setting in which peers $\FDP$ and $\LDP$ have a small set of their
features which is present in both peers, features that we
call \textit{shared features} and that are used for entity resolution. This setting is realistic
considering \textit{e.g.} that many businesses or government bodies would share basic
information about their customers (such as gender, postal code, age,
contact number, etc.) \cite{pnhcFL}. We then put noise in those shared features as a
slider to
vary the hardness of the task. Notice the challenging aspect of the
task: entity resolution is computed from a relatively small set of eventually
noisy features, after which learning is carried out. This typically
corresponds to the example of Figure \ref{fig:entity-resolution}. We
adopt a simple noise injection process, inspired by thorough analyses in the
area \cite{cpAS}. Let $p$ be the noise probability. Each shared value is
replaced with probability $p$ by a \textit{neighbor} in the
feature's domain, \textit{i.e.} if we assume a total order in the feature
values (which is available for most: binary, real or ordinal), we replace with probability $p$ the feature value by a
\textit{neighbor} in the order: if the feature is binary, then it is
replaced by the other value; otherwise, we pick uniformly at random a
value in the set of neighboring $\pm u$ indexes, clamped to the observed
set of values --- \textit{i.e.} we do not generate unobserved feature
values. If there are more than 20 recorded values for the feature,
then $u=10$; otherwise, $u=2$. Such a neighbor noise process follows
the observed pattern that errors in the real world often generate
neighboring values, for a neighbor relationship that can belong to the
phonetic, typographic, OCR or just keyboard spaces \cite{cpAS}. We
measure the similarity of observed shared vectors using the cosine
similarity, which is a convenient similarity measure for token-based
entity-resolution approaches (the other leading approaches are called
edit based \cite{kssRL}). Given one observation from
$\FDP$, $\ve{x}_{\FDP}$, and one from $\LDP$,
$\ve{x}_{\LDP}$, the cosine similarity between the subvectors of
shared values is denoted $\mathrm{cosSim}(\mathrm{shared}(\ve{x}_{\FDP}),
   \mathrm{shared}(\ve{x}_{\LDP}))$. 

To make sure that there is no difference between the learning
algorithm used after entity-resolution, we always use AdaBoost
\cite{ssIBj}, run for 1000 iterations to learn a linear classifier. We use AdaBoost because of its
guaranteed convergence rates under a weak assumption which fits well
to our setting. We also notice that AdaBoost provably minimizes the
exponential loss, which fits to our theory. However, it is not
possible to find the optimal classifier $\ve{\theta}_0^*$ in closed
form for this loss. Thus, we shall rather \textit{learn} it from
ideally entity-resolved data. This is what we discuss in the following Section.

\subsection{Algorithms and baselines for entity-resolution}

\begin{algorithm}[tb]
   \caption{\greedy($\mathcal{S}$)}
   \label{alg:gree}
   \noindent\hspace{0.5cm} {\bfseries Input:} set $ [m]^2 \supset \mathcal{S} \defeq
   \{(i_{\FDP},i_{\LDP})\}$, where $i_{\FDP}$ (resp. $i_{\LDP}$)
   belongs to indexes of $\FDP$ (resp. $\LDP$)\;
\noindent\hspace{0.5cm} $\mathcal{S}_{\mbox{\tiny{g}}} \leftarrow \emptyset$\;
\noindent\hspace{0.5cm} \textbf{repeat}\\
\noindent\hspace{0.9cm} let $(i^*_{\FDP},i^*_{\LDP}) \defeq \arg\max_{(i, i') \in
  \mathcal{S}} \mathrm{cosSim}(\mathrm{shared}({\ve{x}_{\FDP}}_i),
   \mathrm{shared}({\ve{x}_{\LDP}}_{i'}))$\;
\noindent\hspace{0.9cm} $\mathcal{S}_{\mbox{\tiny{g}}} \leftarrow
\mathcal{S}_{\mbox{\tiny{g}}} \cup \{(i^*_{\FDP},i^*_{\LDP})\}$\;
\noindent\hspace{0.9cm} delete $i^*_{\FDP}$ from $\mathcal{S}$\;
\noindent\hspace{0.9cm} delete $i^*_{\LDP}$ from $\mathcal{S}$\;
\noindent\hspace{0.5cm} \textbf{until} {$\mathcal{S} = \emptyset$}\;
  \noindent\hspace{0.5cm} \textbf{return} $\mathcal{S}_{\mbox{\tiny{g}}}$
\end{algorithm}
\noindent \textbf{The max-weighted matching problem and the \greedy~routine} --- there is a particularly interesting
routine that we call \greedy, which delivers a fast approximation to a
problem that generalizes ours for entity-resolution: maximum weighted
matching for balanced bipartite graphs \cite{aAS}. The instance of this problem
is a balanced complete bipartite graph with non-negative weights, a feasible
solution is a subset of edges covering all vertices, in which each
vertex appears once. The criterion to be maximized is the sum of
weights. If we take as the total (sum-of) cosine similarity the criterion to be
maximized for entity-resolution and note that maximizing the
criterion for the cosine similarities is equivalent to maximizing the
same criterion for (1+cosine similarity)\textit{es}, which is non-negative, then
\greedy, provided in Algorithm \ref{alg:gree}, provides a fast
approximation to entity-resolution, namely $O(|\mathcal{S}|^2 \log
|\mathcal{S}|)$ for a non-optimized implementation. Let us denote
$C^*$ the optimal value of the total cosine similarity. There exists a
long-known method, the Hungarian algorithm, that provably achieves the optimum \cite{kTH}, yet
it requires a significantly more sophisticated implementation to even
reach $O(|\mathcal{S}|^3)$ time complexity. We stick to the greedy
algorithm \greedy~not just for computational reasons and its
straightforwardness of implementation: we in fact do not seek the
optimal solution to entity-resolution but rather wish to find one that
is going to prove good for learning. Whether we can win from both
standpoints --- having a good approximation of the entity resolution
criterion while having the best possible solution for learning ---
shall be discussed as well, and we can
already remark that
\greedy~provides a guaranteed very good constant approximation to $C^*$.
\begin{lemma}\cite{aAS} (Theorem 4)
Let us denote $C_{\mbox{{\tiny \greedy}}}$ as the total entity-resolution similarity retrieved
by \greedy. Then $C_{\mbox{{\tiny \greedy}}} \geq C^*/2$.
\end{lemma}
It is also believed that the actual worst-case approximation provided
by \greedy~is even better \cite{dmsTC}. In our experiments, we test
and compare several algorithms for entity-resolution in various environments.\\

\noindent \textbf{$\LDP$ does not use classes: \greedyER} ---  In
this case, peer $\LDP$ does not have the knowledge of classes and does
not use the knowledge of classes for entity resolution: linking
proceeds from a straightforward use of routine \greedy, as explained
in the boxed algorithm below, where ${\mathcal{S}} \defeq \{(i, i'), i \in [m], i' \in [m]\}$.
\begin{center}
\colorbox{gray!10}{\fbox{%
    \parbox{0.95\textwidth}{Algorithm \greedyER(${\mathcal{S}}$) --- Let
\begin{eqnarray}
\mathcal{S}_{\mbox{\tiny{g}}} & \leftarrow & \greedy(\mathcal{S})\:\:.
\end{eqnarray}
Link all data following $\mathcal{S}_{\mbox{\tiny{g}}}$, return $\hat{S}$.
}}
}
\end{center}

\noindent \textbf{$\LDP$ has classes: \greedyERPC} --- This approach can be implemented
when both $\FDP$ and $\LDP$ have the knowledge of the true class for
their respective observations, which is the setting of \cite{pnhcFL}. The algorithm simply
consists in running \greedy~over the positive class only, then
\greedy~over the negative class only and finally linking the datasets
according to the outputs of \greedy. More formally, if we denote for
short $S_{\FDP} \defeq \{({\ve{x}_{\FDP}}_i , {y_{\FDP}}_i) : i = 1,
2, ..., m\}$ the sample from $\FDP$, and $S_{\LDP} \defeq \{({\ve{x}_{\LDP}}_i , {y_{\LDP}}_i) : i = 1,
2, ..., m\}$ the sample from $\LDP$, then the algorithm can be
summarized as follows, with ${\mathcal{S}}^+ \defeq \{(i, i') : {y_{\FDP}}_i =
      {y_{\LDP}}_{i'} = +1\}$ and ${\mathcal{S}}^- \defeq \{(i, i') : {y_{\FDP}}_i =
      {y_{\LDP}}_{i'} = -1\}$.
\begin{center}
\colorbox{gray!10}{\fbox{%
    \parbox{0.95\textwidth}{Algorithm \greedyERPC(${\mathcal{S}}^+, {\mathcal{S}}^-$) --- Let
\begin{eqnarray}
\mathcal{S}^+_{\mbox{\tiny{g}}}
      & \leftarrow & \greedy(\mathcal{S}^+)\:\:,\\
\mathcal{S}^-_{\mbox{\tiny{g}}}
      & \leftarrow & \greedy(\mathcal{S}^-)\:\:.
\end{eqnarray}
Link the datasets following
$\mathcal{S}^+_{\mbox{\tiny{g}}}$ and
$\mathcal{S}^-_{\mbox{\tiny{g}}}$, return $\hat{S}$.
}}
}
\end{center}

\noindent \textbf{$\LDP$ does not have classes but \textit{learns} them: \greedyERPCZ} ---  In
this case, peer $\LDP$ does not have the knowledge of classes but
\textit{computes} classes using a simple four-steps practical approach relying on
shared features: (i) we run \greedy~as in \greedyER~and then
\textit{discard} couples in $\mathcal{S}_{\mbox{\tiny{g}}}$ whose
similarity is below the median similarity. We then assign a label to
the observations of $\LDP$ still appearing in
$\mathcal{S}_{\mbox{\tiny{g}}}$, by using the correspondence with
$\FDP$ in $\mathcal{S}_{\mbox{\tiny{g}}}$. To complete labelling in
$\LDP$, (ii) we use a simple $k$-NN algorithm inside $\LDP$ which
gives a label to the remaining observations based on the labels
computed from step (i) only. At this
stage, all observations in $\LDP$ are given a class. We then (iii) run
$\greedyERPC$ using the predicted classes for $\LDP$. 

\textit{Notice that we
have no guarantee that the proportion of classes in $\LDP$ will be the
same as in $\FDP$}. For that reason, we end up in general with a subset of
observations in $\FDP$ and $\LDP$ being not linked. To complete
linkage, (iv) we
just run \greedyER~in the subset of remaining observations. The
overall algorithm is sketched in the box below.

\begin{center}
\colorbox{gray!10}{\fbox{%
    \parbox{0.95\textwidth}{Algorithm \greedyERPCZ --- Let
      ${\mathcal{S}} \defeq \{(i, i'), i \in [m], i' \in [m]\}$, and
\begin{eqnarray}
\mathcal{S}_{\mbox{\tiny{g}}} & \leftarrow & \greedy(\mathcal{S})\:\:.
\end{eqnarray}
Let $\varsigma$ be the median similarity in
$\mathcal{S}_{\mbox{\tiny{g}}}$. Discard from
$\mathcal{S}_{\mbox{\tiny{g}}}$ all couples with similarity below
$\varsigma$ and affect classes to observations of $\LDP$ using the
remaining couples:
\begin{eqnarray}
\forall (i, i') \in \mathcal{S}_{\mbox{\tiny{g}}}, {y_{\LDP}}_{i'} &
\leftarrow & {y_{\FDP}}_{i}\:\:.
\end{eqnarray}
Let $S^{\emptyset}_{\LDP}$ denote the subset of observations of $\LDP$
without a label, and $S^{c}_{\LDP}$ denote the subset of observations of $\LDP$
with a label (the total set of observations of $\LDP$ is
$S^{\emptyset}_{\LDP} \cup S^{c}_{\LDP}$). Use a $k$-NN rule to give a
label to observations from $S^{\emptyset}_{\LDP}$:
\begin{eqnarray}
\forall {\ve{x}_{\LDP}}_{i'} \in S^{\emptyset}_{\LDP}, {y_{\LDP}}_{i'} &
\leftarrow & k\mbox{-NN}(S^{c}_{\LDP})\:\:.
\end{eqnarray}
Let
      ${\mathcal{S}}^+ \defeq \{(i, i') : {y_{\FDP}}_i =
      {y_{\LDP}}_{i'} = +1\}$ and ${\mathcal{S}}^- \defeq \{(i, i') : {y_{\FDP}}_i =
      {y_{\LDP}}_{i'} = -1\}$. 
Run \greedyERPC(${\mathcal{S}}^+, {\mathcal{S}}^-$). Let $\mathcal{S}_{\FDP} \subseteq [m]$ and
$\mathcal{S}_{\LDP} \subseteq [m]$ denote (indexes of) the subsets of
observations not linked in $\FDP$ and $\LDP$ (we have $|\mathcal{S}_{\FDP}|=|\mathcal{S}_{\LDP}|$). Run
\greedyER($\mathcal{S}_{\FDP} \times \mathcal{S}_{\LDP}$),  link all
data, return $\hat{S}$.
}}
}
\end{center}

\noindent \textbf{$\LDP$ has \textit{noisy} classes: \greedyERPCN} ---  This corresponds to
running \greedyERPC~in an environment where $\FDP$ has the knowledge of the true class but $\LDP$ has a
knowledge of \textit{noisy} classes. To conform with the vertical
partition setting, we simulate permutation noise over classes in $\LDP$ by the following
process: starting from setting \greedyERPC~/ true classes, given a proportion $p'$, we permute a random positive class
and a random negative class for $[m p']$ iterations in $\LDP$, where $[.]$ gives
integer rounding. We then run \greedyERPC~as in the noise-free
setting. To distinguish with the noise-free environment, we call this
approach \greedyERPCN($p'$). We consider $p' \in \{0.01,
0.02, 0.03, 0.04, 0.05, 0.1, 0.15, 0.2\}$.

\noindent \textbf{"\ideal"} --- because we use simulated domains, we are
able to compute the performances of the ideal entity-resolution
algorithm which essentially returns $S$ instead of $\hat{S}$, and so
$\PERM_* = \mathrm{I}_m$ (the identity matrix) in
eq. (\ref{defXHAT}). 
\begin{center}
\colorbox{gray!10}{\fbox{%
    \parbox{0.95\textwidth}{Algorithm "\ideal" --- return $S$.
}}
}
\end{center}
This gives our "optimal" baseline to compare
against the practical approaches to entity-resolution developed
thereafter. Remark the quotes: we are in fact running AdaBoost to
learn the classifier as seen in Section \ref{sub:sett}, so we cannot
ascertain that we indeed learn $\ve{\theta}^*_0$, but rather compute
an approximation to $\ve{\theta}^*_0$. What we can however certify is
that approximations to the ideal classifier $\ve{\theta}^*_0$ come
only from AdaBoost and are not due to errors in entity resolution.

\subsection{Domains}

To have reliable baselines against which to
compare our algorithms, we have used UCI domains \cite{bkmUR} from
which we have generated our distributed data using the following
process: given a set of shared features, split randomly the remaining
features between $\FDP$ and $\LDP$. The shared features of $\LDP$ are
then noisified using the process describes above in Subsection
\ref{sub:sett}. $\FDP$ always has access to the classes. Remark that
since only the shared features of $\LDP$ are noisified, this
guarantees that the final observation matrix, $\hat{\X}$, obtained
after entity-resolution indeed meets
\begin{eqnarray}
\hat{\X} & \defeq & \left[
\begin{array}{c}
\X_\anchor\\\cline{1-1}
\X_\shuffle \PERM_* 
\end{array}
\right]\label{defXHAT}
\end{eqnarray}
for some unknown $\PERM_*$. This guarantees that the differences
between learning algorithms are not due to the (variable) effect of noise in
features but to the errors of $\PERM_*$ following mistakes in entity
resolution. Table \ref{tab:dom} presents the domains we have used. For
two of them (phishing, transfusion), we have considered two versions,
one in which the shared attributes are highly correlated with the
class ($H$) and one in which they are not ($L$).

\begin{table*}[t]
\centering
{\footnotesize
\begin{tabular}{lccc||cc|c}
Domain & $m$ & $d$ & $s$ & shared & linear correlations wrt class & C.Err\\ \hline \hline
magic & 19020 & 10 & 4 & 0, 1, 2, 3 & $0.29, 0.25, 0.11, -0.02$& $10^{-4}$ \\ \hline
page & 5473  &
10 &
3 & 0, 1, 2 & $-0.12, -0.03, -0.09$& 1.83\\\hline
sonar & 208 & 60 & 3 & 0, 1, 2 & $0.27, 0.23, 0.19$& 3.69 \\ \hline
winered & 1599 & 11 & 2 & 7, 8 & $-0.15, -0.003$& 6.02 \\ \hline
eeg & 14980 & 14 & 4 & 0, 1, 2, 3 & $0.01, -0.08, 0.04, -0.08$& 6.08
\\ \hline 
phishing$_H$ & 11055  &
30 &
5 & 5, 6, 7, 13, 25 & $0.34, 0.30, 0.71, 0.69, 0.34$  & 7.39 \\ \hline
winewhite & 4898 & 11 & 3 & 0, 1, 2 & $-0.08, -0.21, -0.0007$& 8.57\\ \hline
breast-wisc & 699  &
9 &
2 & 0, 1 & $-0.68, -0.78$ & 9.21\\ \hline
fertility & 100  &
9 &
3 & 2, 3, 4 & $-0.02, -0.09, 0.03$ & 12.22\\\hline
banknote & 1372 & 4 & 1 & 0 & $-0.72$& 13.14 \\ \hline
creditcard & 14599  & 23  & 4  & 1, 2, 3, 4 & $-0.02, 0.01, -0.02, 0.004$& 14.96\\ \hline
qsar & 1055  &
41  &
4 & 2, 5, 8, 9 & $-0.28, -0.16, -0.05, 0.16$& 16.67\\\hline
transfusion$_H$ & 748  &
4  &
1 & 0 & $-0.24$& 17.36\\ \hline
transfusion$_L$ & 748  &
4  &
1 & 3 & $-0.03$& 17.80 \\ \hline
firmteacher & 10800  &
16  &
2 & 0, 1, 2, 3, 4 & $-0.22, 0.29, -0.25, 0.18, 0.10$& 19.78 \\ \hline
ionosphere & 351  &
33  &
1 & 0 & 0.45 & 20.57\\ \hline
phishing$_L$ & 11055  &
30 &
4 & 0, 1, 2, 3 & $0.09, 0.05, -0.06, 0.05$& 24.35 \\ \hline \hline
\end{tabular}
}
\caption{UCI domains used \cite{bkmUR}. For each domain, we indicate
  the total number of examples ($m$), total number of features ($d$)
  and the number of shared features used in our simulations ($s$). We
  then indicate the list of shared features (indexes as recorded in
  the UCI) and the list of linear correlations with the class for each
  of them. We finally indicate the average
  class errors in entity resolution for \greedyER~(C.Err), \textit{i.e.} the proportion of examples
  from one class matched with examples from the other class. Domains
  are listed in increasing value of C.Err.
\label{tab:dom}}
\end{table*}

\subsection{General results}


\begin{table*}[t]
\centering
\scalebox{.55}{
\begin{tabular}{cc||d||d|rrrr|r|rrrrrrrr}\hline \hline
Domain & Noise $p$ & "Ideal" & \multicolumn{14}{a}{\greedyER[as is |
  \greedyERPCZS~| \greedyERPCS~| \greedyERPCNS]} \\
& & & as is & \multicolumn{4}{c|}{\greedyERPCZS($k$)} &
\greedyERPCS & \multicolumn{8}{c}{\greedyERPCNS($p'$)}\\
  & & & & $1$ & $2$ & $5$ & $10$ & & $0.01$ & $0.02$ & $0.03$ &
  $0.04$ & $0.05$ & $0.10$ & $0.15$ & $0.20$ \\ \hline \hline
\multirow{3}{*}{magic} & 0.05 & 21.14 & 21.15 & \specialcell{21.08} & \specialcell{21.08} & \specialcell{21.06} &
\specialcell{21.11} & 21.15 & \specialcell{21.04} & \specialcell{21.10} & \specialcell{21.08} & 21.19 & 21.21 & 21.65 & \tred{*}{22.28} &
\tred{*}{22.70}\\
 & 0.1 & 21.14 & 21.19 & 21.42 & \tred{*}{21.53} & 21.21 & 21.17 & \specialcell{21.08} & 21.21 &
 \specialcell{21.08} & 21.17 & 21.18 & 21.33 & \tred{*}{21.82} & \tred{*}{22.57} & \tred{*}{23.58}\\
 & 0.3 & 21.16 & \specialcell{21.14} & 21.26 & 21.33 & \specialcell{21.14} & 21.16 & \specialcell{21.14} & \specialcell{21.06} &
 21.16 & 21.24 & \tred{*}{21.62} & \tred{*}{21.72} & \tred{**}{22.34} & \tred{**}{23.75} & \tred{**}{25.31}\\ \hline
\multirow{3}{*}{page} & 0.05 & 27.62 & \specialcell 25.16 & \specialcell 25.31	&\specialcell 25.36	&\specialcell 25.31	&\specialcell 25.14 &
\specialcell 25.85	& \specialcell 25.82	& \specialcell\tred{*}{27.26}& 	\specialcell\tred{*}{27.28}	& \tred{*}{27.66}	& \tred{*}{28.02}	& \tred{*}{29.23}	&
\tred{*}{30.31}	& \tred{*}{31.04}\\ 
 & 0.1 & 27.17 & \specialcell 26.11 & \specialcell 26.03 &	\specialcell 25.43 &	\specialcell\tgreen{*}{24.65} &\specialcell	\tgreen{*}{24.61} &\specialcell 25.63
 & \specialcell 26.63	& \specialcell 26.31	& \specialcell 26.55	& \specialcell 26.48	& 27.39	& 27.96	& \tred{*}{31.32}	&
 \tred{*}{34.59}\\
 & 0.3 & 27.66 & \specialcell 24.83 & \specialcell\tred{*}{26.43} &	\specialcell 26.18 &	\specialcell 25.79 &	\specialcell 25.14 & \specialcell\tred{*}{25.87}
 & \specialcell 25.67	& \specialcell 26.24	& \specialcell 26.49	 &\specialcell 26.82	& \specialcell 26.65	& 28.10	&
 29.21	 & 33.44\\ \hline
\multirow{3}{*}{sonar} & 0.05 & 26.93 & \specialcell 25.95 & \specialcell 25.00 & 28.31 & \specialcell 24.95
& 
\specialcell 26.40 & \specialcell 26.92 & \specialcell 24.50 & \specialcell 25.45 & \specialcell 22.55 & \specialcell 24.05 & \specialcell 24.48 & \specialcell 23.07 & \specialcell 23.05
& \specialcell 25.90\\
 & 0.1 & 26.88 &\specialcell 26.45 & \specialcell 25.45 & 28.38 & \specialcell 24.52 & 27.45 & 26.88 & \specialcell 24.98 &
 27.40 & \specialcell 25.48 & \specialcell 26.50 & 28.40 & \specialcell 25.02 & 29.33 & \specialcell 25.55\\
 & 0.3 & 26.02 & \specialcell 25.05 & 26.45 & \specialcell 24.55 & \specialcell 24.98 & \specialcell 25.02 & \specialcell 25.50 & \specialcell 23.59 &
 \specialcell 22.14 & \specialcell 22.62 & \specialcell 24.07 & \specialcell 24.55 & \specialcell 25.50 & 27.02 & 26.90\\ \hline
\multirow{3}{*}{winered} & 0.05 & 26.08 & 26.58 & 26.20 & 26.32 & 26.45 &
26.58 & 26.77 & 26.45 & 26.83 & 26.20 & 26.45 & 26.33 & 26.64 & 28.20 &
\tred{*}{28.77} \\
 & 0.1 & 26.57 & 26.76 & 27.26 & 27.26 & 26.82 & 27.57 & 27.01 & 26.82 &
 26.82 & 26.76 & 26.39 & 26.88 & 26.76 & 27.95 & 28.64 \\
 & 0.3 & 26.58 & 27.58 & 27.01 & 27.20 & 27.64 & 26.89 & 26.89 & 26.83
 & 
 26.70 & 27.01 & 26.64 & 26.83 & 26.76 & 26.82 & 27.08 \\ \hline
\multirow{3}{*}{eeg} & 0.05 & 45.18 & \specialcell 45.05 & \specialcell 44.43 & \specialcell 44.99 & \specialcell 43.88
& 
\specialcell 44.16 & \specialcell 45.16 & 45.29 & 45.59 & 45.72 & 45.62 & 45.84 & 46.46 & 46.47 &
46.50\\ 
 & 0.1 & 45.79 & 45.92 & \specialcell 45.47 & \specialcell 45.61 & \specialcell 45.20 & 45.80 & 45.79 & 45.83 &
 46.03 & \specialcell 45.74 & 46.19 & 46.45 & \specialcell 43.96 & \specialcell\tgreen{*}{43.42} & \specialcell 43.99 \\
 & 0.3 & 45.19 & 46.10 & 46.08 & 46.58 & 46.68 & 46.60 & \specialcell 45.18 & 45.40 &
 45.45 & 45.96 & \specialcell 45.02 & 45.61 & \specialcell 45.07 & \specialcell 45.07 & \specialcell 45.04\\ \hline
\multirow{3}{*}{phishing$_H$} & 0.05 & 8.03 & 8.40 &
8.08 & 8.18 & 8.70 & 8.62 & 8.17 & 8.05 & 8.05 & 8.14 & 8.23 &
\tred{*}{8.95} & \tred{*}{10.84} & \tred{**}{12.85} & \tred{**}{15.41}\\ 
 & 0.1 & 7.92 & 8.35 & 8.23 & 8.16 & 8.36 & 8.51 & 7.92 & \specialcell\tgreen{*}{7.76} &
 8.01 & 8.21 & 8.72 & 8.76 & \tred{*}{9.61} & \tred{**}{12.86} & \tred{**}{15.29}\\ 
 & 0.3 & 7.96 & 8.90 & 9.15 & 8.91 & 9.01 & 8.86 & 8.39 & \tgreen{*}{8.09} &
 \tgreen{*}{8.17} & 8.46 & 8.67 & 9.05 & \tred{*}{11.49} & \tred{**}{13.5} & \tred{**}{15.46}\\ \hline
\multirow{3}{*}{winewhite} & 0.05 & 30.58 & \specialcell 30.34 & \specialcell 30.30 & \specialcell 30.28 &
\specialcell 30.36 & \specialcell 30.31 & 30.60 & 30.65 & 30.65 & 30.69 & 30.65 & 30.85 & \specialcell 30.30 &
30.69 & \specialcell 30.30\\
 & 0.1 & 30.99 & 31.11 & \specialcell 30.91 & \specialcell 30.97 & \specialcell 30.97 & \specialcell 30.85 & \specialcell 30.75 & \specialcell 30.95 &
 31.03 & 31.26 & 31.03 & 31.11 & 31.58 & 31.60 & 31.95\\
 & 0.3 & 30.95 & 32.79 & \tgreen{*}{31.44} & \tgreen{*}{31.36} & \tgreen{*}{31.31} & \tgreen{*}{31.16} & \tgreen{*}{30.97} & \tgreen{*}{31.19} &
 \tgreen{*}{31.17} & \tgreen{*}{30.89} & \tgreen{*}{31.07} &
 \tgreen{*}{30.99} & \tgreen{*}{31.48} & 32.01 & 32.97\\  \hline
\multirow{3}{*}{breast-wisc} & 0.05 & 3.00 & 3.71 & 3.43 & 3.57 &
3.57 & 3.43 & \specialcell\tgreen{*}{2.43} & \specialcell 2.57 & 3.29 & 3.29 & 3.14 & 3.57 & 3.43 &
3.57 & 3.00\\
 & 0.1 & 3.00 & 3.86 & 3.86 & 3.43 & 3.29 & 4.29 & 3.15 & 3.29 &
 3.00 & 3.29 & 3.00 & 3.72 & 4.01 & 4.29 & 5.30 \\
 & 0.3 & 2.71 & 6.29 & 5.58 & 5.72 & 5.57 & \tgreen{*}{4.86} & \tgreen{*}{3.28} & \tgreen{*}{3.43} &
 \tgreen{*}{3.43} & \tgreen{*}{3.85} & \tgreen{*}{3.86} &
 \tgreen{*}{4.14} & \tgreen{*}{4.43} & 4.86 & 6.01\\ \hline
\multirow{3}{*}{fertility} & 0.05 & 43.00 & 49.00 & 43.00 & \specialcell 33.00 & \specialcell 41.00 &
44.00 &\specialcell 42.00 & 43.00 & 52.00 & 50.00 & 52.00 & 52.00 & 43.00 & 48.00 & 55.00\\
 & 0.1 & 43.00 & \specialcell 41.00 & \specialcell 41.00 & 47.00 & \specialcell 42.00 & 52.00 & 45.00 & 44.00 & 44.00 &
 50.00 & 47.00 & 44.00 & 53.00 & 47.00 & 55.00\\
 & 0.3 & 46.00 & 50.00 & 55.00 & 58.00 & 49.00 & 59.00 & 49.00 & 49.00 & 54.00 &
 50.00 & 54.00 & 54.00 & 53.00 & 55.00 & 43.00\\ \hline
\multirow{3}{*}{banknote} & 0.05 & 2.77 & 13.26 & 13.71 & 13.92 & 12.83 &
13.92 & \tgreen{*}{7.95} & \tgreen{*}{7.65} & \tgreen{*}{7.43} & \tgreen{*}{7.80} & \tgreen{*}{7.72} & \tgreen{*}{8.31} & \tgreen{*}{9.98} & 12.82 &
14.93\\
 & 0.1 & 2.77 & 14.94 & 14.79 & 14.50 & 15.23 & 14.79 & \tgreen{*}{11.88} & \tgreen{*}{11.51} &
 12.68 & \tgreen{*}{12.31} & 12.53 & 13.63 & 14.72 & 16.25 & 17.27 \\
 & 0.3 & 2.91 & 12.89 & 13.84 & 12.74 & 12.39 & 12.97 & 10.06 & 10.64 &
 11.15 & 10.78 & 11.73 & 12.03 & 13.55 & 14.69 & 16.91 \\\hline
\multirow{3}{*}{creditcard} & 0.05 & 23.26 & 23.26 &  23.26 &  23.26 &
23.26 &  23.26 &  23.26 &  23.26 &  23.26 &  23.26 &  23.26 &  23.26 &
23.26 &  23.26 &  23.26 \\
 & 0.1 & 23.26 & 41.87 & 40.66 & 41.46 & \tgreen{*}{36.91} & 36.89 & \tgreen{**}{23.26} & \tgreen{**}{23.26} &
 \tgreen{**}{23.26} & \tgreen{**}{23.26} & \tgreen{**}{23.26} & \tgreen{*}{26.19} & 42.65 & 43.08 & 44.36 \\
 & 0.3 & 23.26 & 42.49 & 41.19 & 42.03 & \tgreen{*}{38.82} & \tgreen{*}{36.51} & \tgreen{**}{23.26} & \tgreen{**}{23.26} &
 \tgreen{**}{23.26} & \tgreen{*}{24.72} & \tgreen{*}{25.01} &
 \tgreen{*}{32.28} & 40.87 & 41.75 & 40.89\\ \hline
\multirow{3}{*}{qsar} & 0.05 & 21.80 & 23.51 & 23.60 & 22.94 & 23.03 &
24.17 & \specialcell 21.62 & 21.90 & \specialcell 21.72 & \specialcell 21.72 & 22.19 & 22.19 & 22.28 & 22.38 &
23.51\\ 
 & 0.1 & 21.51 & 23.22 & 23.02 & 23.40 & 23.78 & 23.31 & 22.27 & 21.79 &
 21.70 & 21.99 & 21.79 & 21.89 & 22.36 & 22.75 & 22.75\\ 
 & 0.3 & 21.81 & 22.85 & 23.13 & 22.19 & 23.13 & 22.27 & 22.19 & 22.28 &
 22.28 & 21.81 & \specialcell 21.71 & 21.90 & 22.00 & 22.66 & 22.76 \\ \hline
\multirow{3}{*}{transfusion$_H$} & 0.05 & 39.57 & \specialcell 36.10 & \specialcell\tgreen{*}{33.82} & \specialcell\tgreen{*}{34.09} &
\specialcell\tgreen{*}{34.89} & \specialcell\tgreen{*}{34.36} & \specialcell\tred{*}{39.03} & \tred{*}{39.84} & \tred{*}{39.71} & \specialcell\tred{*}{38.77} & \specialcell\tred{*}{38.64} & \specialcell 37.44 & \specialcell\tgreen{*}{34.89} &
\specialcell\tgreen{*}{33.43} & \specialcell 35.03\\
 & 0.1 & 39.72 & \specialcell 35.83 & \specialcell 36.09 & \specialcell 35.95 & \specialcell\tgreen{*}{33.55} & \specialcell\tgreen{*}{33.28} & \tred{*}{40.38} & \specialcell\tred{*}{38.92} &
 \specialcell\tred{*}{37.57} & \specialcell 36.89 & \specialcell 34.88 & \specialcell 34.89 & \specialcell\tgreen{*}{33.16} & \specialcell\tgreen{*}{33.82} & \specialcell 34.63\\
 & 0.3 & 38.37 & \specialcell 35.55 & \specialcell 35.83 & \specialcell 34.08 & \specialcell 34.89 & \specialcell 34.89 & \tred{*}{38.65} & \specialcell\tred{*}{37.98} &
 \specialcell\tred{*}{37.44} & \specialcell 37.04 & \specialcell 37.17 & \specialcell 35.97 & \specialcell 35.62 & \specialcell 35.03 & \specialcell 34.89\\ \hline
\multirow{3}{*}{transfusion$_L$} & 0.05 & 38.64 & \specialcell 34.65 & \specialcell 34.77 &
\specialcell 34.76 &
\specialcell 34.22 & \specialcell 34.89 & \specialcell 38.23 & \specialcell 36.90 & \specialcell 37.43 & \specialcell 34.90 & \specialcell 33.96 & \specialcell 34.23 & \specialcell 33.17 &\specialcell
34.63 & \specialcell 33.57\\
 & 0.1 & 39.02 & \specialcell 35.15 & \specialcell 34.75 & \specialcell 35.29 & \specialcell 35.16 & \specialcell 34.48 &\specialcell 37.16 & \specialcell 38.09 &
\specialcell 38.09 & \specialcell 37.17 & \specialcell 36.36 & \specialcell 33.68 & \specialcell 33.55 & \specialcell 33.55 & \specialcell 33.41 \\
 & 0.3 & 39.29 & \specialcell 34.09 & \specialcell 33.41 & \specialcell 35.16 & \specialcell 34.76 & \specialcell 34.63 & \tred{*}{39.29} & \specialcell\tred{*}{38.77} &\specialcell
 \tred{*}{37.82} & \specialcell 35.82 & \specialcell\tred{*}{37.02} & \specialcell 35.82 & \specialcell 35.56 & \specialcell 34.08 & \specialcell 32.48\\ \hline
\multirow{3}{*}{firmteacher} & 0.05 & 12.45 & 17.57 & 18.03 & 18.23 &
17.75 & 18.00 & \tgreen{**}{12.71} & \tgreen{**}{12.68} & \tgreen{**}{12.71} & \tgreen{*}{13.06} & \tgreen{**}{13.02} & \tgreen{**}{13.35} & \tgreen{*}{14.81} &
\tgreen{*}{15.90} & 17.38\\ 
 & 0.1 & 12.39 & 21.03 & 21.06 & 21.29 & 21.51 & 21.54 & \tgreen{**}{12.89} & \tgreen{**}{12.71} &
 \tgreen{**}{12.73} & \tgreen{**}{12.72} & \tgreen{**}{13.14} & \tgreen{**}{13.36} & \tgreen{**}{14.82} & \tgreen{*}{16.98} & \tgreen{*}{18.06} \\
 & 0.3 & 12.35 & 20.45 & \tred{*}{21.12} & \tred{*}{21.16} & 20.32 & 20.34 & \tgreen{**}{12.54} & \tgreen{**}{12.45} &
 \tgreen{**}{12.42} & \tgreen{**}{12.73} & \tgreen{**}{12.81} &
 \tgreen{**}{13.00} & \tgreen{**}{14.54} & \tgreen{**}{16.05} &
 \tgreen{**}{17.44} \\ \hline
\multirow{3}{*}{ionosphere} & 0.05 & 11.95 & 19.04 & 19.34 & 19.61 &
20.19 & 19.34 & 14.51 & 14.23 & 16.49 & 14.79 & 16.48 & 14.80 & 17.63 &
19.35 & 19.35 \\
 & 0.1 & 10.28 & 16.25 & 14.54 & 15.13 & 15.96 & 16.54 & 13.42 & 14.56 &
 15.41 & 15.68 & 15.39 & 15.68 & 15.39 & 15.97 & 16.26\\ 
 & 0.3 & 10.84 & 17.95 & 19.38 & \tred{*}{22.50} & 20.23 & 17.93 & \tgreen{*}{13.97} & \tgreen{*}{13.40} &
 13.70 & 14.55 & 15.98 & 17.38 & 15.68 & 18.24 & 18.23\\ \hline
\multirow{3}{*}{phishing$_L$} & 0.05 & 7.97 & 14.80 & 14.82 & 14.99 &
15.02 &
14.98 & \specialcell\tgreen{**}{7.91} & \tgreen{**}{8.27} & \tgreen{**}{8.44} & \tgreen{**}{8.45} & \tgreen{**}{8.61} & \tgreen{**}{8.83} & \tgreen{**}{9.94} & \tgreen{**}{10.16} &
\tgreen{**}{11.18}\\ 
 & 0.1 & 7.89 & 11.11 & 11.11 & 11.11 & 11.11 & 11.11 & \tgreen{**}{8.02} & \tgreen{**}{7.92} &
 \specialcell\tgreen{**}{7.82} & \specialcell\tgreen{**}{7.82} & \tgreen{**}{7.91} & \tgreen{**}{8.11} & \tgreen{**}{8.50} & \tgreen{**}{9.32} & \tgreen{*}{10.65} \\
 & 0.3 & 7.91 & 13.73 & 13.73 & 13.73 & 13.73 & 13.73 & \tgreen{**}{8.29} & \tgreen{**}{8.51} &
 \tgreen{**}{8.47} & \tgreen{**}{8.44} & \tgreen{**}{8.60} &
 \tgreen{**}{8.80} & \tgreen{**}{9.16} & \tgreen{**}{9.81} & \tgreen{**}{10.54} \\
\hline \hline
\end{tabular}
}
\caption{{ Results (test errors) comparing, for three values of the shared features
  noise ($p$), the various approaches built on top of \greedyER~to
  "Ideal". Domains are listed in the same order as in Table
  \ref{tab:dom}. Grey shaded cells are the results of "Ideal" and
  \greedyER~(indicated "as is"). Blue shaded cells denote results that are better (but not
  necessarily statistically better) than "Ideal". Red text denote
  results that are statistically outperformed by \greedyER; green text
  denote results of \greedyER[\greedyERPCZS~| \greedyERPCS~|
  \greedyERPCNS] statistically better than greedyER. One
  star ($*$) indicated $p$-value in $(10^{-6}, 10^{-2}]$, two stars
  ($**$) indicated $p$-value $\leq 10^{-6}$ (best viewed in color).
\label{tab:res}}}
\end{table*}

Results are displayed in Table \ref{tab:res}. From those results, several
observations come to the fore. First, the larger the number of errors
of entity resolution among classes for \greedyER~(Table \ref{tab:dom},
C.Err), the more beneficial
are the approaches using the class information for entity
resolution. On domains firmteacher, ionosphere, phishing, using the
class information is almost always on par with or (significantly)
better than \greedyER. Second, the improvement can be extremely
significant as witnessed by domains creditcard or firmteacher, with
almost 20 $\%$ improvement when using (even noisy) classes on
creditcard, and still up to 6$\%$ improvement when using predicted
classes (\greedyERPCZ) on creditcard. This is
very good news because the shared features we used on creditcard --- sex,
education, marriage, age --- are typically those that would be shared
in a federated learning setting. 

Another observation may be made: on all domains but one (banknote),
carrying entity-resolution is susceptible to compete against
"Ideal". On the majority of domains, there exists a version of
\greedyER[as is |
  \greedyERPCZS~| \greedyERPCS~| \greedyERPCNS]~which \textit{beats} "Ideal" --- even when not statistically
in most cases ---. On few domains, page, sonar, transfusion (both $H$
and $L$), using class
information yields results that almost always beat the "Ideal"
baseline. One explanation comes from the fact that all learners,
including "Ideal", use AdaBoost for the same number of
iterations. On these domains, the models learned after entity-resolution tend to be
slightly less sparse than for "Ideal". So, it seems reasonable that entity resolution, when carefully used as may be the case
with class information, may force the spread of AdaBoost's feature leveraging to a
larger number of relevant features, compared to "Ideal" which focuses
on a smaller set during the thousand iterations allocated and thus comes up with
a model that can be more sparse but less accurate. Also, considering phishing, we see
that having shared features that are more correlated with the class
(phishing$_H$ vs phishing$_L$) certainly helps to compete against
"Ideal", in particular when one peer does not have classes. The same
observation can be made for transfusion, even when the gains with more
correlated features are less important in this case, which can be due
to the small number of shared features.

If we now compare the two approaches of \greedyER~using class information (with classes, even
noisy, vs without), then it is apparent that having noisy classes --- with up to $20\%$ noise --- can very
significantly help against \greedyER~compared to carrying out entity
resolution without ground class information (but learning classes) as
in \greedyERPCZ. Our approach that learns classes in \greedyERPCZ~is
simple but still manages to deliver significant improvements in some
cases, typically high noise for shared features (winewhite, creditcard) or shared features
sufficiently correlated with class (transfusion$_H$).

Finally, we keep in mind that these results are obtained for
simulations that include in general a small number of shared features
(2.8 on average) and a shared feature noise that ranges up to
$p=$30$\%$, which would correspond to relatively challenging practical
settings. This suggests that if we exclude pathological domains like
banknote in our benchmark, there would be for most domains good reasons to carry out
tailored approaches to entity resolution for learning with the
ambition to challenge the unknown learner having access to
the ideally entity-resolved data. This is not surprising: it is known
that the sufficient statistics for the class is very simple for many
relevant losses \cite{pnrcAN}, so we should not expect perfect entity
resolution to be necessary to improve learning performance.

\subsection{Experimental check of immunity to $\PERM_*$ of large
margin classification}

\begin{figure}[!t]
\centering
\begin{tabular}{cc} \hline\hline
\includegraphics[width=.45\linewidth]{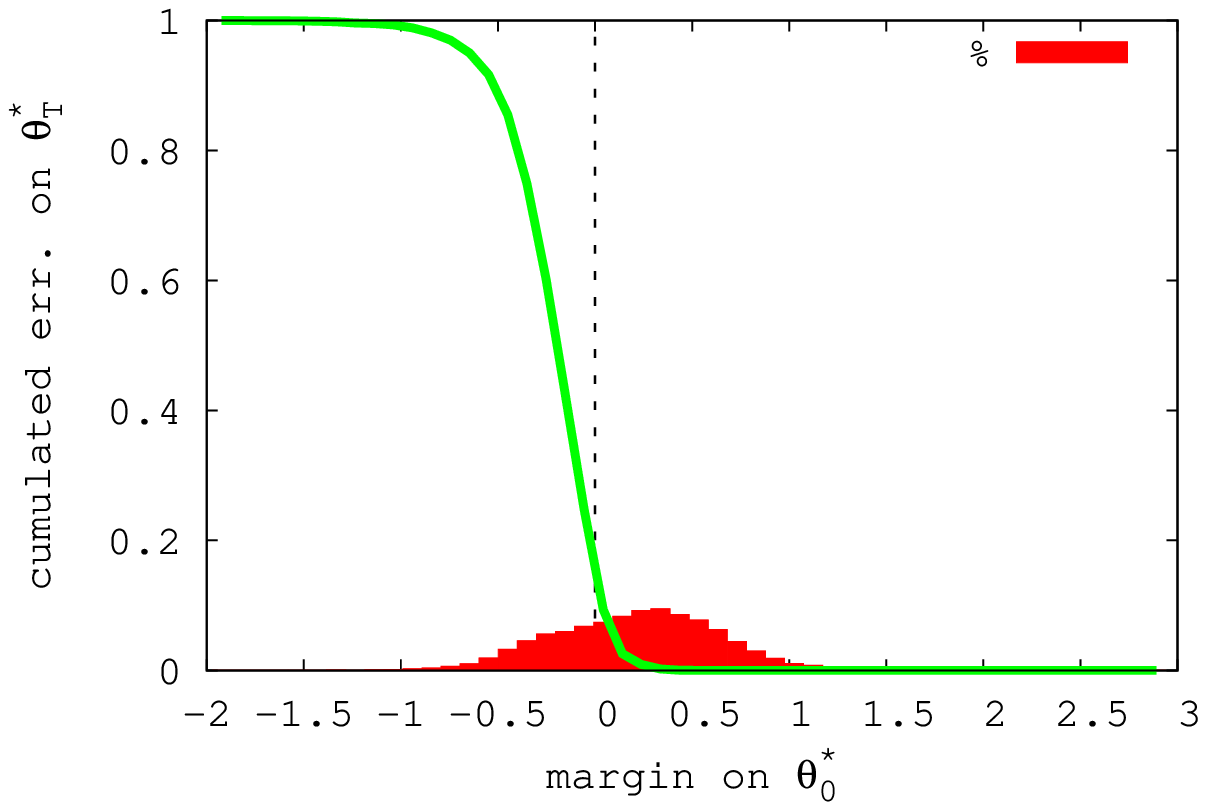}
& \includegraphics[width=.45\linewidth]{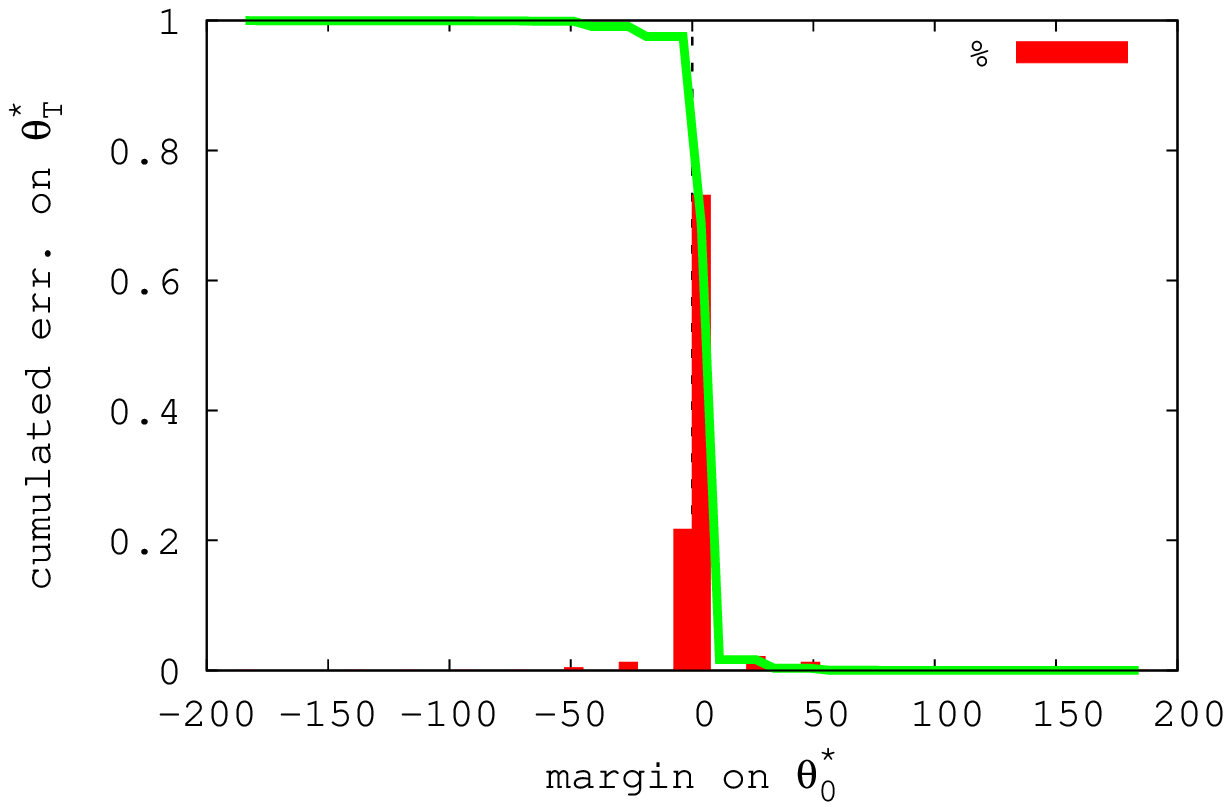}\\
winewhite, \greedyERPCZ(5) & creditcard, \greedyERPCZ(10)
\\ \hline\hline
\end{tabular}
\caption{Margin distribution on two domains with shared attribute
  noise $p=0.3$. The red histogram displays the distribution of
  margins of $\ve{\theta}_0^*$ on training. The green curve is
  the cumulated relative error of $\ve{\theta}_T^*$ \textit{above} some margin
  $x$. For example, on winewhite, less than $20\%$ of the errors on training happen on
  examples with positive margin, and approximately
  \textit{no} error happens on examples with positive margin
  above 0.5 --- in other words, \textit{all} examples with margin above 0.5 on
  $\ve{\theta}_0^*$ receive the right class from $\ve{\theta}_T^*$ and
  so, following Definition \ref{defIMMUNEMARGIN}, $\ve{\theta}_T^*$
  happens to be \textbf{immune} to entity
  resolution at margin 0.5. Since the maximal margin recorded for
  $\ve{\theta}_0^*$ is $\approx 3.0$, we see in this example that
  immunity occurs for a comparatively small positive margin (best viewed in color, see text for details).
\label{fig:marg}}
\end{figure}

\begin{figure}[!t]
\centering
\begin{tabular}{c} \hline\hline
\includegraphics[width=.93\linewidth]{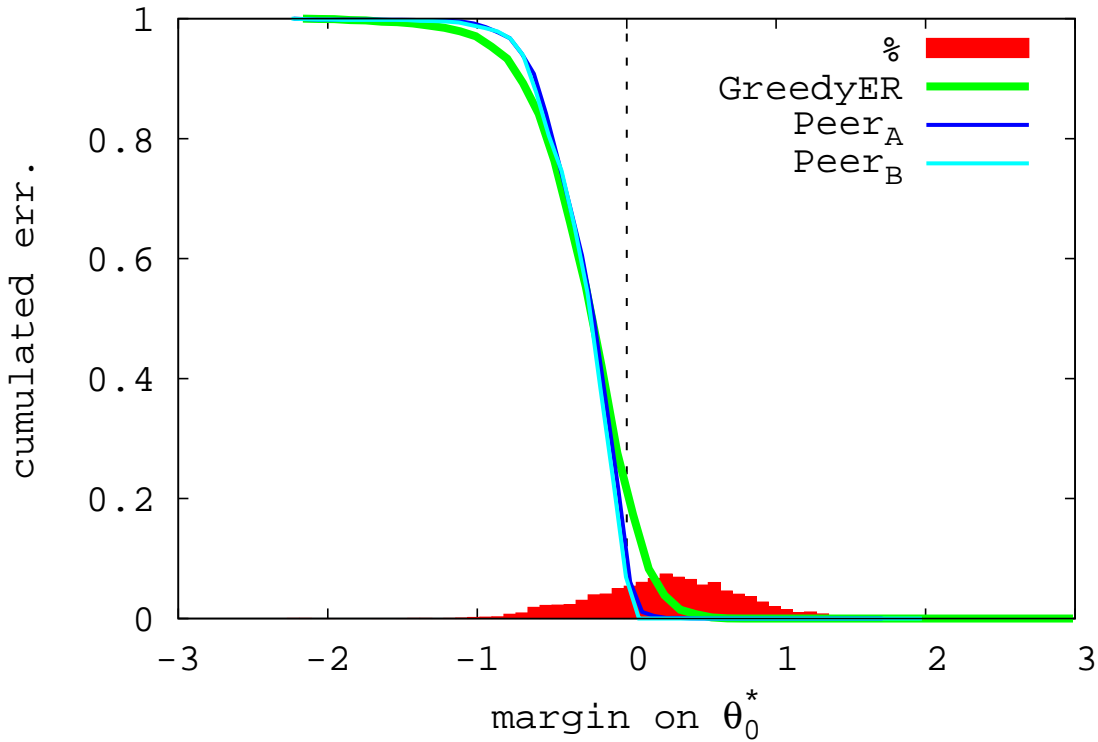}
\end{tabular}
\caption{Margin distribution on $\ve{\theta}_0^*$ on domain winered, and cumulative errors
  comparing \greedyER~and the two peers $\FDP$ and $\LDP$. Convention
  follows Figure \ref{fig:marg} (best viewed in color, see text for details).
\label{fig:marg2}}
\end{figure}

In Section \ref{sec-marg}, we essentially
show that all examples receiving large margin classification on
$\ve{\theta}_0^*$ are given the right class by
$\ve{\theta}_T^*$. To our knowledge, such a result has never been
documented, even experimentally, but it would represent a significant
support for federated learning since one can hope, by joining diverse
databases, to increase not just the accuracy of classifiers but in
fact the margins over examples, thereby bringing immunity to the mistakes of
entity resolution for examples that would attain sufficiently large
margins. But how "large" a margin is necessary ? On each
domain, we have computed the
margin distributions of $\ve{\theta}_0^*$ --- approximated by the
output of AdaBoost
ran on the training sample $S$ for twice the usual number of
iterations, that is, 2000\footnote{In fact, we do this for all cross
  validation folds.}. We
then compute, for all examples, whether they are given the right
class by $\ve{\theta}_T^*$. We finally compute the cumulative error
distribution, in between 0 and 1, of $\ve{\theta}_T^*$. For any $x \in
[\kappa_{m}, \kappa_M]$ (the interval of observed margins), the
cumulative error on $x$ is just the proportion of errors occurring for
margins in the interval $[x, \kappa_M]$. When $x = \kappa_{m}$, this
is just 1. Figure \ref{fig:marg} provides two examples of curves obtained,
which does not just validate immunity: on winewhite, it shows that it can happen for
a quite small margin ($\approx 0.5$) with respect to the maximal
margin ($\kappa_M \approx 3.0$), which reinforces the
support for federated learning. On creditcard, we have $\kappa_M
\approx 188$ while immunity happens at margin $\approx 100$. Less than
$1\%$ of mistakes have margin larger than $30$.

In Figure \ref{fig:marg2}, we provide an example comparison for
domain winered, against the two peers $\FDP$ and $\LDP$. In this case,
\greedyER~achieves error more than $1.7\%$ lower than both peers. We can see
from the plots that errors occur on peers for smaller margins than for
\greedyER, yet the cumulative error slope is much steeper for both peers,
indicating that \greedyER~achieves a better job at classifying hard
examples (small optimal margin), an observation that can perhaps be attributed to the fact
that \greedyER~successfully handles a set of features which is bigger
than that of each peer.

Finally, in table \ref{tab:immmag}, we provide the minimal immunity margin on
one domain for \greedyERPCN, for different values of $p'$, that is,
the minimal $x$ for which there is no error on examples with margin
$x$ on $\ve{\theta}_0^*$. We can see that this minimal margin
largely increases with noise, and so increasing noise in the
entity resolution process degrades the margin picture, which is also
consistent with the fact that the error of $\ve{\theta}_T^*$ also
significantly increases.

\begin{table*}[t]
\centering
{\footnotesize
\begin{tabular}{rrrrrrrrr}\hline \hline
$p' = 0$ & $p' = 0.01$ & $p' = 0.02$ & $p' = 0.03$ & $p' = 0.04$ & $p'
= 0.05$ & $p' = 0.1$ & $p' = 0.15$ & $p' = 0.2$\\ \hline
0.068 & 0.086 & 0.218 & 0.359 & 0.362 & 0.513 & 0.891 & 1.113 & 0.913 \\ \hline \hline
\end{tabular}
}
\caption{Minimal immunity margin on domain magic ($p = 0.3$) for \greedyERPCN.
\label{tab:immmag}}
\end{table*}
\section{Discussion and conclusion}
\label{sec:conclusion}

This paper describes a global picture guaranteeing that the errors
of an approximate entity resolution algorithm do not snowball with
those of learning linear models, in the framework of federated
learning. The key parts rely on essential properties of the entity
resolution algorithm and, to a lesser extent, on the design
(regularization) of the loss. At this moderate price, the main message that comes from our
results is very general as it roughly states that
\begin{center}
\textit{"\textbf{any} entity resolution algorithm making errors
  bounded in size and magnitude, used before minimizing \textbf{any}
  sufficiently regularized loss, yields a minimizer classifier that converges to the optimum learned knowing the \textbf{perfect} entity resolution"}
\end{center}

Indeed our result holds for
a broad class of losses, not even restricted to convex nor
classification calibrated losses, thereby
generalizing very significantly a result developed in the privacy setting
for a Taylor approximation to the logistic loss
\cite{hhinpstPF}. 

Experimentally, the part of our theory that relies on entity resolution
suggests some very simple
modification(s) that can be carried out on existing entity resolution
algorithms to bring algorithms tailored to be a pre-processing
stage to learning. Drilling down into such a link is not the purpose
of our paper, yet our experiments on simple modifications of
a greedy token-based approach displays potential for
significant improvements. We exemplify this on two modifications: (i)
when both peers have classes but one has noisy classes and
(ii) when only one peer has classes. Even with such simple approaches
to integrate the knowledge of classes, our experiments already display the
possibility to compete with the learner that would have access to
the ideally linked data.

We leave two important open questions: (1) on the formal side, the extension of our results to the case
where vertical partition does not hold anymore and some examples of
one peer do not necessarily have a correspondence in the other peer
(and we do not know which ones), (2) on the privacy side,
the question as to how our
results can be pushed to efficient algorithms in a \textit{secure}
federated learning environment where entity resolution
has to comply with privacy constraints. The strength of our results makes it reasonable
to believe that a substantial weakening of the vertical partition
setting to get to (1) is available at affordable formal expense
for the pipeline entity resolution-learning. This is crucial because
this pipeline is pivotal to federated learning: to our knowledge,
there is only \textit{one} exception to this
pipeline \cite{pnhcFL}. It was shown there how one can learn a model
from sufficient statistics of the class instead of examples, many of
which would not require entity-resolution to be considered. However, this
approach suffers four shortcomings with respect to ours: (a)
the results are developed for the square loss only, (b) 
building these sufficient statistics always require all peers to have
the classes, (c) the federated learning theory does not give a quantitative account of
the deviations to the ideal classifier that compares with ours and (d)
experimentally, the approach does not compare to the ideal
classifier, even when shared features are noise-free.

In all cases, our results are a very strong advocacy for federated
learning, and signals the existence of non-trivial tradeoffs for
entity-resolution to be optimized with the objective of learning from
linked data. We hope such results will contribute to spur related
research in the active and broad field of entity resolution, and
contribute in a broader agenda to technically shape data marketplaces.

\bibliographystyle{plain}
\bibliography{references,bibgen}

\clearpage

\section{Appendix: table of contents}\label{sec-toc}

\noindent Proof of Theorem \ref{lemmaTaylor}\hrulefill Pg \pageref{proof_lemmaTaylor}\\
\noindent Proof of Lemma \ref{lemmaTaylor2}\hrulefill Pg \pageref{proof_lemmaTaylor2}\\
\noindent Proof of Theorem \ref{thAPPROX1}\hrulefill Pg \pageref{proof_thAPPROX1}\\
\noindent Proof of Theorem \ref{thIMMUNE}\hrulefill Pg \pageref{app:proof-thIMMUNE}\\
\noindent Proof of Theorem \ref{thDIFFLOSS}\hrulefill Pg \pageref{app:proof-thDIFFLOSS}\\
\noindent Proof of Theorem \ref{thGENTHETA}\hrulefill Pg \pageref{app:proof-thGENTHETA}\\

\newpage

\subsection{Proof of Theorem \ref{lemmaTaylor}}\label{proof_lemmaTaylor}

We proceed in two steps, first assuming that $F$ is convex and then
relaxing the assumption.\\

\noindent Case 1 --- $F$ convex. In order not to laden our notations, we fold $\gamma$ and $\Gamma_F$ in
the regularizer and consider without loss of generality any convex Ridge regularized loss $\loss_F({\hat{S}}, \ve{\theta};
\Gamma_F) = L + R$ with $R \defeq
\ve{\theta}^\top \Gamma_F \ve{\theta} $ and
\begin{eqnarray}
L\defeq \frac{1}{m}
\cdot \sum_i F(y_i \ve{\theta}^\top \hat{\ve{x}}_i) \label{ridgeregLoss2}\:\:,
\end{eqnarray} 
for some convex twice differentiable $F$. We first focus on the
approximation of $L$ via a Taylor loss.
We perform a local Taylor-Lagrange expansion of each $F(y_i
\ve{\theta}^\top \hat{\ve{x}}_i)$ in eq. (\ref{ridgeregLoss2}) around
0 and
obtain that there exists $c_1, c_2, ..., c_m \in F''(\mathbb{I}(\Theta_*))
\subseteq \mathbb{R}_+$ such
that
\begin{eqnarray}
L & = & F(0) + \frac{F'(0)}{m}
\cdot \sum_i y_i \ve{\theta}^\top \hat{\ve{x}}_i + J \:\:,\label{taylor1appr}
\end{eqnarray}
where $\mathbb{I} \defeq [-\hat{X}_* \Theta_*, \hat{X}_* \Theta_*]$ 
(since $|y_i
\ve{\theta}^\top \hat{\ve{x}}_i| \leq \hat{X}_* \Theta_*$ by Cauchy-Schwartz
inequality)
 and $J \defeq (1/2m) \cdot \sum_i c_i (y_i \ve{\theta}^\top
\hat{\ve{x}}_i)^2$. Here, we have assumed that there exists some
$\Theta_* > 0$ such that $\|\ve{\theta}\|_2 \leq \Theta_*$; we shall see
that such a bound $\Theta_*$ indeed exists for the $\ve{\theta}$ which
interests us.
Let 
\begin{eqnarray}
c' & \defeq & \frac{\sum_i c_i (y_i \ve{\theta}^\top
\hat{\ve{x}}_i)^2}{\sum_i (y_i \ve{\theta}^\top
\hat{\ve{x}}_i)^2}\:\:.
\end{eqnarray} 
It trivially follows that $c' \in
F''(\mathbb{I})$ and 
\begin{eqnarray}
J & = & \frac{c'}{2m} \cdot \sum_i (y_i \ve{\theta}^\top
\hat{\ve{x}}_i)^2\:\:.
\end{eqnarray}
What we thus get is that for any $\forall \hat{S}, \ve{\theta}$, there
exists $c \in (1/2)\cdot F''(\mathbb{I}) \subseteq \mathbb{R}_{+}$ such that
\begin{eqnarray}
L & = & F(0) + \frac{F'(0)}{m}
\cdot \sum_i y_i \ve{\theta}^\top \hat{\ve{x}}_i + \frac{c}{m} \cdot \sum_i (y_i \ve{\theta}^\top
\hat{\ve{x}}_i)^2\:\:,\label{approx1}
\end{eqnarray}
and we also observe that $L$ is convex. We now consider the choice $c \defeq c^*$ obtained for 
\begin{eqnarray}
\ve{\theta}^* & \defeq & \arg\min_{\ve{\theta}} \loss_F({\hat{S}}, \ve{\theta};
\Gamma_F) \:\:.\label{eqMIN1}
\end{eqnarray}
Let us denote $\opttaylorloss$ the particular Taylor loss obtained,
which therefore matches $\loss_F({\hat{S}}, \ve{\theta};
\Gamma_F)$ for the choice $\ve{\theta} = \ve{\theta}^*$. We now
design the regularizer of the Taylor loss to ensure that its
\textit{optimum} is also achieved by $\ve{\theta}^*$. It is not hard
to check that the optimum of the Ridge regularized Taylor loss $\opttaylorloss({\hat{S}}, \ve{\theta};
\Gamma_T)$,
$\ve{\theta}^\circ$, satisfies:
\begin{eqnarray}
c^* \hat{\X}\hat{\X}^\top \ve{\theta}^\circ + 2m \Gamma_T
\ve{\theta}^\circ & = &
-F'(0) \ve{\mu}_{\hat{S}}\:\:, \label{approxTaylor1}
\end{eqnarray}
where $\ve{\mu}_{\hat{S}} \defeq \sum_i y_i \hat{\ve{x}}_i$ is the
mean operator \cite{pnrcAN}. Let us find the equivalent expression
for loss $\ell_F$ via a series of Taylor-Lagrange expansions, letting
$z_i \defeq y_i {\ve{\theta}^*}^\top \hat{\ve{x}}_i$ for short:
\begin{eqnarray}
\forall i\in [m], \exists c'_i \in F''(\mathbb{I}): F(z_i) & = & F(0) + F'(0) z_i + \frac{c'_i}{2} z^2_i\:\:.
\end{eqnarray}
Define $\ve{c} \in \mathbb{R}_+^m$ the vector with $c_i \defeq
c'_i/2$. It follows that because of eq. (\ref{eqMIN1}),
$\ve{\theta}^*$ satisfies $\sum_i c_i ( {\ve{\theta}^*}^\top
\hat{\ve{x}}_i) \hat{\ve{x}}_i + 2m \Gamma_F \ve{\theta}^* = -F'(0)
\ve{\mu}_{\hat{S}}$, or more concisely,
\begin{eqnarray}
\hat{\X}\mathrm{Diag}(\ve{c})\hat{\X}^\top \ve{\theta}^* + 2m \Gamma_F
\ve{\theta}^* & = &
-F'(0) \ve{\mu}_{\hat{S}}\:\:. \label{approxTaylor2}
\end{eqnarray}
Now, we want $\ve{\theta}^\circ=\ve{\theta}^*$, which imposes from
eqs (\ref{approxTaylor1}) and (\ref{approxTaylor2}), $c^* \hat{\X}\hat{\X}^\top \ve{\theta}^* + 2m \Gamma_T
\ve{\theta}^*  = \hat{\X}\mathrm{Diag}(\ve{c})\hat{\X}^\top \ve{\theta}^* + 2m \Gamma_F
\ve{\theta}^*$, or equivalently, after simplifying,
\begin{eqnarray}
\Gamma_T \ve{\theta}^* & = & \matrice{k} \ve{\theta}^* + \Gamma_F \ve{\theta}^*\:\:,\label{eqDELTA1}
\end{eqnarray}
where 
\begin{eqnarray}
\matrice{k} & \defeq & \hat{\X}\left(\frac{1}{2m}(\mathrm{Diag}(\ve{c} -
c^*\ve{1}))\right)\hat{\X}^\top 
\end{eqnarray} 
is symmetric but not necessarily positive
definite. We clearly have $\ve{\theta}^\top \matrice{k} \ve{\theta}
\geq - \hat{X}_*^2 \sup F''(\mathbb{I})  / 2$ for any unit $\ve{\theta}$. So, if we
fix 
\begin{eqnarray}
\Gamma_T & \defeq & \matrice{k} + \Gamma_F
\end{eqnarray}
after picking $\Gamma_F$ such that its smallest eigenvalue satisfies,
for some fixed $\lambda^\circ > 0$,
\begin{eqnarray}
\lambda_1^\uparrow(\Gamma_F) & \geq & \lambda^\circ + \frac{\hat{X}^2_*}{2} \sup
F''(\mathbb{I}) \:\:,\label{blambdastar}
\end{eqnarray} 
then we
shall have eq. (\ref{eqDELTA1}) ensured with $\Gamma_T$ symmetric
positive definite with $\lambda_1^\uparrow(\Gamma_T) \geq \lambda^\circ$. We can also remark that eq. (\ref{approxTaylor2})
yields, because $\hat{\X}\mathrm{Diag}(\ve{c})\hat{\X}^\top$ is
positive semi-definite,
\begin{eqnarray}
\|\ve{\theta}^*\|_2 & \leq & \frac{|F'(0)|\hat{X}_*}{2\lambda_1^{\uparrow}(\Gamma_F)}\:\:,\label{bsupnorm}
\end{eqnarray}
so we can posit $\Theta_* \defeq
|F'(0)|\hat{X}_*/(2\lambda_1^{\uparrow}(\Gamma_F))$ and in fact we can pick
\begin{eqnarray}
\mathbb{I} & \defeq & \frac{|F'(0)|\hat{X}^2_*}{2\lambda_1^{\uparrow}(\Gamma_F)}
\cdot \left[-1,1\right]\:\:.\label{defII}
\end{eqnarray}
For any finite
$\lambda^\circ, \hat{X}_*$, let us define
\begin{eqnarray}
\mathbb{J}(\lambda^\circ, \hat{X}_*) & \defeq & \left\{z \in \mathbb{R} :
  z \geq \lambda^\circ + \frac{\hat{X}^2_*}{2} \sup F''\left(\lim_{z'
      \rightarrow z} \frac{|F'(0)|\hat{X}^2_*}{2z'}
\cdot \left[-1,1\right]\right)\right\} \cap \mathbb{R}_{+} \:\:.\label{defJJ}
\end{eqnarray}
Picking $\lambda_1^\uparrow(\Gamma_F)$ in $\mathbb{J}(\lambda^\circ,
\hat{X}_*)$ guarantees that it satisfies eq. (\ref{blambdastar}). Let
ue denote for short $\mathbb{J}'$ to be the leftmost set in the
intersection in eq. (\ref{defJJ}). Because the argument of $F''$ is
the same for any $\pm z$, if there exists any $z<0$ in $\mathbb{J}'$,
then $-z$ is also in $\mathbb{J}'$. We 
remark that because $|F'(0)|\ll \infty$, the argument set of $F''(.)$ converges to $\{0\}$ with
$z \rightarrow \pm \infty$; since $F''$ is continuous and
$|F''(0)| = F''(0) \ll \infty$ by assumption, we get that $\mathbb{J}'$ is
non-empty, and so $\mathbb{J}' \cap \mathbb{R}_+$ is non-empty, thus
\begin{eqnarray}
\mathbb{J}(\lambda^\circ, \hat{X}_*) & \neq & \emptyset\:\:.\label{propJEMPTY}
\end{eqnarray}
So, let us define
\begin{eqnarray}
\lambda^* & \defeq & \inf \mathbb{J}(\lambda^\circ, \hat{X}_*) \:\: (\geq 0)\:\:,\label{defLAMBDASTAR}
\end{eqnarray}
removing the dependence of $\lambda^*$ in 
$\lambda^\circ, \hat{X}_*$ for clarity. To summarize, for any $\lambda^\circ > 0$ and any Ridge regularized loss $\loss_F({\hat{S}}, \ve{\theta};
\Gamma_F)$ satisfying $F \in C^2$, $|F'(0)|, F''(0) \ll \infty$ and $\lambda_1^\uparrow(\Gamma_F) \geq \lambda^*$ where
$\lambda^*$ is finite and defined in eq. (\ref{defLAMBDASTAR}), there exists a Taylor loss $\opttaylorloss({\hat{S}}, \ve{\theta};
\Gamma_T)$ such that
\begin{enumerate}
\item $\loss_F({\hat{S}}, \ve{\theta}^*;
\Gamma_F) = \opttaylorloss({\hat{S}}, \ve{\theta}^*;
\Gamma_T)$  where $\ve{\theta}^* \defeq \arg\min_{\ve{\theta}} \loss_F({\hat{S}}, \ve{\theta};
\Gamma_F)$;
\item $\arg\min_{\ve{\theta}} \loss_F({\hat{S}}, \ve{\theta};
\Gamma_F) = \arg\min_{\ve{\theta}} \opttaylorloss({\hat{S}}, \ve{\theta};
\Gamma_T)$;
\item $\lambda_1^\uparrow(\Gamma_T) \geq \lambda^\circ$.
\end{enumerate}
We also check that $a = F(0), b = F'(0)$, and we get the statement of the Theorem when $F$ is convex.\\

\noindent Case 2 --- $F$ not convex. When $F$ is not convex, 
we still have for any $\ve{\theta}^* \in \mathcal{C}$ because $F$ is twice differentiable,
\begin{eqnarray}
\loss_F({\hat{S}}, \ve{\theta}^*;
\Gamma_F) & = & \loss_F({\hat{S}}, \ve{0};
\Gamma_F) + {\ve{\theta}^*}^\top \nabla_{\ve{\theta}} \loss_F({\hat{S}}, \ve{\theta};
\Gamma_F)_{|\ve{\theta} = \ve{0}} + \frac{1}{2} \cdot
{\ve{\theta}^*}^\top \nabla \nabla_{\ve{\theta}} \loss_F({\hat{S}}, \ve{\theta};
\Gamma_F)_{|\ve{\theta} = \ve{u}} {\ve{\theta}^*}\:\:,
\end{eqnarray}
for some $\ve{u} = t \cdot \ve{\theta}^*$ with $t\in [0,1]$, where
$\nabla\nabla$ denote the Hessian, given by
\begin{eqnarray}
\nabla \nabla_{\ve{\theta}} \loss_F({\hat{S}}, \ve{\theta};
\Gamma_F)_{|\ve{\theta} = \ve{u}} & = & \sum_i F''(y_i \ve{u}^\top \hat{\ve{x}}_i) \cdot
\hat{\ve{x}}_i \hat{\ve{x}}^{\top}_i
 + 2 \Gamma_F\:\:,
\end{eqnarray}
positive semi-definite since $\ve{\theta}^* \in \mathcal{C}$. $F$
being $C^2$, $F''$ being continuous, $\ve{\theta}^*$ is a local
minimum of the loss in an open neighborhood $\mathcal{N}(\ve{\theta}^*)$
of $\ve{\theta}^*$. We still can build the equivalent Taylor loss and
first its $L$ part as in eq. (\ref{taylor1appr}). However, $L$ is
\textit{not} necessarily convex this time. The Hessian of the Taylor
loss \textit{regularized} is now
\begin{eqnarray}
\nabla \nabla_{\ve{\theta}} \opttaylorloss({\hat{S}}, \ve{\theta};
\Gamma_T) & = & c^* \sum_i \hat{\ve{x}}_i \hat{\ve{x}}^{\top}_i
 + 2 \Gamma_T\:\:,
\end{eqnarray}
and so to obtain a convex regularized Taylor loss, it is sufficient to
ensure,
for some fixed $\lambda^\circ > 0$,
\begin{eqnarray}
2\lambda_1^\uparrow(\Gamma_T) & \geq & \lambda^\circ + \hat{X}^2_* \sup
F''(\mathbb{I}) \:\:,\label{blambdastar2}
\end{eqnarray} 
which is exactly ineq. (\ref{blambdastar}) with its argument
$\lambda^\circ$ halved. So, the regularized Taylor loss is in fact
convex, and the only other modification is to now ensure $|F''(0)|\ll
\infty$ since $F$ can be concave in $0$.\\

\noindent \textbf{Remark}: in all that follows, we assume without loss of generality that the
mean operator $\ve{\mu}_{\hat{S}} \neq \ve{0}$, which implies, from
eqs (\ref{approxTaylor1}) and (\ref{approxTaylor2}) that $\ve{0}$
cannot be a critical point of the losses.

\subsection{Proof of Lemma \ref{lemmaTaylor2}}\label{proof_lemmaTaylor2}

Since $\psi$ is strictly convex differentiable, its convex conjugate
is $\psi^\star(z) = z \psi'^{-1}(z) - \psi(\psi'^{-1}(z))$, from which
we easily get $F''_{\psi}(z) = 1 / (b_\psi
\psi''(\psi'^{-1}(-z)))$. Because $\psi'$ is concave on $[0,1/2]$,
$\psi''$ is decreasing on $[0,1/2]$ and therefore increasing on
$[1/2, 1]$, achieving its minimum for $\psi'^{-1}(-z) = 1/2$, which
gives $-z = \psi'(1/2) = 0$ and $z=0$ for the arg max of
$F''_\psi(z)$. Hence, $\mathbb{J}(\lambda^\circ, \hat{X}_*)$
becomes more explicit:
\begin{eqnarray}
\mathbb{J}(\lambda^\circ, \hat{X}_*) & \defeq & \left\{z \in \mathbb{R}_+ :
  z \geq \lambda^\circ + \frac{F_\psi''(0) \hat{X}^2_*}{2}\right\}\:\:,\label{defJJ2}
\end{eqnarray}
so we can just pick
\begin{eqnarray}
\lambda^* & \defeq & \lambda^\circ + \frac{F_{\psi}''(0) \hat{X}^2_*}{2} \:\:,
\end{eqnarray}
as claimed.

\subsection{Proof of Theorem \ref{thAPPROX1}}\label{proof_thAPPROX1}

The proof is obtained in three steps: we first define additional
assumptions useful for the proof, then prove a helper Theorem of
independent interest, and finally prove Theorem \ref{thAPPROX1}. We
remind that the Taylor loss we are concerned with (main file, Section
\ref{sec:taylor}) is
\begin{eqnarray}
\taylorlossparams{F(0)}{F'(0)}{c} ({\hat{S}}_t, \ve{\theta}; \gamma,
\Gamma)\:\:,\label{mainLoss}
\end{eqnarray}
with $c\neq 0$ and $t = 0, 1, ..., T$.

\subsubsection{Related notations and additional properties}

\begin{definition}\label{defmeano}
The mean operator associated to $\hat{S}_t$ is $\mathbb{R}^d \ni
\ve{\mu}_t \defeq \sum_i y_i \cdot \hat{\ve{x}}_{ti}$.
\end{definition}
The mean operator is a sufficient statistics for the class in linear
models \cite{pnrcAN}. We can make at this point a remark that is
going to be crucial in our results, and obvious from its definition:
the mean operator is \textit{invariant} to permutations made within
classes, \textit{i.e.} $\ve{\mu}_T = \ve{\mu}_0$ if $\PERM_{*}$ factorizes as two permutations, one affecting the positive
class only, and the other one affecting the negative class only. Since
the optimal classifier for the Taylor loss is a linear mapping of the
mean operator (Lemma \ref{lem11} below), our bounds will appear
significantly better when $\PERM_{*}$
factorizes in such a convenient way.

We now show an additional property of our notations in (main file,
Section \ref{sec:defs}).
\begin{lemma}\label{lemUAUB}
The following holds for any $t \geq 1$:
\begin{eqnarray}
(\hat{\ve{x}}_{t\ua{t}})_\shuffle & = &
(\ve{x}_{\ub{t}})_\shuffle\:\:,\label{eqm1EX1}\\
(\hat{\ve{x}}_{t\va{t}})_\shuffle & = & (\ve{x}_{\vb{t}})_\shuffle\:\:. \label{eqm1EX2}
\end{eqnarray}
\end{lemma}
\begin{example}\label{exampleEX1}
Denote for short $\{0,1\}^{m\times m} \ni \Theta_{u,v} \defeq \ve{1}_{u}\ve{1}^\top_{v} +
\ve{1}_{v}\ve{1}^\top_{u} - \ve{1}_{v}\ve{1}^\top_{v} -
\ve{1}_{u}\ve{1}^\top_{u}$ (symmetric) such that $\ve{1}_u$ is the $u^{th}$
canonical basis vector of $\mathbb{R}^n$. For $t=1$, it follows
\begin{eqnarray}
\ub{1} & = & \va{1}\:\:,\label{eq0EX1}\\
\vb{1} & = & \ua{1}\:\:.\label{eq0EX2}
\end{eqnarray} 
Thus, it follows:
\begin{eqnarray}
\lefteqn{\X_\anchor \Theta_{\ua{1},\va{1}}
  \hat{\X}^\top_{1\shuffle}}\nonumber\\
 & = &
(\ve{x}_{\ua{1}})_\anchor (\ve{x}_{1\va{1}})^\top_\shuffle +
(\ve{x}_{\va{1}})_\anchor (\ve{x}_{1\ua{1}})^\top_\shuffle -
(\ve{x}_{\va{1}})_\anchor (\ve{x}_{1\va{1}})^\top_\shuffle -
(\ve{x}_{\ua{1}})_\anchor (\ve{x}_{1\ua{1}})^\top_\shuffle\nonumber\\
 & = & (\ve{x}_{\ua{1}})_\anchor (\ve{x}_{\vb{1}})^\top_\shuffle +
(\ve{x}_{\va{1}})_\anchor (\ve{x}_{\ub{1}})^\top_\shuffle -
(\ve{x}_{\va{1}})_\anchor (\ve{x}_{\vb{1}})^\top_\shuffle -
(\ve{x}_{\ua{1}})_\anchor
(\ve{x}_{\ub{1}})^\top_\shuffle \label{eq1EX1}\\
 & = & (\ve{x}_{\ua{1}})_\anchor (\ve{x}_{\ua{1}})^\top_\shuffle +
(\ve{x}_{\va{1}})_\anchor (\ve{x}_{\va{1}})^\top_\shuffle -
(\ve{x}_{\va{1}})_\anchor (\ve{x}_{\ua{1}})^\top_\shuffle -
(\ve{x}_{\ua{1}})_\anchor
(\ve{x}_{\va{1}})^\top_\shuffle\label{eq1EX2}\\
 & = & (\ve{x}_{\ua{1}}-\ve{x}_{\va{1}})_\anchor(\ve{x}_{\ua{1}}-\ve{x}_{\va{1}})_\shuffle^\top\:\:.
\end{eqnarray}
In eq. (\ref{eq1EX1}), we have used eqs (\ref{eqm1EX1}, \ref{eqm1EX2})
and in eq. (\ref{eq1EX2}), we have used eqs (\ref{eq0EX1}, \ref{eq0EX2}).
\end{example}
\noindent \textbf{Key matrices} --- The proof of our helper Theorem is relatively heavy in linear
algebra notations: for example, it involves $T$ double applications
of Sherman-Morrison's inversion Lemma. We now define a series of
matrices and vectors that will be most useful to simplify notations
and proofs. 
Letting $\nu' \defeq 2m\gamma / c$ (where $\gamma$ is the parameter of the
Ridge regularization in our Taylor loss 
and $c \neq 0$ is defined in eq. (\ref{mainLoss})), we first define the matrix we will
use most often:
\begin{eqnarray}
\matrice{v}_t &\defeq & \left( \mathrm{sign}(c) \cdot \hat{\X}_t \hat{\X}_t^\top + \nu'\cdot \Gamma
  \right)^{-1} \:\:, t = 0, 1, ..., T\:\:,\label{defVT}
\end{eqnarray}
where $\Gamma$ is the 
Ridge regularization parameter matrix in eq. (\ref{mainLoss}).
Another matrix $\matrice{u}_t$, quantifies precisely the local
mistake made by each elementary permutation. To define it, we first
let (for $t = 1, 2, ..., T$):
\begin{eqnarray}
\ve{a}_t & \defeq & (\ve{x}_{\ua{t}}-\ve{x}_{\va{t}})_\anchor\:\:,\label{defAT}\\
\ve{b}_t & \defeq &
(\ve{x}_{\ub{t}}-\ve{x}_{\vb{t}})_\shuffle\:\:.\label{defST}
\end{eqnarray}
Also, let (for $t = 1, 2, ..., T$)
\begin{eqnarray}
\ve{a}^+_t & \defeq& \left[
\begin{array}{c}
(\ve{x}_{\ua{t}}-\ve{x}_{\va{t}})_\anchor\\\cline{1-1}
\ve{0}
\end{array}
\right] \in \mathbb{R}^d \:\:,\label{defAPLUST}\\
\ve{b}^+_t & \defeq &  \left[
\begin{array}{c}
\ve{0} \\\cline{1-1}
(\ve{x}_{\ub{t}} -\ve{x}_{\vb{t}})_\shuffle
\end{array}
\right] \in \mathbb{R}^d \:\:, \label{defBPLUST}
\end{eqnarray}
and finally (for $t = 1, 2, ..., T$),
\begin{eqnarray}
c_{0,t}
& \defeq & {\ve{a}^+_t}^\top  \matrice{v}_{t-1}
    \ve{a}^+_t \:\:,\\
c_{1,t} & \defeq & {\ve{a}^+_t}^\top \matrice{v}_{t-1}
    \ve{b}^+_t \:\:,\\
c_{2,t} & \defeq & {\ve{b}^+_t}^\top  \matrice{v}_{t-1}
    \ve{b}^+_t \:\:.
\end{eqnarray}
We now define $\matrice{u}_t$ as the following block matrix for $t = 1, 2, ..., T$:
\begin{eqnarray}
\matrice{u}_t & \defeq & \frac{1}{(1-\mathrm{sign}(c)\cdot c_{1,t})^2-c_{0,t}c_{2,t}}\cdot \left[\begin{array}{c|c}
c_{2,t} \cdot \ve{a}_t \ve{a}_t^\top & (1-\mathrm{sign}(c)\cdot c_{1,t}) \cdot \ve{a}_t
\ve{b}_t^\top \\ \cline{1-2}
(1-\mathrm{sign}(c)\cdot c_{1,t}) \cdot \ve{b}_t
\ve{a}_t^\top & c_{0,t} \cdot \ve{b}_t\ve{b}_t^\top 
\end{array}\right]\:\:.\label{defUT}
\end{eqnarray}
$\matrice{u}_t$ can be computed only when
$(1-\mathrm{sign}(c)\cdot c_{1,t})^2 \neq c_{0,t}c_{2,t}$. This shall be the subject of the
\textit{invertibility} assumption below.
Hereafter, we suppose without loss of generality that $\ve{b}_t \neq
\ve{0}$, since otherwise permutations would make no mistakes on the
shuffle part.

There is one important thing to remark on $\matrice{u}_t$: it is
defined from the indices $\ua{t}$ and $\va{t}$ in $\anchor$ that are affected by
$\PERM_{t}$. Hence, $\matrice{u}_1$ collects the two first such
indices (see Figure \ref{fig:perm}).
We also define matrix $\Lambda_t$ as follows:
\begin{eqnarray}
\Lambda_t & \defeq & -\frac{F'(0)}{|c|} \cdot \matrice{v}_{t} \matrice{u}_{t+1} \:\:, t = 0,
1, ..., T-1\:\:,\label{defL1}
\end{eqnarray}
where parameters $c, F'(0)$ are those defined in the Taylor loss in
eq. (\ref{mainLoss}). To finish up with matrices, we define a doubly indexed matrices that
shall be crucial to our proofs, $\matrice{h}_{i,j}$ for $0\leq
j\leq i \leq T$:
\begin{eqnarray}
\matrice{h}_{i,j} & \defeq & \left\{
\begin{array}{ccl}
\prod_{k=j}^{i-1}(\matrice{i}_d + \Lambda_k) & \mbox{ if }
& 0\leq j < i\\
\matrice{i}_d & \mbox{ if }
& j = i
\end{array}
\right.\:\:.\label{defHIJ}
\end{eqnarray}

\noindent \textbf{Key vectors} --- we let
\begin{eqnarray}
\ve{\epsilon}_{t} & \defeq & \ve{\mu}_{t+1} -
  \ve{\mu}_{t} \:\:, t = 0, 1, ..., T-1\:\:,\label{defEPSILONT}
\end{eqnarray}
which is the difference between two successive mean operators, and
\begin{eqnarray}
\ve{\lambda}_t & \defeq & -\frac{F'(0)}{|c|}\cdot \matrice{v}_{t+1} \ve{\epsilon}_{t}\:\:, t
= 0, 1, ..., T-1\:\:.\label{defL2}
\end{eqnarray}
\begin{figure}[t]
\centering
\includegraphics[trim=30bp 540bp 630bp
30bp,clip,width=.60\linewidth]{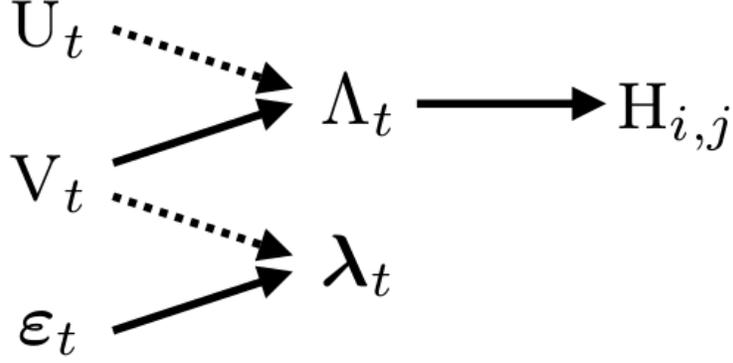} 
\caption{Summary of our key notations on matrices and vectors, and
  dependencies. The dashed arrow means indexes do not match (eq. (\ref{defL2})).\label{fig:notations}}
\end{figure}
Figure \ref{fig:notations} summarizes our key notations in this
Section. We are now ready to proceed through the proof of our key helper Theorem.

\subsubsection{Helper Theorem}

In this Section, we first show (Theorem \ref{thmeq1} below) that under
lightweight assumptions to ensure the existence of $\matrice{v}_t$,
the difference between two successive optimal classifiers in the
progressive computation of the overall permutation matrix that
generates the errors is \textit{exactly} given by:
\begin{eqnarray}
\ve{\theta}^*_{t+1} - \ve{\theta}^*_t & = & \nu\cdot \matrice{v}_t
\matrice{u}_{t+1} \ve{\theta}^*_t + \nu\cdot \matrice{v}_{t+1}
\ve{\epsilon}_{t}\nonumber\\
 & = & \Lambda_t \ve{\theta}^*_t
+ \ve{\lambda}_t\:\:, \forall t\geq 0\:\:,\label{sumeq1}
\end{eqnarray}
where $\Lambda_t, \ve{\epsilon}_{t}, \ve{\lambda}_t$ are defined in
eqs (\ref{defEPSILONT}, \ref{defL1},
\ref{defL2}) and $\nu$ is defined in Lemma \ref{lem11} below. This holds regardless of the permutation matrices in the
sequence.

We start by the trivial solutions to the minimization of a convex Taylor loss $\taylorlossparams{F(0)}{F'(0)}{c} ({\hat{S}}_t, \ve{\theta}; \gamma,
\Gamma)$ for
all $t = 1, 2, ..., T$.
\begin{lemma}\label{lem11}
The minimum of any convex Taylor loss $\taylorlossparams{F(0)}{F'(0)}{c} ({\hat{S}}_t, \ve{\theta}; \gamma,
\Gamma)$, for $c\in \mathbb{R}_*$, is
\begin{eqnarray}
\ve{\theta}^*_t & = & \nu \cdot \left( \mathrm{sign}(c) \cdot \hat{\X}_t \hat{\X}_t^\top +
  \nu' \cdot \Gamma
  \right)^{-1} \ve{\mu}(\hat{S})_t \nonumber\\
 & = & \nu \cdot \matrice{v}_t \ve{\mu}_t\:\:,
\end{eqnarray}
with 
\begin{eqnarray}
\nu & \defeq & -\frac{F'(0)}{|c|}\:\:,\nonumber\\
\nu' & \defeq & \frac{2m\gamma}{|c|}\:\:,
\end{eqnarray}
meeting $\nu, \nu' \neq 0$.
\end{lemma}
\begin{proof}
We reformulate eq. (\ref{approxTaylor1}): $\ve{\theta}^*_t$ satisfies:
\begin{eqnarray}
c \hat{\X}_t\hat{\X}_t^\top \ve{\theta}_t^* + 2m \gamma \cdot \Gamma
\ve{\theta}_t^* & = &
-F'(0) \cdot \ve{\mu}_t\:\:, \label{approxTaylor1b}
\end{eqnarray}
giving $\ve{\theta}_t^* = -F'(0) \cdot (c \hat{\X}_t\hat{\X}_t^\top + 2m \gamma \cdot
\Gamma)^{-1} \ve{\mu}_t = \nu \cdot \left( \mathrm{sign}(c) \cdot \hat{\X}_t \hat{\X}_t^\top +
  \nu' \cdot \Gamma
  \right)^{-1} \ve{\mu}(\hat{S})_t$, as claimed.
\end{proof}
\begin{lemma}\label{lemCONDVT}
Suppose $\matrice{v}_{t-1}$ exists. Then $\matrice{v}_{t}$ exists if
the following holds:
\begin{eqnarray}
\left\{
\begin{array}{rcl}
\mathrm{sign}(c) \cdot c_{1,t} & \neq & 1 \:\:,\label{cond1}\\
(1-\mathrm{sign}(c)  \cdot c_{1,t})^2 & \neq & c_{0,t}c_{2,t}\:\:.\label{cond2}
\end{array}
\right.
\end{eqnarray}
\end{lemma}
\begin{proof}
Throughout the proof, we let
\begin{eqnarray}
\varsigma & \defeq & \mathrm{sign}(c) 
\end{eqnarray}
for short.
We know that $\hat{\X}_{t}$ is obtained from $\hat{\X}_{t-1}$
  after permuting the shuffle part of observations at indexes $\ua{t}$ and $\va{t}$ in
  $\hat{\X}_{(t-1) {\shuffle}}$ by $\PERM_{t}$ (see Figure \ref{fig:perm}). So,
\begin{eqnarray}
\hat{\X}_{t{\shuffle}} & = & \hat{\X}_{(t-1){\shuffle}} + \hat{\X}_{(t-1){\shuffle}} (\PERM_t -
\matrice{i}_n)\nonumber\\
 & = & \hat{\X}_{(t-1){\shuffle}} + \hat{\X}_{(t-1){\shuffle}}
 (\ve{1}_{\ua{t}}\ve{1}^\top_{\va{t}} + \ve{1}_{\va{t}}\ve{1}^\top_{\ua{t}} - \ve{1}_{\va{t}}\ve{1}^\top_{\va{t}} - \ve{1}_{\ua{t}}\ve{1}^\top_{\ua{t}})\:\:,\label{eqth1PR1}
\end{eqnarray}
where $\ve{1}_u \in \mathbb{R}^n$ is the $u^{th}$ canonical basis
vector. We also have 
\begin{eqnarray}
\hat{\X}_{t}\hat{\X}_{t}^\top & = & \left[
\begin{array}{c|c}
\X_\anchor \X_\anchor^\top  & \X_\anchor \hat{\X}_{t\shuffle}^\top \\\cline{1-2}
\hat{\X}_{t\shuffle}\X_\anchor^\top & \hat{\X}_{t\shuffle} \hat{\X}_{t\shuffle}^\top
\end{array}
\right]\nonumber\\
 & = & \left[
\begin{array}{c|c}
\X_\anchor \X_\anchor^\top  & \X_\anchor \hat{\X}_{t\shuffle}^\top \\\cline{1-2}
\hat{\X}_{t\shuffle}\X_\anchor^\top &\hat{\X}_{(t-1)\shuffle} \PERM_{t} \PERM^\top_{t} \hat{\X}_{(t-1)\shuffle}^\top
\end{array}
\right]\nonumber\\
 & = & \left[
\begin{array}{c|c}
\X_\anchor \X_\anchor^\top  & \X_\anchor \hat{\X}_{t\shuffle}^\top \\\cline{1-2}
\hat{\X}_{t\shuffle}\X_\anchor^\top &\hat{\X}_{(t-1)\shuffle} \hat{\X}_{(t-1)\shuffle}^\top
\end{array}
\right]\label{eqth1PR2}\:\:,
\end{eqnarray}
because the
inverse of a permutation matrix is its transpose. We recall that
$\X_\anchor$ does not change throughout permutations, only
$\X_\shuffle$ does. Hence,
\begin{eqnarray}
\hat{\X}_{t}\hat{\X}_{t}^\top & = & \hat{\X}_{t-1}\hat{\X}_{t-1}^\top + \left[
\begin{array}{c|c}
0  & \X_\anchor (\hat{\X}_{t\shuffle}-\X _{(t-1)\shuffle})^\top \\\cline{1-2}
(\hat{\X}_{t\shuffle}-\X _{(t-1)\shuffle}) \X_\anchor^\top & 0
\end{array}
\right]\nonumber\\
 & = & \hat{\X}_{t-1}\hat{\X}_{t-1}^\top + \left[
\begin{array}{c|c}
0  & \X_\anchor \Theta_{\ua{t},\va{t}} \hat{\X}^\top_{(t-1)\shuffle}\\\cline{1-2}
\hat{\X}_{(t-1)\shuffle} \Theta_{\ua{t},\va{t}} \X_\anchor^\top & 0
\end{array}
\right]\:\:,
\end{eqnarray}
with $\Theta_{\ua{t},\va{t}} \defeq
\ve{1}_{\ua{t}}\ve{1}^\top_{\va{t}} +
\ve{1}_{\va{t}}\ve{1}^\top_{\ua{t}} -
\ve{1}_{\va{t}}\ve{1}^\top_{\va{t}} -
\ve{1}_{\ua{t}}\ve{1}^\top_{\ua{t}}$ (symmetric, see
eq. (\ref{eqth1PR1}) and example \ref{exampleEX1}). 
Now, remark that
\begin{eqnarray}
\lefteqn{\X_\anchor \Theta_{\ua{t},\va{t}}
  \hat{\X}^\top_{(t-1)\shuffle}}\nonumber\\
 & = & \X_\anchor (\ve{1}_{\ua{t}}\ve{1}^\top_{\va{t}} +
\ve{1}_{\va{t}}\ve{1}^\top_{\ua{t}} -
\ve{1}_{\va{t}}\ve{1}^\top_{\va{t}} -
\ve{1}_{\ua{t}}\ve{1}^\top_{\ua{t}}) \hat{\X}^\top_{t\shuffle} \nonumber\\
 & = & (\ve{x}_{\ua{t}})_\anchor (\ve{x}_{t\va{t}})_\shuffle^\top +  (\ve{x}_{\va{t}})_\anchor
 (\ve{x}_{t\ua{t}})_\shuffle^\top -  (\ve{x}_{\va{t}})_\anchor
 (\ve{x}_{t\va{t}})_\shuffle^\top -  (\ve{x}_{\ua{t}})_\anchor
 (\ve{x}_{t\ua{t}})_\shuffle^\top\nonumber\\
 & = & (\ve{x}_{\ua{t}})_\anchor (\ve{x}_{\vb{t}})_\shuffle^\top +  (\ve{x}_{\va{t}})_\anchor
 (\ve{x}_{\ub{t}})_\shuffle^\top -  (\ve{x}_{\va{t}})_\anchor
 (\ve{x}_{\vb{t}})_\shuffle^\top -  (\ve{x}_{\ua{t}})_\anchor
 (\ve{x}_{\ub{t}})_\shuffle^\top\label{eqSIMPL1}\\
 & =& -((\ve{x}_{\ua{t}})_\anchor-(\ve{x}_{\va{t}})_\anchor)((\ve{x}_{\ub{t}})_\shuffle-(\ve{x}_{\vb{t}})_\shuffle)^\top\nonumber\\
 & =&
 -(\ve{x}_{\ua{t}}-\ve{x}_{\va{t}})_\anchor(\ve{x}_{\ub{t}}-\ve{x}_{\vb{t}})_\shuffle^\top
 = -\ve{a}_{t} \ve{b}_{t}^\top\:\:.
\end{eqnarray}
Eq. (\ref{eqSIMPL1}) holds because of Lemma \ref{lemUAUB}. We finally get 
\begin{eqnarray}
\hat{\X}_{t}\hat{\X}_{t}^\top & = & \hat{\X}_{t-1}\hat{\X}_{t-1}^\top - \varsigma \cdot \ve{a}^+_t
{\ve{b}^+_t}^\top - \varsigma \cdot  {\ve{b}^+_t}
{{\ve{a}^+_t}}^\top\:\:,\label{eqXTXTM1}
\end{eqnarray}
and so we have
\begin{eqnarray}
\matrice{v}_t & = & \left(\matrice{v}^{-1}_{t-1} - \varsigma \cdot \ve{a}^+_t
{\ve{b}^+_t}^\top - \varsigma \cdot {\ve{b}^+_t}
{{\ve{a}^+_t}}^\top\right)^{-1}\label{eqDEFVT}\:\:.
\end{eqnarray}
We analyze when $\matrice{v}_t$ can be computed. First notice that
assuming $\matrice{v}_{t-1}$ exists implies its inverse also exists, and so
\begin{eqnarray}
\mathrm{det}(\matrice{v}^{-1}_{t-1} - \varsigma \cdot \ve{a}^+_t
{\ve{b}^+_t}^\top) & = & \mathrm{det}(\matrice{v}^{-1}_{t-1})\mathrm{det}(\matrice{i}_d - \varsigma \cdot \matrice{v}_{t-1}\ve{a}^+_t
{\ve{b}^+_t}^\top) \nonumber\\
 & = &  \mathrm{det}(\matrice{v}^{-1}_{t-1}) (1 -
 \varsigma \cdot {\ve{b}^+_t}^\top\matrice{v}_{t-1}\ve{a}^+_t) \nonumber\\
 & = &
 \mathrm{det}(\matrice{v}^{-1}_{t-1}) (1-\varsigma \cdot  c_{1,t})\:\:,\label{eqSYL0}
\end{eqnarray}
where the middle identity comes from Sylvester's determinant
formula. So, if in addition $1-\varsigma \cdot  c_{1,t}\neq 0$, then
\begin{eqnarray}
\lefteqn{\mathrm{det}\left(\matrice{v}^{-1}_{t-1} - \varsigma \cdot \ve{a}^+_t
{\ve{b}^+_t}^\top - \varsigma \cdot {\ve{b}^+_t}
{{\ve{a}^+_t}}^\top\right)}\nonumber\\
 & = & \mathrm{det}(\matrice{v}^{-1}_{t-1} - \varsigma \cdot \ve{a}^+_t
{\ve{b}^+_t}^\top) \mathrm{det}\left(\matrice{i}_d - \varsigma \cdot \left(\matrice{v}^{-1}_{t-1} - \varsigma \cdot \ve{a}^+_t
{\ve{b}^+_t}^\top\right) {\ve{b}^+_t}
{{\ve{a}^+_t}}^\top\right) \nonumber\\
 &= & \mathrm{det}(\matrice{v}^{-1}_{t-1}) (1-\varsigma \cdot c_{1,t}) \mathrm{det}\left(\matrice{i}_d - \varsigma \cdot \left(\matrice{v}^{-1}_{t-1} - \varsigma \cdot \ve{a}^+_t
{\ve{b}^+_t}^\top\right)^{-1} {\ve{b}^+_t}
{{\ve{a}^+_t}}^\top\right) \label{eqSYL2}\\
 & = & \mathrm{det}(\matrice{v}^{-1}_{t-1}) (1-\varsigma \cdot c_{1,t}) \left(1 - \varsigma \cdot {{\ve{a}^+_t}}^\top\left(\matrice{v}^{-1}_{t-1} - \varsigma \cdot \ve{a}^+_t
{\ve{b}^+_t}^\top\right)^{-1} {\ve{b}^+_t}\right)  \label{eqSYL3}\\
 & = & \mathrm{det}(\matrice{v}^{-1}_{t-1}) (1-\varsigma \cdot c_{1,t}) \left(1 -\varsigma \cdot 
   {{\ve{a}^+_t}}^\top\left(\matrice{v}_{t-1} + \frac{\varsigma}{1 -\varsigma \cdot 
       {\ve{b}^+_t}^\top \matrice{v}_{t-1}\ve{a}^+_t}\cdot
     \matrice{v}_{t-1}\ve{a}^+_t{\ve{b}^+_t}^\top\matrice{v}_{t-1}\right)
   {\ve{b}^+_t}\right) \label{eqSYL4}\\
 & = & \mathrm{det}(\matrice{v}^{-1}_{t-1}) (1-\varsigma \cdot c_{1,t}) \left(1 -\varsigma \cdot 
   c_{1,t} - \varsigma^2 \cdot \frac{c_{0,t} c_{2,t}}{1 -\varsigma \cdot 
       c_{1,t}}\right) \nonumber\\
 & = & \frac{ (1 -\varsigma \cdot 
   c_{1,t})^2 - c_{0,t} c_{2,t}}{\mathrm{det}(\matrice{v}_{t-1}) }\:\:.
\end{eqnarray}
Here, eq. (\ref{eqSYL2}) comes from
eq. (\ref{eqSYL0}). Eq. (\ref{eqSYL3}) is another application of
Sylvester's determinant formula. Eq. (\ref{eqSYL3}) is
Sherman-Morrison formula and the last equation uses the fact that
$\varsigma^2 = 1$. We immediately conclude on Lemma \ref{lemCONDVT}.
\end{proof}
If we now assume without loss of generality that $\matrice{v}_0$
exists --- which boils down to taking $\gamma >0, \Gamma \succ 0$ ---,
then we get the existence of the complete sequence of matrices
$\matrice{v}_{t}$ (and thus the existence of the sequence of optimal classifiers
$\ve{\theta}^*_0, \ve{\theta}^*_1, ...$) provided the following \textbf{invertibility}
condition is satisfied.
\begin{mdframed}[style=MyFrame]
(\textbf{invertibility}) For any $t\geq 1$, $(1-\mathrm{sign}(c) \cdot c_{1,t})^2 \not\in \{0, c_{0,t}c_{2,t}\}$.
\end{mdframed}
We shall check later (Corollary \ref{corBOUNDCT}) that the invertibility condition
indeed holds in our setting.
\begin{theorem}\label{thmeq1}
Suppose the invertibility assumption holds. Then we have:
\begin{eqnarray}
\frac{1}{\nu}\cdot( \ve{\theta}^*_{t+1} - \ve{\theta}^*_t) & = & \matrice{v}_t \matrice{u}_{t+1} \ve{\theta}^*_t + \matrice{v}_{t+1} \ve{\epsilon}_{t}\nonumber\:\:, \forall t\geq 0\:\:,
\end{eqnarray}
where $\ve{\epsilon}_{t}$ is defined in eq. (\ref{defEPSILONT}).
\end{theorem}
\begin{proof}
Throughout the proof, we let
\begin{eqnarray}
\varsigma & \defeq & \mathrm{sign}(c) 
\end{eqnarray}
for short.
We have from Lemma \ref{lem11}, for any $t\geq 1$,
\begin{eqnarray}
\frac{1}{\nu}\cdot( \ve{\theta}^*_{t} - \ve{\theta}^*_{t-1}) & = & \matrice{v}_{t} \ve{\mu}_{t} - \matrice{v}_{t-1}\ve{\mu}_{t-1}\nonumber\\
 & = & \Delta_{t-1} \ve{\mu}_{t-1} + \matrice{v}_{t}\ve{\epsilon}_{t-1}\:\:,
\end{eqnarray}
with $\Delta_t \defeq \matrice{v}_{t+1} - \matrice{v}_{t}$. It comes
from eq. (\ref{eqXTXTM1}),
\begin{eqnarray}
\Delta_{t-1} & = & \left( \hat{\X}_{t-1} \hat{\X}_{t-1}^\top + \nu'\cdot \Gamma
  - \varsigma\cdot \ve{a}^+_t 
{\ve{b}^+_t}^\top - \varsigma\cdot\ve{b}^+_t
{\ve{a}^+_t}^\top\right)^{-1} -\matrice{v}_{t}\:\:.
\end{eqnarray}
To simplify this expression, we need two consecutive applications of
Sherman-Morrison's inversion formula:
\begin{eqnarray}
\lefteqn{\left(\hat{\X}_{t-1}\hat{\X}_{t-1}^\top + \nu'\cdot \Gamma - \varsigma\cdot\ve{a}^+_t 
{\ve{b}^+_t}^\top - \varsigma\cdot\ve{b}^+_t
{\ve{a}^+_t}^\top\right)^{-1}}\nonumber\\
 & = & \left(\hat{\X}_{t-1}\hat{\X}_{t-1}^\top + \nu'\cdot \Gamma - \varsigma\cdot\ve{a}^+_t 
{\ve{b}^+_t}^\top \right)^{-1} +
  \frac{\varsigma}{1-\varsigma\cdot{\ve{a}^+_t}^\top\left( \hat{\X}_{t-1}\hat{\X}_{t-1}^\top + \nu'\cdot \Gamma - \varsigma\cdot\ve{a}^+_t 
{\ve{b}^+_t}^\top \right)^{-1}\ve{b}^+_t} \cdot \Q_t\:\:,\label{sm1}
\end{eqnarray}
with
\begin{eqnarray}
\Q_t & \defeq & \left( \hat{\X}_{t-1}\hat{\X}_{t-1}^\top + \nu'\cdot \Gamma  - \varsigma\cdot\ve{a}^+_t {\ve{b}^+_t}^\top \right)^{-1} \ve{b}^+_t {\ve{a}^+_t}^\top \left( \hat{\X}_{t-1}\hat{\X}_{t-1}^\top + \nu'\cdot \Gamma  - \varsigma\cdot{\ve{a}^+_t} 
{\ve{b}^+_t}^\top \right)^{-1}\:\:,\nonumber\\
\end{eqnarray}
and
\begin{eqnarray}
\left(\hat{\X}_{t-1}\hat{\X}_{t-1}^\top + \nu'\cdot \Gamma -\varsigma\cdot {\ve{a}^+_t} 
{\ve{b}^+_t}^\top \right)^{-1}   & = &
\matrice{v}_{t-1} + \frac{\varsigma}{1-\varsigma\cdot{\ve{b}^+_t}^\top  \matrice{v}_{t-1}   {{\ve{a}^+_t}}}\cdot  \matrice{v}_{t-1}{\ve{a}^+_t} {\ve{b}^+_t}^\top  \matrice{v}_{t-1} \:\:.\label{sm2}
\end{eqnarray}
Let us define the
following shorthand:
\begin{eqnarray}
\Sigma_t & \defeq & \matrice{v}_{t-1} + \frac{\varsigma}{1-\varsigma\cdot{\ve{b}^+_t}^\top  \matrice{v}_{t-1}  {{\ve{a}^+_t}}}\cdot
  \matrice{v}_{t-1} {\ve{a}^+_t} {\ve{b}^+_t}^\top  \matrice{v}_{t-1}\:\:.
\end{eqnarray}
Then, plugging together eqs. (\ref{sm1}) and
    (\ref{sm2}), we get:
\begin{eqnarray}
\lefteqn{\left(\hat{\X}_{t-1}\hat{\X}_{t-1}^\top + \nu'\cdot \Gamma - \varsigma\cdot{\ve{a}^+_t} 
{\ve{b}^+_t}^\top - \varsigma\cdot\ve{b}^+_t
{{\ve{a}^+_t}}^\top\right)^{-1}}\nonumber\\
  & = & \matrice{v}_{t-1}  + \frac{\varsigma}{1-\varsigma\cdot{\ve{b}^+_t}^\top  \matrice{v}_{t-1}
    {{\ve{a}^+_t}}}\cdot  \matrice{v}_{t-1} {\ve{a}^+_t} {\ve{b}^+_t}^\top
  \matrice{v}_{t-1} \nonumber\\
 &  & +
  \frac{\varsigma}{1-\varsigma\cdot{{\ve{a}^+_t}}^\top \matrice{v}_{t-1}\ve{b}^+_t -
    \frac{
    {\ve{a}^+_t}^\top \matrice{v}_{t-1} {\ve{a}^+_t} \cdot {\ve{b}^+_t}^\top
    \matrice{v}_{t-1} \ve{b}^+_t}{1-\varsigma\cdot{\ve{b}^+_t}^\top  \matrice{v}_{t-1}
    {{\ve{a}^+_t}}}} \cdot \Sigma_t  \ve{b}^+_t {\ve{a}^+_t}^\top \Sigma_t \nonumber\\
  & = & \matrice{v}_{t-1}  + \frac{\varsigma}{1-\varsigma\cdot c_{1,t}}\cdot  \matrice{v}_{t-1} {\ve{a}^+_t} {\ve{b}^+_t}^\top
  \matrice{v}_{t-1} \nonumber\\
 &  & +
  \frac{\varsigma}{1- \varsigma\cdot c_{1,t} -
    \frac{c_{0,t} c_{2,t}}{1-\varsigma\cdot c_{1,t}}} \cdot \left( 
\begin{array}{c}
\matrice{v}_{t-1}  \\
+\\
  \frac{\varsigma}{1- \varsigma\cdot c_{1,t}}\cdot
  \matrice{v}_{t-1} {\ve{a}^+_t} {\ve{b}^+_t}^\top  \matrice{v}_{t-1}
\end{array}\right) \ve{b}^+_t {\ve{a}^+_t}^\top \left( \begin{array}{c}
\matrice{v}_{t-1}  \\
+\\
  \frac{\varsigma}{1- \varsigma\cdot c_{1,t}}\cdot
  \matrice{v}_{t-1} {\ve{a}^+_t} {\ve{b}^+_t}^\top  \matrice{v}_{t-1}
\end{array}
\right) \nonumber\\
 & = & \matrice{v}_{t-1}  + \frac{\varsigma}{1-\varsigma\cdot c_{1,t}}\cdot  \matrice{v}_{t-1} {\ve{a}^+_t} {\ve{b}^+_t}^\top
  \matrice{v}_{t-1} + \frac{\varsigma\cdot}{1- \varsigma\cdot c_{1,t} -
    \frac{c_{0,t} c_{2,t}}{1-\varsigma\cdot c_{1,t}}} \cdot \matrice{v}_{t-1} \ve{b}^+_t {{\ve{a}^+_t}}^\top
  \matrice{v}_{t-1} \nonumber\\
 && + \frac{c_{0,t}}{(1- \varsigma\cdot c_{1,t})^2 -
    c_{0,t} c_{2,t}} \cdot \matrice{v}_{t-1} \ve{b}^+_t {\ve{b}^+_t}^\top
  \matrice{v}_{t-1} + \frac{c_{2,t}}{(1- \varsigma\cdot c_{1,t})^2 -
    c_{0,t} c_{2,t}} \cdot \matrice{v}_{t-1} {\ve{a}^+_t} {{\ve{a}^+_t}}^\top
  \matrice{v}_{t-1} \nonumber\\
 & &+ \frac{\varsigma c_{0,t} c_{2,t}}{(1- \varsigma\cdot c_{1,t}) ((1- \varsigma\cdot c_{1,t})^2 -
    c_{0,t} c_{2,t})} \cdot \matrice{v}_{t-1} {\ve{a}^+_t} {\ve{b}^+_t}^\top
  \matrice{v}_{t-1}\nonumber\\
 & = & \matrice{v}_{t-1}  + \frac{1-\varsigma\cdot c_{1,t}}{(1-\varsigma\cdot c_{1,t})^2-c_{0,t}c_{2,t}}\cdot  \left(\matrice{v}_{t-1} {\ve{a}^+_t} {\ve{b}^+_t}^\top
  \matrice{v}_{t-1} + \matrice{v}_{t-1} \ve{b}^+_t {{\ve{a}^+_t}}^\top
  \matrice{v}_{t-1} \right) \nonumber\\
 & & + \frac{c_{0,t}}{(1- \varsigma\cdot c_{1,t})^2 -
    c_{0,t} c_{2,t}} \cdot \matrice{v}_{t-1} \ve{b}^+_t {\ve{b}^+_t}^\top
  \matrice{v}_{t-1} + \frac{c_{2,t}}{(1- \varsigma\cdot c_{1,t})^2 -
    c_{0,t} c_{2,t}} \cdot \matrice{v}_{t-1} {\ve{a}^+_t} {{\ve{a}^+_t}}^\top
  \matrice{v}_{t-1} \nonumber\\
 & = & \matrice{v}_{t-1}  + \frac{1}{(1-\varsigma\cdot c_{1,t})^2-c_{0,t}c_{2,t}}\cdot \left\{
\begin{array}{c}
   (1-\varsigma\cdot c_{1,t})\cdot (\matrice{v}_{t-1} {\ve{a}^+_t} {\ve{b}^+_t}^\top
  \matrice{v}_{t-1} + \matrice{v}_{t-1} \ve{b}^+_t {{\ve{a}^+_t}}^\top
  \matrice{v}_{t-1}) \\
+ c_{0,t} \cdot \matrice{v}_{t-1} \ve{b}^+_t {\ve{b}^+_t}^\top
  \matrice{v}_{t-1} \\
+ c_{2,t} \cdot \matrice{v}_{t-1} {\ve{a}^+_t} {{\ve{a}^+_t}}^\top
  \matrice{v}_{t-1}
\end{array}\right\}\nonumber\\
  & = & \matrice{v}_{t-1}  + \matrice{v}_{t-1}\matrice{u}_t \matrice{v}_{t-1}\:\:.
\end{eqnarray}
So,
\begin{eqnarray}
\frac{1}{\nu}\cdot( \ve{\theta}^*_{t} - \ve{\theta}^*_{t-1}) & = &
\Delta_{t-1} \ve{\mu}_{t-1} +
\matrice{v}_{t}\ve{\epsilon}_{t-1}\nonumber\\
 & = & \matrice{v}_{t-1} \matrice{u}_t \matrice{v}_{t-1}\ve{\mu}_{t-1} + \matrice{v}_{t} \ve{\epsilon}_{t-1}\nonumber\\
 & = &  \matrice{v}_{t-1} \matrice{u}_t \ve{\theta}^*_{t-1} + \matrice{v}_{t} \ve{\epsilon}_{t-1} \:\:, \label{eq22}
\end{eqnarray}
as claimed (end of the proof of Theorem \ref{thmeq1}). 
\end{proof}
All that remains to do now is to unravel the relationship in Theorem \ref{thmeq1}
and quantify the exact variation $\ve{\theta}^*_{T} -
\ve{\theta}^*_0$ as a function of $\ve{\theta}^*_0$ (which is the
error-free optimal classifier), holding for any permutation
$\PERM_*$. We therefore suppose that the invertibility assumption holds.
\begin{theorem}\label{thEXACT}
Suppose the invertibility assumption holds. For any $T\geq 1$,
\begin{eqnarray}
\ve{\theta}^*_{T} - \ve{\theta}^*_{0} & = & (\matrice{h}_{T,0} - \matrice{i}_d)
\ve{\theta}^*_0 + \sum_{t=0}^{T-1} \matrice{h}_{T,t+1} \ve{\lambda}_{t}\:\:.
\end{eqnarray}
\end{theorem}
\begin{proof}
We recall first that we have from Theorem \ref{thmeq1},
$\ve{\theta}^*_{t+1} - \ve{\theta}^*_t =  \Lambda_t \ve{\theta}^*_t
+ \ve{\lambda}_t$, $\forall t\geq 0$. Equivalently,
\begin{eqnarray}
\ve{\theta}^*_{t+1} & = & (\matrice{i}_d + \Lambda_t) \ve{\theta}^*_t + \ve{\lambda}_t\:\:.
\end{eqnarray}
Unravelling, we easily get $\forall T \geq 1$,
\begin{eqnarray}
\ve{\theta}^*_{T} & = & \prod_{t=0}^{T-1}(\matrice{i}_d + \Lambda_t)
\ve{\theta}^*_0  + \ve{\lambda}_{T-1} + \sum_{j=0}^{T-2} \prod_{t=j+1}^{T-1} (\matrice{i}_d
+ \Lambda_t) \ve{\lambda}_{j}\nonumber\\
 & = & \matrice{h}_{T,0}
\ve{\theta}^*_0 + \sum_{t=0}^{T-1} \matrice{h}_{T,t+1} \ve{\lambda}_{t}\:\:,
\end{eqnarray}
which yields the statement of Theorem \ref{thEXACT}.
\end{proof}
Since it applies to every permutation matrix, Theorem \ref{thEXACT} applies to \textit{every}
entity resolution algorithm. Theorem \ref{thEXACT} gives us a
interesting expression for the deviation $\ve{\theta}^*_{T} -
\ve{\theta}^*_0 $ which can be used to derive bounds on the distance
between the two classifiers. We apply it
now to derive one such bound.

\subsubsection{Finalizing the proof of Theorem \ref{thAPPROX1}}

We first need an intermediate technical Lemma. Let $\mu(\{a_i\}) \defeq (1/m) \cdot
\sum_i a_i$ denote for short the average of set $\{a_i\}_{i=1}^{m}$ with $a_i\geq 0, \forall i$. Let
$\gamma'\geq 0$ be \textit{any} real such that:
\begin{eqnarray}
\frac{\mu^2(\{a_i\})}{\mu(\{a^2_i\})} & \leq & (1-\gamma')\:\:.\label{eq0001}
\end{eqnarray}
Remark that the result is true for $\gamma' = 0$ since
$\mu(\{a^2_i\})-\mu^2(\{a_i\})$ is just the variance of $\{a_i\}$,
which is non-negative.
Remark also that we must have $\gamma'\leq 1$.
\begin{lemma}\label{lemsim}
$\sum_i \left( (1-\epsilon) a_i - q \right)^2 \geq \gamma'
(1-\epsilon)^2 \sum_i a_i^2$, $\forall \epsilon \leq 1, q\in \mathbb{R}$.
\end{lemma}
\begin{proof}
Remark that 
\begin{eqnarray}
\sqrt{(1-\gamma') \mu(\{a^2_i\})} & = & \inf_{k\geq 0} \frac{1}{2}\cdot \left(k
+ \frac{1}{k}\cdot (1-\gamma') \mu(\{a^2_i\})\right)\:\:,\label{eqLG3}
\end{eqnarray}
so we have:
\begin{eqnarray}
\mu(\{a_i\}) & \leq & \sqrt{(1-\gamma')
  \mu(\{a^2_i\})} \label{eqLG1}\\
 & \leq & \frac{q}{2(1-\epsilon)} +
\frac{(1-\gamma')(1-\epsilon)}{2q} \cdot \mu(\{a^2_i\})\:\:.\label{eqLG2}
\end{eqnarray}
Ineq. (\ref{eqLG1}) holds because of ineq. (\ref{eq0001}) and
ineq. (\ref{eqLG1}) holds because of eq. (\ref{eqLG3}) and
substituting $k \defeq q / (1-\epsilon) \geq 0$. After
reorganising, we obtain:
\begin{eqnarray}
n q^2 - 2(1-\epsilon) q \sum_i a_i +(1-\gamma')
(1-\epsilon)^2 \sum_i a_i^2 & \geq & 0\:\:,
\end{eqnarray}
and so we obtain the inequality of:
\begin{eqnarray}
\sum_i \left( (1-\epsilon) a_i - q \right)^2 & = & n q^2 -
2(1-\epsilon) q \sum_i a_i + (1-\epsilon)^2 \sum_i a_i^2
\nonumber\\
 & \geq & \gamma' (1-\epsilon)^2 \sum_i a_i^2\:\:,\label{eqa01}
\end{eqnarray}
which allows to conclude the proof of Lemma \ref{lemsim}.
\end{proof}
Let
\begin{eqnarray}
\gamma'(\X,\ve{w}) & \defeq & 1 - \frac{\mu^2(\{\vstretch(\ve{x}_i,\ve{w})\}_{i=1}^m) }{\mu(\{\vstretch^2(\ve{x}_i,\ve{w})\}_{i=1}^m)}\:\:,
\end{eqnarray}
where we recall that $\mu(\{a_i\}_{i=1}^m)$ is the average in set
$\{a_i\}$. It is easy to remark that $\gamma'(\X,\ve{w}) \in [0,1]$
and it can be used in Lemma \ref{lemsim} for the choice
\begin{eqnarray}
\{a_i\} & \defeq & \{\vstretch(\ve{x}_i,\ve{w})\}_{i=1}^m\:\:.
\end{eqnarray}
It is also not hard to see that as long as there exists two $\ve{x}_i$
in $\X$
with a different \textit{direction}, we shall have $\gamma'(\X,\ve{w})
> 0, \forall \ve{w}$.

Following \cite{bMA}, for any symmetric matrix $\matrice{m}$, we let
$\ve{\lambda}^\downarrow(\matrice{m})$ (resp. $\ve{\lambda}^\uparrow(\matrice{m})$) denote the vector of eigenvalues
arranged in decreasing (resp. increasing) order. So, $\lambda_1^\downarrow(\matrice{m})$
(resp. $\lambda_1^\uparrow(\matrice{m})$)
denotes the maximal (resp. minimal) eigenvalue of $\matrice{m}$.
\begin{lemma}\label{lemV2}
For any set $\mathcal{S} \defeq \{a_i\}_{i=1}^m$, let
$\mu(\mathcal{S})$ and $\sigma(\mathcal{S})$
denote the mean and standard deviation of $\mathcal{S}$. 
If
$\PERM_t$ is $(\epsilon,
\tau)$-accurate, let us define, $\forall i\in \{-1, 1\}$ and 
$\ve{w}\in \mathbb{R}^d$,
\begin{eqnarray}
M_{\min}(i, \ve{w}) & \defeq & \left\{
\begin{array}{ccl}
(1-\epsilon)^2 \sigma^2(\{\vstretch(\ve{x}_i,\ve{w})\}_{i=1}^n) & \mbox{ if } & i = +1\:\:,\\
-(1+\epsilon)^2 \mu(\{\vstretch^2(\ve{x}_i,\ve{w})\}_{i=1}^n)  -
  \tau^2 & \multicolumn{2}{c}{\mbox{ otherwise}}
\end{array}
\right.\:\:,
\end{eqnarray}
and $M_{\max}(i, \ve{w}) \defeq -M_{\min}(-i, \ve{w})$. Also,
$M_{\min}(i) \defeq \inf_{\ve{w}} M_{\min}(i, \ve{w})$ and $M_{\max}(i) \defeq \sup_{\ve{w}} M_{\max}(i, \ve{w})$.
Then the eigenspectrum of $\matrice{v}_t$ is
bounded as indicated:
\begin{eqnarray}
\lambda_1^\downarrow(\matrice{v}_t) & \leq & \frac{1}{m}\cdot
\frac{1}{M_{\min}(\mathrm{sign}(c))  + \frac{2\gamma}{|c|}
  \lambda_1^\uparrow(\Gamma)} \label{eq002F}\:\:,\\
\lambda_1^\uparrow(\matrice{v}_t) & \geq & \frac{1}{m}\cdot \frac{1}{M_{\max}(\mathrm{sign}(c))  + \frac{2\gamma}{|c|} \lambda_1^\uparrow(\Gamma)} \label{eq002G}\:\:,
\end{eqnarray}
where $c$ is defined in Lemma \ref{lem11}.
\end{lemma}
\noindent \textbf{Remark}: both $M_{\min}$ and $M_{\max}$ are in fact
of order $X_*^2$ in absolute value.
\begin{proof}
If $\PERM_t$ is $(\epsilon, \tau)$-accurate, it comes from the
triangle inequality 
\begin{eqnarray}
|\hat{\ve{x}}_{ti}^\top
\ve{w}| &  = & |\ve{x}_{i}^\top
\ve{w} + (\ve{x}_{ti_\shuffle} - \ve{x}_{i_\shuffle})^\top
\ve{w}_\shuffle| \nonumber\\
 & \geq & |\ve{x}_{i}^\top
\ve{w}| - |(\ve{x}_{ti_\shuffle} - \ve{x}_{i_\shuffle})^\top
\ve{w}_\shuffle| \nonumber\\
 & \geq & (1-\epsilon) |\ve{x}_i^\top
\ve{w}| - \tau\|\ve{w}\|_2 \:\:,\label{binfXW}
\end{eqnarray}
and also
\begin{eqnarray}
|\hat{\ve{x}}_{ti}^\top
\ve{w}| &  = & |\ve{x}_{i}^\top
\ve{w} + (\ve{x}_{ti_\shuffle} - \ve{x}_{i_\shuffle})^\top
\ve{w}_\shuffle| \nonumber\\
 & \leq & |\ve{x}_{i}^\top
\ve{w}| + |(\ve{x}_{ti_\shuffle} - \ve{x}_{i_\shuffle})^\top
\ve{w}_\shuffle| \nonumber\\
 & \leq & (1+\epsilon) |\ve{x}_i^\top
\ve{w}| + \tau\|\ve{w}\|_2 \:\:,\label{bsupXW}
\end{eqnarray}
so
\begin{eqnarray}
|\hat{\ve{x}}_{ti}^\top
\ve{w}| & \in & \left[(1-\epsilon) |\ve{x}_i^\top
\ve{w}| - \tau\|\ve{w}\|_2,  (1+\epsilon) |\ve{x}_i^\top
\ve{w}| + \tau\|\ve{w}\|_2\right]\:\:,\forall i \in [m], \forall t\geq
1\:\:.\label{intXW}
\end{eqnarray}
Using (\ref{intXW}) and Lemma \ref{lemsim} with $q \defeq \tau$ and 
$a_i \defeq \|\ve{x}_i\|_2 |\cos(\ve{x}_{i}, \ve{w})|$, we obtain the
last inequality of:
\begin{eqnarray}
\| \hat{\X}_{t}^\top \ve{w} \|_2^2 & = &  \sum_i (\hat{\ve{x}}_i^\top
\ve{w})^2 \nonumber\\
 & \geq & \sum_i ((1-\epsilon) |\ve{x}_i^\top
\ve{w}| - \tau\|\ve{w}\|_2)^2\nonumber\\
 & & = \|\ve{w}\|_2^2 \cdot \sum_i ((1-\epsilon) \vstretch(\ve{x}_i,\ve{w}) - \tau)^2\nonumber\\
& \geq & \|\ve{w}\|_2^2 (1-\epsilon)^2 \cdot \gamma'(\X, \ve{w})\sum_i \vstretch^2(\ve{x}_i,\ve{w}) \:\:,\label{eqb1}
\end{eqnarray}
but we also have, $\forall \ve{w}\in \mathbb{R}^d$,
\begin{eqnarray}
\sigma^2(\{\vstretch(\ve{x}_i,\ve{w})\}_{i=1}^n) & \defeq & \mu(\{\vstretch^2(\ve{x}_i,\ve{w})\}) - \mu^2(\{\vstretch(\ve{x}_i,\ve{w})\}) \nonumber\\
 & = & \left( 1 - \frac{\mu^2(\{\vstretch(\ve{x}_i,\ve{w})\}) }{\mu(\{\vstretch^2(\ve{x}_i,\ve{w})\})}\right)\cdot \mu(\{\vstretch^2(\ve{x}_i,\ve{w})\}) \nonumber\\
 & = & \left( 1 - \frac{\mu^2(\{\vstretch(\ve{x}_i,\ve{w})\}) }{\mu(\{\vstretch^2(\ve{x}_i,\ve{w})\})}\right)\cdot \frac{1}{m}\cdot\sum_i \vstretch^2(\ve{x}_i,\ve{w}) \nonumber\\
 & = & \frac{1}{m}\cdot \gamma'(\X, \ve{w}) \sum_i \vstretch^2(\ve{x}_i,\ve{w})\:\:,
\end{eqnarray}
so ineq. (\ref{eqb1}) yields
\begin{eqnarray}
\| \hat{\X}_{t}^\top \ve{w} \|_2^2 & \geq & m \|\ve{w}\|_2^2  (1-\epsilon)^2 \sigma^2(\{\vstretch(\ve{x}_i,\ve{w})\}_{i=1}^n)) \:\:.\label{binfXXW}
\end{eqnarray}
Using (\ref{intXW}) also yields
\begin{eqnarray}
\| \hat{\X}_{t}^\top \ve{w} \|_2^2 & = &  \sum_i (\hat{\ve{x}}_i^\top
\ve{w})^2 \nonumber\\
 & \leq & \|\ve{w}\|_2^2 \cdot \sum_i ((1+\epsilon) \vstretch(\ve{x}_i,\ve{w}) + \tau)^2\nonumber\\
& \leq & \|\ve{w}\|_2^2 \cdot \left(2 (1+\epsilon)^2\sum_i
  \vstretch^2(\ve{x}_i,\ve{w}) + 2m \tau^2\right)\nonumber\\
 & & = 2m \|\ve{w}\|_2^2 \left((1+\epsilon)^2 \mu(\{\vstretch^2(\ve{x}_i,\ve{w})\}_{i=1}^n)  + \tau^2 \right)\:\:,\label{eqb22}
\end{eqnarray}
because $(a+b)^2 \leq 2a^2 + 2b^2$. We get that if $\mathrm{sign}(c) =
+1$,
\begin{eqnarray}
\| \hat{\X}_{t}^\top \ve{w} \|_2^2 & \in & m \|\ve{w}\|_2^2\cdot
\left[(1-\epsilon)^2 \sigma^2(\{\vstretch(\ve{x}_i,\ve{w})\}_{i=1}^n)), (1+\epsilon)^2
  \mu(\{\vstretch^2(\ve{x}_i,\ve{w})\}_{i=1}^n)  +
  \tau^2\right]\:\:,\forall \ve{w}\:\:,\label{intXXW1}
\end{eqnarray}
while if $\mathrm{sign}(c) =
-1$,
\begin{eqnarray}
\| \hat{\X}_{t}^\top \ve{w} \|_2^2 & \in & -m \|\ve{w}\|_2^2\cdot
\left[(1+\epsilon)^2 \mu(\{\vstretch^2(\ve{x}_i,\ve{w})\}_{i=1}^n)  +
  \tau^2, (1-\epsilon)^2 \sigma^2(\{\vstretch(\ve{x}_i,\ve{w})\}_{i=1}^n))\right]\:\:,\forall \ve{w}\:\:.\label{intXXW2}
\end{eqnarray}
Now, we define $\forall i\in \{-1, 1\}$ and 
$\ve{w}\in \mathbb{R}^d$
\begin{eqnarray}
M_{\min}(i, \ve{w}) & \defeq & \left\{
\begin{array}{ccl}
(1-\epsilon)^2 \sigma^2(\{\vstretch(\ve{x}_i,\ve{w})\}_{i=1}^n)) & \mbox{ if } & i = +1\:\:,\\
-(1+\epsilon)^2 \mu(\{\vstretch^2(\ve{x}_i,\ve{w})\}_{i=1}^n)  -
  \tau^2 & \multicolumn{2}{c}{\mbox{ otherwise}}
\end{array}
\right.\:\:,
\end{eqnarray}
and $M_{\max}(i, \ve{w}) \defeq -M_{\min}(-i, \ve{w})$. We also let
$M_{\min}(i) \defeq \inf_{\ve{w}} M_{\min}(i, \ve{w})$ and $M_{\max}(i) \defeq \sup_{\ve{w}} M_{\max}(i, \ve{w})$.
Putting this
altogether, we obtain that 
if $\PERM_t$ is $(\epsilon, \tau)$-accurate, we have
\begin{eqnarray}
\lambda_1^\downarrow(\matrice{v}_t) & \defeq & \left(\inf_{\ve{w}} \frac{\ve{w}^\top\left( \mathrm{sign}(c)\hat{\X}_{t} \hat{\X}_{t}^\top + \nu'\cdot \Gamma
  \right) \ve{w}}{\|\ve{w}\|_2^2}\right)^{-1}\nonumber\\
 & \leq & \frac{1}{m \inf_{\ve{w}} M_{\min}(\mathrm{sign}(c), \ve{w}) + \nu'\lambda_1^\uparrow(\Gamma)}\nonumber\\
 & & = \frac{1}{m}\cdot \frac{1}{M_{\min}(\mathrm{sign}(c), \ve{w})  + \frac{2\gamma}{|c|} \lambda_1^\uparrow(\Gamma)}\:\:,
\end{eqnarray}
and
\begin{eqnarray}
\lambda_1^\uparrow(\matrice{v}_t) & \defeq & \left(\sup_{\ve{w}}
  \frac{\ve{w}^\top\left( \mathrm{sign}(c)\hat{\X}_{t} \hat{\X}_{t}^\top + \nu' \cdot \Gamma
  \right) \ve{w}}{\|\ve{w}\|_2^2}\right)^{-1}\nonumber\\
 & \geq & \frac{1}{ m \sup_{\ve{w}} M_{\max}(\mathrm{sign}(c), \ve{w}) + \nu'\lambda_1^\uparrow(\Gamma)}\nonumber\\
 & & = \frac{1}{m}\cdot \frac{1}{M_{\max}(\mathrm{sign}(c))  + \frac{2\gamma}{|c|} \lambda_1^\uparrow(\Gamma)}\:\:.
\end{eqnarray}
This ends the proof of Lemma \ref{lemV2}.
\end{proof}

\begin{lemma}\label{lemLAMBDAUT}
Suppose $(1-c_{1,t})^2-c_{0,t}c_{2,t} \neq 0$\footnote{This is implied
by the invertibility assumption.} and $\ve{a}_t
\neq \ve{0}$. Then $\matrice{u}_t$ is negative semi-definite iff
$(1-c_{1,t})^2-c_{0,t}c_{2,t} < 0$. Otherwise, $\matrice{u}_t$ is
indefinite. In all cases, for any $z\in
\{\lambda_1^\downarrow(\matrice{u}_t),
|\lambda_1^\uparrow(\matrice{u}_t)|\}$, we have
\begin{eqnarray}
z & \leq & \frac{2+ 3 (c_{0,t} + c_{2,t})}{2|(1-c_{1,t})^2 - c_{0,t}c_{2,t}|} \cdot \max\{\|\ve{a}_t\|^2_2,
\|\ve{b}_t\|^2_2\}\:\:.\label{ineqLAMBDAGENGEN}
\end{eqnarray}
\end{lemma}
\begin{proof}
Consider a block-vector following the column-block partition of $\matrice{u}_t$,
\begin{eqnarray}
\tilde{\ve{x}} & \defeq & \left[\begin{array}{c}
\ve{x}\\ \hline
\ve{y}
\end{array}\right]\:\:.\label{defVB}
\end{eqnarray}
Denote for short $\zeta \defeq (1-c_{1,t})^2-c_{0,t}c_{2,t}$. We have
\begin{eqnarray}
\matrice{u}_t \tilde{\ve{x}} & = & \frac{1}{\zeta}\cdot \left[\begin{array}{c}
(c_{2,t} (\ve{a}_t^\top \ve{x}) + (1-c_{1,t})(\ve{b}_t^\top \ve{y}))\cdot \ve{a}_t\\ \hline
((1 - c_{1,t}) (\ve{a}_t^\top \ve{x}) + c_{0,t}(\ve{b}_t^\top \ve{y}))\cdot \ve{b}_t
\end{array}\right]\:\:.\label{eqEIGV}
\end{eqnarray}
We see that the only possibility for $\tilde{\ve{x}}$ to be an
eigenvector is that $\ve{x} \propto \ve{a}_t$ and $\ve{y} \propto
\ve{b}_t$ (including the null vector for at most one vector). We now distinguish two cases.\\

\noindent \textbf{Case 1.} $c_{1,t} = 1$. In this case,
$\matrice{u}_t$ is block diagonal and so we get two eigenvectors:
\begin{eqnarray}
\matrice{u}_t \left[\begin{array}{c}
\ve{a}_t\\ \hline
\ve{0}
\end{array}\right] & = & -\frac{1}{c_{0,t}c_{2,t}}\cdot \left[\begin{array}{c|c}
c_{2,t} \cdot \ve{a}_t \ve{a}_t^\top & \matrice{0} \\ \cline{1-2}
\matrice{0} & c_{0,t} \cdot \ve{b}_t\ve{b}_t^\top 
\end{array}\right]\left[\begin{array}{c}
\ve{a}\\ \hline
\ve{0}
\end{array}\right]\nonumber\\
 & = & -\frac{1}{\lambda(\ve{a}^+_t)}\cdot \left[\begin{array}{c}
\ve{a}_t\\ \hline
\ve{0}
\end{array}\right]\:\:,
\end{eqnarray}
with (since $\|\ve{a}^+_t\|_2^2 = \|\ve{a}_t\|_2^2$):
\begin{eqnarray}
\lambda(\ve{a}^+_t) & \defeq & \frac{{\ve{a}^+_t}^\top  \matrice{v}_{t-1}
    \ve{a}^+_t}{\|\ve{a}^+_t\|_2^2}\:\:,
\end{eqnarray}
and
\begin{eqnarray}
\matrice{u}_t \left[\begin{array}{c}
\ve{0}\\ \hline
\ve{b}_t
\end{array}\right] & = & -\frac{1}{\lambda(\ve{b}^+_t)}\cdot \left[\begin{array}{c}
\ve{0}\\ \hline
\ve{b}_t
\end{array}\right] \:\:, \lambda(\ve{b}^+_t) \defeq \frac{{\ve{b}^+_t}^\top  \matrice{v}_{t-1}
    \ve{b}^+_t}{\|\ve{b}^+_t\|_2^2}\:\:.
\end{eqnarray}
We also remark that $\matrice{u}_t$ is negative semi-definite.

\noindent \textbf{Case 2.} $c_{1,t} \neq 1$. In this case, let us assume without loss of generality that for some
$\alpha \in \mathbb{R}_{*}$,
\begin{eqnarray}
\ve{x} & = & \alpha\cdot \ve{a}_t \:\:, \nonumber\\
\ve{y} & = & \ve{b}_t \:\:.\nonumber
\end{eqnarray}
In this case, we obtain 
\begin{eqnarray}
\matrice{u}_t \tilde{\ve{x}} & = & \frac{(1 - c_{1,t}) (\ve{a}_t^\top \ve{x}) + c_{0,t}(\ve{b}_t^\top \ve{y})}{(1-c_{1,t})^2-c_{0,t}c_{2,t}}\cdot \left[\begin{array}{c}
\frac{c_{2,t} (\ve{a}_t^\top \ve{x}) + (1-c_{1,t})(\ve{b}_t^\top
  \ve{y})}{(1 - c_{1,t}) (\ve{a}_t^\top \ve{x}) +
  c_{0,t}(\ve{b}_t^\top \ve{y})} \cdot \ve{a}_t\\ \hline
\ve{b}_t
\end{array}\right] \nonumber\\
 & = & \frac{\alpha (1 - c_{1,t})\|\ve{a}_t\|_2^2 + c_{0,t}\|\ve{b}_t\|_2^2}{(1-c_{1,t})^2-c_{0,t}c_{2,t}}\cdot \left[\begin{array}{c}
\frac{\alpha  c_{2,t} \|\ve{a}_t\|_2^2 + (1-c_{1,t})\|\ve{b}_t\|_2^2}{\alpha  (1-c_{1,t}) \|\ve{a}_t\|_2^2 + c_{0,t}\|\ve{b}_t\|_2^2}\cdot \ve{a}_t\\ \hline
\ve{b}_t
\end{array}\right]  \defeq \lambda \cdot \tilde{\ve{x}}\:\:,
\end{eqnarray}
and so we obtain the eigenvalue 
\begin{eqnarray}
\lambda & = & \frac{\alpha (1 - c_{1,t})\|\ve{a}_t\|_2^2 + c_{0,t}\|\ve{b}_t\|_2^2}{(1-c_{1,t})^2-c_{0,t}c_{2,t}}\:\:,
\end{eqnarray}
and we get from the eigenvector that $\alpha$ satisfies
\begin{eqnarray}
\alpha & = & \frac{\alpha  c_{2,t} \|\ve{a}_t\|_2^2 + (1-c_{1,t})\|\ve{b}_t\|_2^2}{\alpha (1-c_{1,t}) \|\ve{a}_t\|_2^2 + c_{0,t}\|\ve{b}_t\|_2^2}\:\:,
\end{eqnarray}
and so
\begin{eqnarray}
(1-c_{1,t}) \|\ve{a}_t\|_2^2 \alpha^2 +
(c_{0,t}\|\ve{b}_t\|_2^2-c_{2,t}\|\ve{a}_t\|_2^2) \alpha -
(1-c_{1,t})\|\ve{b}_t\|_2^2 & = & 0\:\:.
\end{eqnarray}
We note that the discriminant is
\begin{eqnarray}
\tau & = & 
(c_{0,t}\|\ve{b}_t\|_2^2-c_{2,t}\|\ve{a}_t\|_2^2)^2 + 4 (1-c_{1,t})^2
\|\ve{a}_t\|_2^2\|\ve{b}_t\|_2^2\:\:,
\end{eqnarray}
which is always $>0$.  Therefore we always have
two roots,
\begin{eqnarray}
\alpha_{\pm} & = & \frac{c_{2,t}\|\ve{a}_t\|_2^2 - c_{0,t}\|\ve{b}_t\|_2^2 \pm\sqrt{(c_{0,t}\|\ve{b}_t\|_2^2-c_{2,t}\|\ve{a}_t\|_2^2)^2 + 4 (1-c_{1,t})^2
\|\ve{a}_t\|_2^2\|\ve{b}_t\|_2^2}}{2 (1-c_{1,t}) \|\ve{a}_t\|_2^2 }\:\:.
\end{eqnarray}
yielding two non-zero eigenvalues,
\begin{eqnarray}
\lambda_{\pm}(\matrice{u}_t) & = & \frac{1}{2\zeta}\cdot\left(c_{2,t}\|\ve{a}_t\|_2^2 + c_{0,t}\|\ve{b}_t\|_2^2 \pm\sqrt{(c_{0,t}\|\ve{b}_t\|_2^2-c_{2,t}\|\ve{a}_t\|_2^2)^2 + 4 (1-c_{1,t})^2
\|\ve{a}_t\|_2^2\|\ve{b}_t\|_2^2}\right)\:\:.
\end{eqnarray}
Let us analyze the sign of both eigenvalues. For the
numerator of $\lambda_-$ to be negative, we have equivalently after simplification
\begin{eqnarray}
(c_{2,t}\|\ve{a}_t\|_2^2 + c_{0,t}\|\ve{b}_t\|_2^2)^2 & < & (c_{0,t}\|\ve{b}_t\|_2^2-c_{2,t}\|\ve{a}_t\|_2^2)^2 + 4 (1-c_{1,t})^2
\|\ve{a}_t\|_2^2\|\ve{b}_t\|_2^2\:\:,
\end{eqnarray}
which simplifies in $c_{0,t}c_{2,t} < (1-c_{1,t})^2$, \textit{i.e.} $\zeta
> 0$. Hence, $\lambda_- < 0$.

Now, for $\lambda_+$, it is easy to check that its sign is that of
$\zeta$. When $\zeta>0$, we have $\lambda_+ \geq |\lambda_-|$, and
because $a^2 + b^2 \leq (|a|+|b|)^2$, we get
\begin{eqnarray}
\lambda_1^\downarrow(\matrice{u}_t) = \lambda_+ & \leq & \frac{1}{2}\cdot\left(c_{2,t}\|\ve{a}_t\|_2^2 +
  c_{0,t}\|\ve{b}_t\|_2^2 + |c_{0,t}\|\ve{b}_t\|_2^2-c_{2,t}\|\ve{a}_t\|_2^2| + 2 (1-c_{1,t})
\|\ve{a}_t\|_2\|\ve{b}_t\|_2\right)\nonumber\\
 & \leq & c_{2,t}\|\ve{a}_t\|_2^2 +
  c_{0,t}\|\ve{b}_t\|_2^2 + (1-c_{1,t})
\|\ve{a}_t\|_2\|\ve{b}_t\|_2\:\:.\label{defLAMBDA}
\end{eqnarray}
Now, remark that because $\matrice{v}_t$ is positive definite,
\begin{eqnarray}
c_{0,t} - 2 c_{1,t} + c_{2,t} & \defeq & {{\ve{a}^+_t}}^\top  \matrice{v}_{t}
    {\ve{a}^+_t} - 2 {{\ve{a}^+_t}}^\top  \matrice{v}_{t}
    {\ve{b}^+_t} + {{\ve{b}^+_t}}^\top  \matrice{v}_{t}
    {\ve{b}^+_t}\nonumber\\
 & = & (\ve{a}^+_t-\ve{b}^+_t)^\top\matrice{v}_t
 (\ve{a}^+_t-\ve{b}^+_t) \nonumber\\
 & \geq & 0\:\:,
\end{eqnarray}
showing that $c_{1,t}  \leq (c_{0,t}+ c_{2,t})/2$. So we get from
ineq. (\ref{defLAMBDA}),
\begin{eqnarray}
\lambda_1^\downarrow(\matrice{u}_t) & \leq &  \frac{1}{\zeta}\cdot\left( c_{2,t}\|\ve{a}_t\|_2^2 +
  c_{0,t}\|\ve{b}_t\|_2^2 + \left(1+\frac{c_{0,t} + c_{2,t}}{2}\right)
\|\ve{a}_t\|_2\|\ve{b}_t\|_2\right)\nonumber\\
 & \leq & \frac{1}{\zeta}\cdot\left(1 +
  \frac{3}{2}\cdot (c_{0,t} + c_{2,t}) \right) \cdot \max\{\|\ve{a}_t\|^2_2,
\|\ve{b}_t\|^2_2\}\nonumber\\
 & \leq & \frac{2+ 3 (c_{0,t} + c_{2,t})}{2((1-c_{1,t})^2 - c_{0,t}c_{2,t})} \cdot \max\{\|\ve{a}_t\|^2_2,
\|\ve{b}_t\|^2_2\}\:\:.\label{ineqLAMBDAGEN}
\end{eqnarray}
When $\zeta<0$, we remark that $\lambda_+ < \lambda_-$ and so
$\matrice{u}_t$ is negative semi-definite.\\

Whenever $c_{1,t} \neq 1$, it is then easy to check that for any $z
\in \{|\lambda_+|, |\lambda_-|\}$,
ineq. (\ref{ineqLAMBDAGEN}) brings
\begin{eqnarray}
z & \leq &  \frac{2+ 3 (c_{0,t} + c_{2,t})}{2|(1-c_{1,t})^2 - c_{0,t}c_{2,t}|} \cdot \max\{\|\ve{a}_t\|^2_2,
\|\ve{b}_t\|^2_2\}\:\:.\label{ineqLAMBDAGEN2}
\end{eqnarray}
Whenever $c_{1,t} = 1$ (Case 1.), it is also immediate to check that
for any $z
\in \{|-1/\lambda(\ve{a}^+_t)|, |-1/\lambda(\ve{b}^+_t)|\}$,
\begin{eqnarray}
z & \leq & \max\left\{\frac{1}{c_{0,t}}, \frac{1}{c_{2,t}}\right\}\cdot \max\{\|\ve{a}_t\|^2_2,
\|\ve{b}_t\|^2_2\} \nonumber\\
 & < & \left(1+\frac{3}{c_{0,t}}+ \frac{3}{c_{2,t}}\right) \cdot \max\{\|\ve{a}_t\|^2_2,
\|\ve{b}_t\|^2_2\}\nonumber\\
 & & = \frac{2+ 3 (c_{0,t} + c_{2,t})}{2|(1-c_{1,t})^2 - c_{0,t}c_{2,t}|} \cdot \max\{\|\ve{a}_t\|^2_2,
\|\ve{b}_t\|^2_2\}\:\:.
\end{eqnarray}
Once we remark that $c_{1,t} = 1$ implies $\zeta < 0$, we obtain the statement of Lemma \ref{lemLAMBDAUT}.
\end{proof}

\begin{lemma}\label{lemV}
If $\PERM_t$ is $(\epsilon, \tau)$-accurate, then the following
holds true:
\begin{eqnarray}
\|\ve{b}^+_t\|_2^2 = \|\ve{b}_t\|_2^2 & \leq & 2\xi \cdot X_*^2\:\:,\label{b11}\\
\|\ve{a}^+_t\|_2^2 = \|\ve{a}_t\|_2^2 & \leq & 2\xi \cdot X_*^2\:\:,\label{b12}
\end{eqnarray}
where $\xi$ is defined in eq. (\ref{defXI}).
\end{lemma}
\begin{proof}
To prove ineq. (\ref{b11}), we make two applications of point 2. in the $(\epsilon, \tau)$-accuracy
assumption with $\F \defeq \shuffle$:
\begin{eqnarray}
\vstretch((\ve{x}_{\ub{t}} -
\ve{x}_{\vb{t}})_{\F},\ve{w}_{\F}) & \leq & \epsilon \cdot \max_{i\in
  \{\ub{t}, \vb{t}\}} \vstretch(\ve{x}_{i},\ve{w})+
\tau \:\:, \nonumber\\
 & & \forall \ve{w}\in \mathbb{R}^d :
\|\ve{w}\|_2 = 1 \:\:.
\end{eqnarray}
Fix $\ve{w} \defeq (1/\|\ve{x}_{\vb{t}}\|_2)\cdot \ve{x}_{\vb{t}}$. We get:
\begin{eqnarray}
|(\ve{x}_{\vb{t}} - \ve{x}_{\ub{t}})_\shuffle^\top \ve{x}_{\vb{t}_\shuffle}|
& \leq & \epsilon \cdot \max\{|\ve{x}_{\ub{t}}^\top \ve{x}_{\vb{t}}|,
\|\ve{x}_{\vb{t}} \|_2^2\} + \tau \cdot \|\ve{x}_{\vb{t}} \|_2
\nonumber\\
 & \leq & \epsilon \cdot X_*^2 + \tau \cdot X_*  = \xi \cdot X_*^2\:\:.\label{feqq1}
\end{eqnarray} 
Fix $\ve{w} \defeq (1/\|\ve{x}_{\ub{t}}\|_2)\cdot \ve{x}_{\ub{t}}$. We get:
\begin{eqnarray}
|(\ve{x}_{\ub{t}} - \ve{x}_{\vb{t}})_\shuffle^\top \ve{x}_{\ub{t}_\shuffle}|
& \leq & \epsilon \cdot \max\{|\ve{x}_{\ub{t}}^\top \ve{x}_{\vb{t}}|,
\|\ve{x}_{\ub{t}} \|_2^2\} + \tau \cdot \|\ve{x}_{\ub{t}} \|_2
\nonumber\\
 & \leq & \epsilon \cdot X_*^2 + \tau \cdot X_* = \xi \cdot X_*^2 \:\:. \label{feqq2}
\end{eqnarray} 
Folding together ineqs. (\ref{feqq1}) and (\ref{feqq2}) yields
\begin{eqnarray}
\lefteqn{\|(\ve{x}_{\vb{t}} - \ve{x}_{\ub{t}})_\shuffle\|_2^2 = (\ve{x}_{\vb{t}} -
\ve{x}_{\ub{t}})^\top_\shuffle (\ve{x}_{\vb{t}} - \ve{x}_{\ub{t}})_\shuffle}
\nonumber\\
 & \leq & |(\ve{x}_{\vb{t}} -
\ve{x}_{\ub{t}})^\top_\shuffle \ve{x}_{\vb{t} _\shuffle}| + |(\ve{x}_{\vb{t}} -
\ve{x}_{\ub{t}})^\top_\shuffle \ve{x}_{\ub{t} _\shuffle}| \nonumber\\
 & \leq & 2  \xi \cdot X_*^2\:\:.
\end{eqnarray}
We get 
\begin{eqnarray}
\|\ve{b}^+_t\|_2^2 = \|\ve{b}_t\|_2^2 = \|(\ve{x}_{\vb{t}} - \ve{x}_{\ub{t}})_\shuffle\|_2^2 & \leq & 2  \xi \cdot X_*^2\:\:,\label{b11F}
\end{eqnarray}
which yields ineq. (\ref{b11}).
To get ineq. (\ref{b12}),
we switch $\F \defeq \shuffle$ by $\F \defeq \anchor$ in our
application of point 2. in the $(\epsilon, \tau)$-accuracy
assumption. 
\end{proof} 
\begin{lemma}\label{lemBOUNDCT}
If $\PERM_t$ is $(\epsilon, \tau)$-accurate and the data-model calibration assumption holds,
\begin{eqnarray}
c_{i,t} & \leq & \frac{1}{12}\:\:, \forall i \in \{0, 1, 2\}\:\:.
\end{eqnarray}
\end{lemma}
\begin{proof}
We remark that
\begin{eqnarray}
c_{0,t} & \defeq & {\ve{a}^+_t}^\top  \matrice{v}_{t}
    {\ve{a}^+_t}\nonumber\\
 & \leq & \lambda_1^\downarrow(\matrice{v}_t) \|\ve{a}^+_t\|_2^2\nonumber\\
 & \leq & 2 \lambda_1^\downarrow(\matrice{v}_t) \xi \cdot X_*^2\:\:,\nonumber
\end{eqnarray}
and for the same reasons, $c_{2,t} \leq 2 \lambda_1^\downarrow(\matrice{v}_t) \xi \cdot X_*^2$. Hence, it comes from the proof of Lemma
 \ref{lemLAMBDAUT} that we also have $c_{2,t} \leq 2 \lambda_1^\downarrow(\matrice{v}_t) \xi
 \cdot X_*^2$. Using ineq. (\ref{eq002F}) in Lemma \ref{lemV2}, we thus obtain for any $i \in \{0,
 1, 2\}$:
\begin{eqnarray}
c_{i,t} & \leq & \frac{1}{m}\cdot \frac{2 \xi \cdot X_*^2}{ M_{\min}(\mathrm{sign}(c), \ve{w})  + \frac{2\gamma}{|c|} \lambda_1^\uparrow(\Gamma)} \nonumber\\
 & & = \frac{\xi}{m}\cdot |c| \cdot \frac{X_*^2}{
   \frac{|c|}{2}  M_{\min}(\mathrm{sign}(c), \ve{w})  +  \gamma \lambda_1^\uparrow(\Gamma)}\nonumber\\
 & \leq & \frac{1}{4} \cdot \frac{1}{4} < \frac{1}{12}\:\:,
\end{eqnarray}
as claimed. The last inequality uses the data-model calibration assumption.
\end{proof}

\begin{corollary}\label{corBOUNDCT}
Suppose $\PERM_t$ is $(\epsilon, \tau)$-accurate for any $t\geq 1$ and the data-model
calibration assumption holds. Then the invertibility assumption holds.
\end{corollary}
\begin{proof}
From Lemma \ref{lemBOUNDCT}, we conclude that $(1-c_{1,t})^2 > 121/144
> 1/144 > c_{0,t}c_{2,t} > 0$, hence the invertibility assumption holds.
\end{proof}

\begin{lemma}\label{boundLAMBDAT}
If $\PERM_t$ is $(\epsilon, \tau)$-accurate and the data-model
calibration assumption holds, the following
holds true: $\matrice{i}_d + \Lambda_t\succ 0$ and 
\begin{eqnarray}
\lambda_1^\downarrow\left(\Lambda_t\right) & \leq & \frac{\xi}{m} \:\:.\label{condLAMBDAT}
\end{eqnarray}
\end{lemma}
\begin{proof}
First note that $\lambda_1^\uparrow(\matrice{v}_{t}) \geq 1/(\gamma
\lambda_1^\downarrow(\Gamma)) > 0$ and so $\matrice{v}_{t} \succ 0$,
which implies that $\Lambda_t \defeq \nu\matrice{v}_{t}\matrice{u}_t =
\nu\matrice{v}^{1/2}_{t}(\matrice{v}^{1/2}_{t} \matrice{u}_t
\matrice{v}^{1/2}_{t}) \matrice{v}^{-1/2}_{t}$, \textit{i.e.}
$\Lambda_t$ is similar to a symmetric matrix ($\matrice{v}^{1/2}_{t} \matrice{u}_t
\matrice{v}^{1/2}_{t} $) and therefore has only
real eigenvalues. 
We get
\begin{eqnarray}
\lambda_1^\downarrow\left(\Lambda_t\right) & = &
\lambda_1^\downarrow\left(\nu \cdot \matrice{v}_{t}\matrice{u}_t\right) \nonumber\\
 & \leq & |\nu| \lambda_1^\downarrow(\matrice{v}_t)  \cdot \left(1 +
  \frac{3}{2}\cdot (c_{0,t} + c_{2,t}) \right) \cdot \max\{\|\ve{a}_t\|^2_2,
\|\ve{b}_t\|^2_2\} \label{eq111}\\
 & \leq & \frac{2 +
  3(c_{0,t} + c_{2,t})}{|(1-c_{1,t})^2 - c_{0,t}c_{2,t}|} \cdot  |\nu| \lambda_1^\downarrow(\matrice{v}_t)  \xi \cdot X_*^2 \label{eq112}\:\:.
\end{eqnarray}
Ineq. (\ref{eq111}) is due to Lemma \ref{lemLAMBDAUT} and
ineq. (\ref{eq112}) is due to Lemma \ref{lemV}. We now use Lemma
\ref{lemBOUNDCT} and its proof, which shows that 
\begin{eqnarray}
(1-c_{1,t})^2 -
c_{0,t}c_{2,t} & \geq & \left(1-\frac{1}{12}\right)^2 - \frac{1}{144}
\nonumber\\
 & & = \frac{5}{6}\:\:. 
\end{eqnarray}
Letting $U\defeq |\nu| \lambda_1^\downarrow(\matrice{v}_t)\xi \cdot X_*^2$ for short, we thus get from the proof of Lemma
\ref{lemBOUNDCT}:
\begin{eqnarray}
\lambda_1^\downarrow\left(\Lambda_t\right) & \leq & \frac{6}{5} \cdot (2+3(U +
U))U\nonumber\\
 & & = \frac{6}{5} \cdot (2U+6U^2)\:\:.
\end{eqnarray}
Now we want $\lambda_1^\downarrow\left(\Lambda_t\right) \leq \xi /m$, which translates into a second-order inequality for $U$,
whose solution imposes the following upperbound on $U$:
\begin{eqnarray}
6 U & \leq & -1 + \sqrt{1+\frac{5\xi}{m}}\:\:.\label{condUU}
\end{eqnarray}
We can indeed forget the lowerbound for $U$, whose sign is negative while
$U\geq 0$. 

Since $\sqrt{1+x}\geq 1 + (x/2) - (x^2/8)$ for $x\geq 0$
(and $\xi/m \geq 0$), we get the
sufficient condition for ineq. (\ref{condUU}) to be satisfied:
\begin{eqnarray}
6|\nu| \lambda_1^\downarrow(\matrice{v}_t) \xi
 \cdot X_*^2 & \leq & \frac{5 \xi}{2m} -\frac{25}{8}\cdot \left(\frac{\xi}{m}\right)^2\:\:.\label{constLAMBDA2}
\end{eqnarray}
Now, it comes from Lemma \ref{lemV2} that a sufficient condition for
ineq. (\ref{constLAMBDA2}) is that
\begin{eqnarray}
\frac{\xi}{m}\cdot \frac{6|\nu| X_*^2}{ M_{\min}(\mathrm{sign}(c), \ve{w})  + \frac{2\gamma}{|c|}
  \lambda_1^\uparrow(\Gamma)}  & \leq & \frac{5\xi}{2m} -\frac{25}{8}\cdot \left(\frac{\xi}{m}\right)^2\:\:,\label{constLAMBDA3}
\end{eqnarray}
which, after simplification, is equivalent to
\begin{eqnarray}
\frac{12}{5}\cdot \frac{|\nu| X_*^2}{ M_{\min}(\mathrm{sign}(c), \ve{w})  + \frac{2\gamma}{|c|}
  \lambda_1^\uparrow(\Gamma)}  + \frac{5\xi}{4m} & \leq & 1\:\:,\label{constLAMBDA31}
\end{eqnarray}
or,
\begin{eqnarray}
\frac{6 |F'(0)|}{5}\cdot \frac{ X_*^2}{\frac{|c|}{2} M_{\min}(\mathrm{sign}(c), \ve{w})  + \gamma
  \lambda_1^\uparrow(\Gamma)}  + \frac{5\xi}{4m} & \leq & 1\:\:,\label{constLAMBDA32}
\end{eqnarray}
But, the data-model calibration assumption implies that the left-hand
side is no more than $(3/5) + (5/16) = 73/80 < 1$, and 
ineq. (\ref{condLAMBDAT}) follows.\\

It also trivially follows that $\matrice{i}_d + \Lambda_t$ has only
real eigenvalues. To prove that they are all strictly positive, we
know that the only potentially negative eigenvalue of $\matrice{u}_t$,
$\lambda_-$ (Lemma \ref{lemLAMBDAUT}) is smaller in absolute value to
$\lambda_1^\downarrow(\matrice{u}_t)$. $\matrice{v}_t$ being positive
definite, we thus have under the $(\epsilon, \tau)$-accuracy
assumption and data-model calibration:
\begin{eqnarray}
\lambda_1^\uparrow(\matrice{i}_d + \Lambda_t) & \geq & 1 -
\frac{\xi}{m} \nonumber\\
 & \geq & 1 - \frac{1}{4}  = \frac{3}{4} > 0\:\:,
\end{eqnarray}
showing $\matrice{i}_d + \Lambda_t$ is positive definite.
This ends the
proof of Lemma \ref{boundLAMBDAT}.
\end{proof}

We recall that $0\leq T_+\leq T$ denote the number of elementary permutations
that act between classes, and $\rho \defeq T_+ / T$ denote the
proportion of such elementary permutations among all.

\begin{theorem}\label{thAPPROX1SMB}
Suppose $\PERM_*$ is $(\epsilon, \tau)$-accurate and
$\alpha$-bounded, and the data-model calibration assumption holds. Then the following holds for all $T\geq 1$:
\begin{eqnarray}
\|\ve{\theta}^*_{T} - \ve{\theta}^*_0 \|_2 & \leq & \frac{\xi}{m} \cdot T^2 \cdot \left( \|\ve{\theta}^*_0 \|_2 +
\frac{\sqrt{\xi}}{4 X_*} \cdot \rho\right)\nonumber\\
 & \leq & \left(\frac
  {\xi}{m}\right)^{\alpha} \cdot \left(
  \|\ve{\theta}^*_0 \|_2 + \frac{\sqrt{\xi}}{4 X_*}
  \cdot \rho
\right)\:\:.\label{eqthAPPROX0}
\end{eqnarray}
\end{theorem}
\begin{proof}
We use Theorem \ref{thEXACT}, which yields from the triangle inequality:
\begin{eqnarray}
\|\ve{\theta}^*_{T} - \ve{\theta}^*_{0}\|_2 & = & \|(\matrice{h}_{T,0} - \matrice{i}_d)
\ve{\theta}^*_0\|_2 + \left\|\sum_{t=0}^{T-1} \matrice{h}_{T,t+1} \ve{\lambda}_{t}\right\|_2\:\:.\label{condAT2}
\end{eqnarray}
Denote for short $q\defeq \xi/m$. It comes from the definition of $\matrice{h}_{i,j}$ and
Lemma \ref{boundLAMBDAT} the first inequality of:
\begin{eqnarray}
\lambda_1^\downarrow\left(\matrice{h}_{T,0} - \matrice{i}_d\right) &
\leq & (1+q)^T-1 \nonumber\\
 & \leq & T^2 q\:\:,\label{condLAMBDAH}
\end{eqnarray}
where the second inequality holds because  ${T \choose k} q^k\leq
(Tq)^k\leq Tq$ for $k\geq 1$ whenever $Tq \leq 1$, which is equivalent
to
\begin{eqnarray}
T & \leq & \frac{m}{\xi}\:\:,\label{condT}
\end{eqnarray}
which is implied by the condition of $\alpha$-bounded permutation
size ($n/\xi \geq 4 \geq 1$ from the data-model
calibration assumption). We thus get 
\begin{eqnarray}
\|\left(\matrice{h}_{T,0} - \matrice{i}_d\right) \ve{\theta}^*_{0}\|_2 & \leq & T^2 q \cdot \|\ve{\theta}^*_{0}\|_2\:\:.\label{eq001}
\end{eqnarray}
Using ineq. (\ref{condAT2}), this shows the statement of the Theorem
with (\ref{condAT}). The upperbound comes from the fact that the factor in the right hand side is no more than
$(\xi/m)^\alpha$ for some $0\leq \alpha \leq 1$
provided this time the stronger constraint holds:
\begin{eqnarray}
T & \leq & \left(\frac{m}{\xi}\right)^{\frac{1-\alpha}{2}}\:\:,\label{condT2}
\end{eqnarray}
which is the condition of $\alpha$-boundedness.

Let us now have a look at the shift term in eq. (\ref{condAT2}), which depends only on the
mistakes between classes done during the permutation (which changes
the mean operator between permutations),
\begin{eqnarray}
\matrice{r} & \defeq & \sum_{t=0}^{T-1} \matrice{h}_{T,t+1} \ve{\lambda}_{t}\:\:.
\end{eqnarray}
Using eq. (\ref{defL2}), we can simplify $\matrice{r}$ since $\ve{\lambda}_{t} = \nu \matrice{v}_{t+1} \ve{\epsilon}_{t}$, so if we define $\matrice{g}_{.,.}$ from
$\matrice{h}_{.,.}$ as follows, for $0\leq j \leq i$:
\begin{eqnarray}
\matrice{g}_{i,j} & \defeq & \nu \matrice{h}_{i, j} \matrice{v}_{j} \:\:,\label{defGIJTILDE}
\end{eqnarray}
then we get
\begin{eqnarray}
\matrice{r} & \defeq & \sum_{t=0}^{T-1} \matrice{g}_{T,t+1} \ve{\epsilon}_{t}\:\:,
\end{eqnarray}
where we recall that $\ve{\epsilon}_{t} \defeq \ve{\mu}_{t+1} -
  \ve{\mu}_{t}$ is the shift in the mean operator, \textit{which is
    the null vector whenever $\PERM_t$ acts in a specific class} ($y_{\ua{t}}
  = y_{\va{t}}$). To see this, 
we remark
\begin{eqnarray}
\ve{\epsilon}_{t} & \defeq & \ve{\mu}_{t+1} -
  \ve{\mu}_{t}\nonumber\\
 & = & \sum_i y_i \cdot \left[
\begin{array}{c}
\ve{x}_{i_\anchor}\\\cline{1-1}
\ve{x}_{{(t+1)i}_\shuffle} 
\end{array}
\right] - \sum_i y_i \cdot \left[
\begin{array}{c}
\ve{x}_{i_\anchor}\\\cline{1-1}
\ve{x}_{{ti}_\shuffle} 
\end{array}
\right]\nonumber\\
 & = & \sum_i y_i \cdot \left[
\begin{array}{c}
0\\\cline{1-1}
\ve{x}_{{(t+1)i}_\shuffle} 
\end{array}
\right] - \sum_i y_i \cdot \left[
\begin{array}{c}
0\\\cline{1-1}
\ve{x}_{{ti}_\shuffle} 
\end{array}
\right]\nonumber\\
 & = & \left[
\begin{array}{c}
0\\\cline{1-1}
\sum_i y_i \cdot (\ve{x}_{{(t+1)i}_\shuffle} - \ve{x}_{{ti}_\shuffle}) 
\end{array}
\right] \defeq \left[
\begin{array}{c}
0\\\cline{1-1}
{\ve{\epsilon}_{t}}_\shuffle
\end{array}
\right]\:\:,
\end{eqnarray}
which can be simplified further since we work with the elementary
permutation $\PERM_{t}$,
\begin{eqnarray}
{\ve{\epsilon}_{t}}_\shuffle & = & y_{\ua{t}} \cdot
(\ve{x}_{{\vb{t}}}-\ve{x}_{{\ub{t}}})_\shuffle +  y_{\va{t}} \cdot
(\ve{x}_{{\ub{t}}}-\ve{x}_{{\vb{t}}})_\shuffle\nonumber\\
 & = & (y_{\ua{t}} - y_{\va{t}}) \cdot (\ve{x}_{{\vb{t}}}-\ve{x}_{{\ub{t}}})_\shuffle\:\:.
\end{eqnarray}
Hence, 
\begin{eqnarray}
\|\ve{\epsilon}_{t} \|_2 =\|{\ve{\epsilon}_{t}}_\shuffle \|_2 & = &
1_{y_{\ua{t}} \neq y_{\va{t}}} \cdot
\|(\ve{x}_{{\vb{t}}}-\ve{x}_{{\ub{t}}})_\shuffle\|_2\nonumber\\
 & \leq & 1_{y_{\ua{t}} \neq y_{\va{t}}} \cdot \sqrt{2\xi)} X_*\:\:,\label{beps}
\end{eqnarray}
from Lemma \ref{lemV},  and we see that indeed $\|\ve{\epsilon}_{t} \|_2 = 0$ when the
elementary permutation occurs within observations of the same class.\\

It follows from the data-model calibration assumption and Lemma \ref{lemV2} that
\begin{eqnarray}
\lambda_1^\downarrow(\matrice{v}_t) & \leq & \frac{1}{m}\cdot
\frac{1}{M_{\min}(\mathrm{sign}(c))  + \frac{2\gamma}{|c|}
  \lambda_1^\uparrow(\Gamma)} \nonumber\\
 & & = \frac{|c|}{2mX_*^2}\cdot
\frac{X_*^2}{\frac{|c|}{2}\cdot M_{\min}(\mathrm{sign}(c))  + \gamma
  \lambda_1^\uparrow(\Gamma)}\nonumber\\
& \leq & \frac{|c|}{2mX_*^2}\cdot\frac{1}{2}\cdot \min\left\{\frac{1}{|F'(0)|}, \frac{1}{2|c|}\right\}\nonumber\\
 & \leq & \frac{1}{8mX_*^2}\:\:.\label{boundLAMBDAVT2}
\end{eqnarray}
Using \cite{bMA} (Problem III.6.14), Lemma \ref{boundLAMBDAT} and
ineq. (\ref{boundLAMBDAVT2}), we also obtain
\begin{eqnarray}
\lambda^\downarrow_1\left(\matrice{g}_{T,t+1}\right) & \leq & 2 \cdot \left(1+
  \frac{\xi}{n}
\right)^{T-t-1} \cdot \frac{1}{8mX_*^2}\:\:.
\end{eqnarray}
So, 
\begin{eqnarray}
\|\matrice{r}\|_2 & \leq & \sum_{t=0}^{T-1}
\lambda_{\mathrm{max}}\left(\matrice{g}_{T,t+1}\right)
\|\ve{\epsilon}_{t}\|_2\nonumber\\
 & \leq & \frac{1}{2\sqrt{2}} \cdot \sum_{t=0}^{T-1}
1_{y_{\ua{t}} \neq y_{\va{t}}} \cdot \left(1+
  \frac{\xi}{n}
\right)^{T-t-1} \cdot \frac{\sqrt{\xi}}{mX_*}\nonumber\\
 & & = \frac{1}{2X_*} \cdot \sqrt{\frac{\xi}{2}}\cdot \sum_{t=0}^{T-1}
1_{y_{\ua{t}} \neq y_{\va{t}}} \cdot \left(1+
  \frac{\xi}{n}
\right)^{T-t-1} \cdot \frac{\xi}{m}\:\:,\label{upperRT}
\end{eqnarray}
from ineq. (\ref{beps}). Assuming $T_+ \leq T$ errors are made by permutations
between classes and recalling $q \defeq \xi/m$, we see that the largest upperbound for $\|\matrice{r}\|_2 $ in
ineq. (\ref{upperRT}) is obtained when all $T_+$ errors happen at the
last elementary permutations in the sequence in $\PERM_*$, so we get that
\begin{eqnarray}
\|\matrice{r}\|_2 & \leq & \frac{1}{2X_*} \cdot \sqrt{\frac{\xi}{2}}\cdot
\sum_{t=0}^{T_+-1} q(1+q)^{T-t-1}\nonumber\\
 & & =  \frac{1}{2X_*} \cdot \sqrt{\frac{\xi}{2}} \cdot q(1+q)^{T-T_+}
\sum_{t=0}^{T_+-1} (1+q)^{T_+-t-1}\nonumber\\
 & =  & \frac{1}{2X_*} \cdot \sqrt{\frac{\xi}{2}} \cdot (1+q)^{T-T_+}((1+q)^{T_+}-1) \:\:.\label{ineq14}
\end{eqnarray}
It comes from ineq. (\ref{condLAMBDAH}) $(1+q)^{T_+}-1 \leq T_+^2 q$
and 
\begin{eqnarray}
 (1+q)^{T-T_+} & \leq & (T-T_+)^2q + 1\nonumber\\
 & \leq & \left(\frac{n}{\xi}\right)^{1-\alpha} \cdot \frac{\xi}{n} + 1
 \nonumber\\
 & & = \left(\frac{\xi}{n}\right)^\alpha + 1\nonumber\\
 & \leq & \frac{1}{4} + 1 < \sqrt{2}\:\:.
\end{eqnarray}
The last line is due to the data-model calibration assumption. We
finally get from ineq. (\ref{ineq14})
\begin{eqnarray}
\|\matrice{r}\|_2 & \leq & \frac{\sqrt{\xi}}{2X_*} \cdot
\frac{\xi}{m} \cdot T_+^2\nonumber\\
 &  & = \frac{\xi^{\frac{3}{2}}}{4 X_* m} \cdot T_+^2 \:\:.
\end{eqnarray}
We also remark that if $\PERM_*$ is $\alpha$-bounded, since $T_+ \leq
T$, we also have:
\begin{eqnarray}
\frac{\xi^{\frac{3}{2}}}{4 X_* m} \cdot T_+^2 & \leq &
\frac{\xi^{\frac{3}{2}}}{4 X_* m} \cdot
\left(\frac{m}{\xi}\right)^{1-\alpha}\nonumber\\
 & & = \frac{\sqrt{\xi}}{4 X_*} \cdot
\left(\frac {\xi}{m}\right)^{\alpha}\:\:.
\end{eqnarray}
Summarizing, we get
\begin{eqnarray}
\|\ve{\theta}^*_{T} - \ve{\theta}^*_0 \|_2 & \leq & a(T) \cdot
\|\ve{\theta}^*_0 \|_2 + b(T_+)\:\:,\label{eqthAPPROX1}
\end{eqnarray}
where 
\begin{eqnarray}
a(T) & \defeq &  \frac{\xi}{m} \cdot T^2 \leq \left(\frac{\xi}{m}\right)^\alpha \:\:,\label{condAT}\\
b(T_+) & \defeq & \frac{\xi^{\frac{3}{2}}}{4 X_* m}
\cdot T_+^2 \leq \frac{\sqrt{\xi}}{4 X_*} \cdot
\left(\frac{\xi}{m}\right)^{\alpha}\:\:,\label{condBT}
\end{eqnarray}
which yields the proof
of Theorem \ref{thAPPROX1SMB}.
\end{proof}
Theorem \ref{thAPPROX1SMB}  easily yields the proof of Theorem
\ref{thAPPROX1}.

\subsection{Proof of Theorem \ref{thIMMUNE}}\label{app:proof-thIMMUNE}

Remark that for any example $(\ve{x}, y)$, we have from
Cauchy-Schwartz inequality:
\begin{eqnarray}
|y (\ve{\theta}^*_{T} - \ve{\theta}^*_0)^\top \ve{x}| = |(\ve{\theta}^*_{T} - \ve{\theta}^*_0)^\top \ve{x}| & \leq &
\|\ve{\theta}^*_{T} - \ve{\theta}^*_0\|_2 \|\ve{x}\|_2\nonumber\\
 & \leq &  \left(\frac
  {\xi}{m}\right)^{\alpha} \cdot \left(
  \|\ve{\theta}^*_0 \|_2 + \frac{\sqrt{\xi}}{4 X_*}
  \cdot \rho
\right) \cdot X_*\nonumber\\
 & & = \left(\frac
  {\xi}{m}\right)^{\alpha} \cdot \left(
  \|\ve{\theta}^*_0 \|_2 X_* + \frac{\sqrt{\xi}}{4}
  \cdot \rho
\right) \:\:.
\end{eqnarray}
So, to have $|y (\ve{\theta}^*_{T} - \ve{\theta}^*_0)^\top \ve{x}|
< \kappa$ for some $\kappa > 0$, it is sufficient that
\begin{eqnarray}
 m & > &  \xi \cdot \left(
  \frac{\|\ve{\theta}^*_0 \|_2 X_*}{\kappa} + \frac{\sqrt{\xi}}{4\kappa}
  \cdot \rho
\right)^{\frac{1}{\alpha}}\:\:.
\end{eqnarray}
In this case, for any example $(\ve{x},
y)$ such that $y (\ve{\theta}^*_0)^\top \ve{x} > \kappa$, then 
\begin{eqnarray}
y (\ve{\theta}^*_T)^\top \ve{x} & = & y (\ve{\theta}^*_0)^\top \ve{x}
+ y (\ve{\theta}^*_T - \ve{\theta}^*_0)^\top \ve{x}\nonumber\\
 & \geq & y (\ve{\theta}^*_0)^\top \ve{x} - |y (\ve{\theta}^*_T -
 \ve{\theta}^*_0)^\top \ve{x}|\nonumber\\
 & > & \kappa - \kappa = 0\:\:,
\end{eqnarray}
and we get the statement of the Theorem.

\subsection{Proof of Theorem \ref{thDIFFLOSS}}\label{app:proof-thDIFFLOSS}

We want to bound the difference between the loss \textit{over the true data}
for the optimal (unknown) classifier $\ve{\theta}^*_0$ and the
classifier we learn from entity resolved data, $\ve{\theta}^*_T$:
\begin{eqnarray}
\Delta_S(\ve{\theta}^*_0, \ve{\theta}^*_T) & \defeq & \taylorlossparams{F(0)}{F'(0)}{c} (S, \ve{\theta}^*_T; \gamma,
\Gamma) - \taylorlossparams{F(0)}{F'(0)}{c} (S, \ve{\theta}^*_0; \gamma,
\Gamma)\:\:.\label{defDelta}
\end{eqnarray}
We have
\begin{eqnarray}
\Delta_S(\ve{\theta}^*_0, \ve{\theta}^*_T) & = & \frac{1}{m}\cdot
\left( -F'(0) (\ve{\theta}^*_T-\ve{\theta}^*_0)^\top\left(\sum_i  y_i
    \ve{x}_i\right) + c \cdot \sum_i \left(((\ve{\theta}^*_0)^\top
\ve{x}_i)^2 - ((\ve{\theta}^*_T)^\top
\ve{x}_i)^2\right)\right)\nonumber\\
 & & + \gamma (\ve{\theta}^*_T)^\top \Gamma \ve{\theta}^*_T - \gamma
 (\ve{\theta}^*_0)^\top \Gamma \ve{\theta}^*_0\nonumber\\
 & = & \frac{1}{m}\cdot
\left( A + B\right)+ C\:\:,
\end{eqnarray}
with
\begin{eqnarray}
A & \defeq & -F'(0) (\ve{\theta}^*_T-\ve{\theta}^*_0)^\top\ve{\mu}_0  \nonumber\\
B & \defeq & c \cdot \sum_i \left(((\ve{\theta}^*_0)^\top
\ve{x}_i)^2 - ((\ve{\theta}^*_T)^\top
\ve{x}_i)^2\right)\nonumber\\
C & \defeq & \gamma \cdot ((\ve{\theta}^*_T)^\top \Gamma \ve{\theta}^*_T - 
 (\ve{\theta}^*_0)^\top \Gamma \ve{\theta}^*_0)\nonumber
\end{eqnarray}
Let $\Gamma \defeq \sum_i \lambda^\downarrow_i(\Gamma) \ve{u}_i
\ve{u}_i^\top$, where the $\ve{u}_i$s are orthonormal. Cauchy-Schwartz
inequality and the fact that the $\ve{u}_i$s are unit yields:
\begin{eqnarray}
(\ve{\theta}^*_T)^\top \Gamma \ve{\theta}^*_T - 
 (\ve{\theta}^*_0)^\top \Gamma \ve{\theta}^*_0 & = & \sum_i
 \lambda^\downarrow_i(\Gamma) (\ve{\theta}^*_T - \ve{\theta}^*_0)^\top\ve{u}_i
\ve{u}_i^\top (\ve{\theta}^*_T + \ve{\theta}^*_0)\nonumber\\
 & \leq & \sum_i
 \lambda^\downarrow_i(\Gamma) \|\ve{\theta}^*_T-\ve{\theta}^*_0\|_2
   \|\ve{\theta}^*_T+\ve{\theta}^*_0\|_2 \|\ve{u}_i\|_2^2\nonumber\\
& & = \|\ve{\theta}^*_T-\ve{\theta}^*_0\|_2
   \|\ve{\theta}^*_T+\ve{\theta}^*_0\|_2 \cdot \sum_i
 \lambda^\downarrow_i(\Gamma) \nonumber\\
 & \leq & d \lambda^\downarrow_1(\Gamma) \|\ve{\theta}^*_T-\ve{\theta}^*_0\|_2
   \|\ve{\theta}^*_T+\ve{\theta}^*_0\|_2 \:\:.
\end{eqnarray}
We also have, mutatis mutandis:
\begin{eqnarray}
B & = & c\left( \sum_i \left( (\ve{\theta}^*_0-\ve{\theta}^*_T)^\top \ve{x}_i
  \right)\left( (\ve{\theta}^*_0+\ve{\theta}^*_T)^\top \ve{x}_i
  \right)\right)\nonumber\\
 & \leq & m |c| \|\ve{\theta}^*_T-\ve{\theta}^*_0\|_2
   \|\ve{\theta}^*_T+\ve{\theta}^*_0\|_2 X_*^2\:\:,
\end{eqnarray}
and finally $A \leq |F'(0)| \|\ve{\theta}^*_T-\ve{\theta}^*_0\|_2
   \|\ve{\mu}_0\|_2 $. So,
\begin{eqnarray}
\Delta_S(\ve{\theta}^*_0, \ve{\theta}^*_T) & \leq &
\|\ve{\theta}^*_T-\ve{\theta}^*_0\|_2 \cdot \left( \frac{|F'(0)| 
   \|\ve{\mu}_0\|_2}{m} + 
   \|\ve{\theta}^*_T+\ve{\theta}^*_0\|_2 (|c|X_*^2 + d \gamma\lambda^\downarrow_1(\Gamma))\right)
\end{eqnarray}
We now need to bound $\|\ve{\theta}^*_T+\ve{\theta}^*_0\|_2$ in a
convenient way:
\begin{eqnarray}
\|\ve{\theta}^*_T+\ve{\theta}^*_0\|_2 & =& \|\ve{\theta}^*_T -
\ve{\theta}^*_0 + 2 \ve{\theta}^*_0\|_2\nonumber\\
 & \leq & \|\ve{\theta}^*_T -
\ve{\theta}^*_0\|_2 + 2 \|\ve{\theta}^*_0\|_2\:\:,
\end{eqnarray}
and so Theorem \ref{thAPPROX1SMB} yields:
\begin{eqnarray}
\|\ve{\theta}^*_T+\ve{\theta}^*_0\|_2 & \leq & 2 \|\ve{\theta}^*_0\|_2  + \frac{\xi}{n} \cdot T^2
\cdot \left( \|\ve{\theta}^*_0\|_2 +
\frac{\sqrt{\xi}}{4 X_*} \cdot \rho\right)\nonumber\\
 & & = 2 \|\ve{\theta}^*_0\|_2  + \frac{\xi (\deltatheta + \deltaperm)}{mX_*} \cdot T^2
\:\:.
\end{eqnarray}
Denote for short
\begin{eqnarray}
\eta & \defeq & \|\ve{\theta}^*_0\|_2 +
\frac{\sqrt{\xi}}{4 X_*} \cdot \rho = \frac{\deltatheta + \deltaperm}{X_*}\:\:.
\end{eqnarray}
We obtain:
\begin{eqnarray}
\Delta_S(\ve{\theta}^*_0, \ve{\theta}^*_T) & \leq & \frac{\xi \eta}{m} \cdot T^2
\cdot \left(
\begin{array}{l}
\frac{|F'(0)|\|\ve{\mu}_0\|_2}{m} + 2 |c|\|\ve{\theta}^*_0\|_2X^2_* + 2d
\gamma\lambda^\downarrow_1(\Gamma)\|\ve{\theta}^*_0\|_2 + \frac{|c|\xi
  X_*(\deltatheta + \deltaperm)}{m} \cdot T^2 \nonumber\\
+ \frac{d\xi
  (\deltatheta + \deltaperm) \gamma\lambda^\downarrow_1(\Gamma)}{mX_*} \cdot
T^2 
\end{array}\right)\nonumber\\
 & & = \frac{\xi (\deltatheta + \deltaperm)}{m} \cdot T^2
\cdot \left(
\begin{array}{l}
|F'(0)|\deltaset + 2 |c|\deltatheta + \frac{2d
\gamma\lambda^\downarrow_1(\Gamma)\deltatheta}{X^2_*} + \frac{|c|\xi
  (\deltatheta + \deltaperm)}{m} \cdot T^2 \nonumber\\
+ \frac{d\xi
  (\deltatheta + \deltaperm) \gamma\lambda^\downarrow_1(\Gamma)}{mX^2_*} \cdot
T^2 
\end{array}\right)\nonumber\\
& \leq & (\deltatheta + \deltaperm) C(m)
\cdot \left(|F'(0)|\deltaset + \left(|c|+\frac{d
      \gamma\lambda^\downarrow_1(\Gamma)}{X_*^2}\right)\left( 2\deltatheta +
  C(m)(\deltatheta + \deltaperm)\right)\right)\nonumber\\
& \leq & (\deltatheta + \deltaperm) C(m)
\cdot \left(|F'(0)|\deltaset + \left(|c|+\frac{d
      \gamma\lambda^\downarrow_1(\Gamma)}{X_*^2}\right)\left( 3\deltatheta +
  2\deltaperm\right)\right)\label{eqDELTA}\:\:.
\end{eqnarray}
We have used in the last inequality the fact that under the
data-model calibration assumption, $C(m) \leq (1/4)^\alpha \leq
1$. This ends the proof of Theorem \ref{thDIFFLOSS}.

\subsection{Proof of Theorem \ref{thGENTHETA}}\label{app:proof-thGENTHETA}

Letting $\Sigma_m \defeq \{-1,1\}^m$, the empirical
Rademacher complexity of hypothesis class ${\mathcal{H}}$ is \cite{bmRA}:
\begin{eqnarray}
R_m & \defeq & \expect_{\sigma \sim \Sigma_m}
 \sup_{h \in {\mathcal{H}}}
 \left\{
   \expect_{S}[\sigma(\ve{x}) h(\ve{x})]
   \right\} \:\:.\label{defRC1}
\end{eqnarray}
It is well-known that we have for linear classifiers whose $L_2$ norm is
bounded by $\theta_*$ \cite{kstOT} (Theorem 3):
\begin{eqnarray}
R_m & \leq & \frac{X_*\theta_*}{\sqrt{m}}\:\:.\label{bsupRC1}
\end{eqnarray}
We have from Theorem \ref{thAPPROX1} and
the triangle inequality:
\begin{eqnarray}
\|\ve{\theta}^*_T\|_2 & \leq & \left( 1+ C(m) \cdot \left( 1+ \frac{\deltaperm}{\deltatheta}\right)
\right) \cdot \|\ve{\theta}^*_0\|_2 \:\:,
\end{eqnarray}
so we can consider that 
\begin{eqnarray}
R_m & \leq & \frac{\deltatheta}{\sqrt{m}} \cdot \left( 1+ C(m) \cdot \left( 1+ \frac{\deltaperm}{\deltatheta}\right)
\right) \nonumber\\ 
 & & = R^*_m+ \frac{C(m)}{\sqrt{m}} \cdot
 \left( \deltatheta + \deltaperm\right) \:\:.\label{bsupRC12}
\end{eqnarray}
We split the Rademacher complexity this way because
$R^*_m \defeq \deltatheta/\sqrt{m}$ is the bound on the Rademacher complexity with
which we could have worked for $\ve{\theta}^*_0$.
Let
\begin{eqnarray}
A & \defeq & (\deltatheta + \deltaperm) C(m)
\cdot \left(|F'(0)|\deltaset + \left(|c|+\frac{d
      \gamma\lambda^\downarrow_1(\Gamma)}{X_*^2}\right)\left( 3\deltatheta +
  2\deltaperm\right)\right)
\end{eqnarray}
be the penalty appearing in ineq. (\ref{eqDELTA}).
Letting $L$ denote the Lipschitz constant for the Taylor loss, we get
from \cite{bmRA} (Theorem 7) that with probability $\geq 1 - \delta$
over the drawing of $S \sim \mathcal{D}^n$,
\begin{eqnarray}
\Pr_{(\ve{x}, y)\sim \mathcal{D}} \left[y (\ve{\theta}^*_T)^\top \ve{x}
\leq 0 \right] & \leq & \taylorlossparams{F(0)}{F'(0)}{c} (S, \ve{\theta}^*_T; \gamma,
\Gamma) + 2L R_m + \sqrt{\frac{\ln(2/\delta)}{2m}}\nonumber\\
 & \leq & \taylorlossparams{F(0)}{F'(0)}{c} (S, \ve{\theta}^*_0; \gamma,
\Gamma)  + A + 2LR^*_m + \frac{2LC(m)}{\sqrt{m}} \cdot
 \left( \deltatheta + \deltaperm\right) \nonumber\\
 & & + \sqrt{\frac{\ln(2/\delta)}{2m}}\:\:.\label{eqGEN1}
\end{eqnarray}
So let us denote
\begin{eqnarray}
B & \defeq & \taylorlossparams{F(0)}{F'(0)}{c} (S, \ve{\theta}^*_0; \gamma,
\Gamma) + 2LR^*_m + \sqrt{\frac{\ln(2/\delta)}{2m}}\:\:,
\end{eqnarray}
which would be \cite{bmRA}'s (Theorem 7) bound
guarantee (with high probability) on the optimal
classifier $\ve{\theta}^*_0$. We instead get
\begin{eqnarray}
\Pr_{(\ve{x}, y)\sim \mathcal{D}} \left[y (\ve{\theta}^*_T)^\top \ve{x}
\leq 0 \right] & \leq & \taylorlossparams{F(0)}{F'(0)}{c} (S, \ve{\theta}^*_T; \gamma,
\Gamma) + 2L R_m + \sqrt{\frac{\ln(2/\delta)}{2m}}\nonumber\\
 & \leq & B + D\:\:.\label{eqGEN22}
\end{eqnarray}
where
\begin{eqnarray}
D & \defeq & A + \frac{2LC(m)}{\sqrt{m}} \cdot
 \left( \deltatheta + \deltaperm\right)\nonumber\\
 & & = (\deltatheta + \deltaperm) C(m)
\cdot \left(\frac{2L}{\sqrt{m}} + |F'(0)|\deltaset + \left(|c|+\frac{d
      \gamma\lambda^\downarrow_1(\Gamma)}{X_*^2}\right)\left( 3\deltatheta +
  2\deltaperm\right)\right)
\end{eqnarray}

\end{document}